\documentclass[11pt,a4paper,reqno]{amsart}

\usepackage{placeins}
\usepackage{tabularx}
\usepackage{array}
\usepackage{graphicx}%
\usepackage{multirow}%
\usepackage{amsmath,amssymb,amsfonts,mathtools,bm}%
\usepackage{amsthm}%
\usepackage{mathrsfs}%
\usepackage[title]{appendix}%
\usepackage{xcolor}%
\usepackage{textcomp}%
\usepackage{setspace}
\usepackage{manyfoot}%
\usepackage{booktabs}%
\usepackage[ruled]{algorithm2e}%
\usepackage{algorithmicx}%
\usepackage{algpseudocode}%
\usepackage{listings}%
\usepackage{booktabs}
\usepackage{pifont}
\usepackage{subcaption}%
\usepackage{makecell}
  \usepackage[top=1.5in, bottom=1.2in,
  left=1in, right=1in]{geometry}
\usepackage{hyperref}
\hypersetup{
    colorlinks=true,
    urlcolor=red,
    linkcolor=red,
    citecolor=blue
}
\usepackage[round,authoryear]{natbib}
\newcommand{\Lev}{\ensuremath{\textnormal{lev}}}
\newcommand{\SSGEN}{{\tt{SSGEN}}}

\newcommand{\R}{{\mathbb{R}}}

\newcommand{\mv}[1]{\boldsymbol{#1}}

\newcommand{\xx}{\boldsymbol{x}}
\newcommand{\XX}{\boldsymbol{\xi}}

\newcommand{\ZZ}{\boldsymbol{Z}}
\newcommand{\yy}{\boldsymbol{y}}
\newcommand{\zz}{\boldsymbol{z}}

\newcommand{\pp}{\boldsymbol{p}}

\newcommand{\Law}{\textsf{Law}}
\newcommand{\Prob}{P}

\newcommand{\Pphi}{\boldsymbol{\Phi}}
\newcommand{\pphi}{\boldsymbol{\phi}}

\newcommand{\fgen}{f_{\textsf {gen}}}

\newtheorem{assumption}{Assumption}
\newtheorem{theorem}{Theorem}
\newtheorem{definition}{Definition}
\newtheorem{example}{Example}
\newtheorem{lemma}{Lemma}
\newtheorem{proposition}{Proposition}
\newtheorem{corollary}{Corollary}
\newtheorem{remark}{Remark}
\renewcommand{\arraystretch}{1.5}

\theoremstyle{thmstyleone}%

\theoremstyle{thmstyletwo}%

\theoremstyle{thmstylethree}%

\newcommand{\authoremail}{%
mantu.gupta23@iimb.ac.in, anand.deo@iimb.ac.in%
}
\makeatletter

\renewcommand{\@settitle}{%
    \begin{center}
        \baselineskip14\p@\relax
        \large
        \bfseries
        \@title
    \end{center}%
}

\renewcommand{\@setauthors}{%
    \begingroup
        \def\thanks{\protect\thanks@warning}%
        \trivlist
            \centering
            \@topsep30\p@\relax
            \advance\@topsep by -\baselineskip
            \item\relax
            \author@andify\authors
            \def\\{\protect\linebreak}%
            {\large\authors\par}
            \vspace{3pt}
            {\normalsize
                \href{mailto:\authoremail}
                {\texttt{\footnotesize \authoremail}}\par}
        \endtrivlist
    \endgroup
}

\makeatother

\title[Simulating Stress Laws under Extremal Dependence]
{Simulating Stress Laws under Extremal Dependence:
Characterizing What Generative Models Must Preserve}

\author{Mantu Gupta, Anand Deo}

\address{%
Indian Institute of Management Bangalore,
Bannerghatta Road, Bangalore 560076
}

\begin{document}
\maketitle
\begin{abstract}
We study stress-scenario generation for systems driven by multivariate heavy-tailed risk factors. Within regions where several financial losses are simultaneously extreme, stress analysis concerns both the conditional law of the risk factors and the most plausible configurations producing those losses. We show that both are governed by the same limiting tail law. Preserving its measure recovers rare-event probabilities and scaled conditional stress laws, while misspecifying extremal dependence distorts some regular joint-stress probability. Its density governs reverse-stress optimization, whose maximizers identify the most plausible stress configurations. To exploit this common structure in finite samples, we develop SSGEN (Self-Similar Generative Estimation), which learns extremal dependence from intermediate exceedances and extrapolates to rarer levels using a Pareto radial component. Even when the target event is absent from the sample, we establish convergence rates for the generated conditional law, and data-driven reverse-stress solutions.
\end{abstract}

\section{Introduction}
When financial risk is driven by multiple underlying factors, tail-risk analysis
requires more than estimating the probability of a large loss. Stress testing, systemic risk analysis, and tail sensitive risk measurement all require information
about the \textit{joint} risk-factor configurations that arise under rare yet adverse
conditions. The relevant object is therefore often the law of the full system state
conditional on an extreme loss event. This creates a fundamental statistical tension: the scenarios that generate severe
losses are precisely those least represented in available data. Direct empirical inference in such regimes can
therefore be unreliable or infeasible.

Let $\XX$ denote the vector of underlying risk factors, and let
$(L_1,\ldots,L_K)$ denote loss functionals associated with portfolios,
institutions, or components of a financial system. For
\(\mathcal J\subseteq[K]\), stress analysis is governed not by a single tail
probability but by the conditional law
\begin{equation}\label{eq:intro_stress_law}
    \textsf{Law}\bigl(\XX \mid L_j(\XX)\geq u_j,\ j\in \mathcal J\bigr),
\end{equation}
or its appropriately scaled version as the thresholds \(u_j\) become large. This
law determines aggregate stressed losses, their allocation and concentration
across the system, and the plausible risk-factor configurations conditional on
the joint-loss event in \eqref{eq:intro_stress_law}. Such features are central
in reverse stress testing \cite{glasserman2015stress,breuer2009find},
systemic-risk measurement
\cite{tobias2016covar,acharya2017measuring,brownlees2017srisk}, and 
financial networks where common shocks generate joint losses and spillovers
\cite{eisenberg2001systemic,glasserman2015likely,kley2016risk}. Thus
\eqref{eq:intro_stress_law} is 
the primitive object from which severity, spillover,
concentration, and plausibility diagnostics are derived.

Classical extreme-value methods characterize multivariate tail laws and provide
estimators of tail probabilities, extremal-dependence summaries, and marginal
risk measures. Existing inferential procedures, however, generally target a
specified tail quantity or dependence object rather than providing a
law-level validity criterion for generators used across rare multi-loss
conditioning events.

We work in the
heavy-tailed setting,\footnote{Evidence that the tail behavior of risk factors
is governed by power laws is given, for instance, in
\cite{jansen1991frequency,gabaix2009power,kelly2014tail}.} where extreme shocks
asymptotically separate into magnitude and direction: magnitude determines
severity, while direction determines which risk factors become large together.
This separation suggests a route to generation that remains valid at extreme
levels: learn directional behavior from intermediate extremes and model
magnitudes using the power-law tail structure. It leads to two central questions
concerning the construction and validity of such a stress-law generator:

\begin{enumerate}
\item[\textbf{(Q1)}] \textbf{Validity criterion:} What must a generative
mechanism preserve in order to recover the first-order probabilities of rare
events of the form $\{L_j(\XX)\geq c_j u,\ j\in\mathcal J\},\quad c_j>0$, and the conditional laws they induce?
\item[\textbf{(Q2)}] \textbf{Data-driven learning:} Can such a tail-preserving model
be learned from data, using intermediate extremes to extrapolate to rarer
regimes that are effectively unobserved?
\end{enumerate}
The limiting tail measure identifies the first-order tail structure that a
generated model must preserve in order to reproduce rare joint-stress events and
their conditional laws, and thereby supplies the validity criterion sought in
\textbf{(Q1)}. Our generative mechanism, SSGEN (Self-Similar Generative Estimation), answers
\textbf{(Q2)} by turning this criterion into a data-driven procedure of the form
described above, in which the angular-learning step is modular: any nonparametric
learner meeting the required angular-error bound may be used, including kernel
methods, normalizing flows, GANs, or diffusion models; see
\cite{Kingma2013,goodfellow2014generative,Rezende2015,Sohl-Dickstein2015,Ho2020}. 
This paper makes the following contributions.

\noindent \textbf{i) A first-order criterion for stress-law validity:}
We answer \textbf{(Q1)} by proving that the limiting tail measure is the sharp
criterion for first-order stress-law validity. Equality of limiting tail measures
is necessary and sufficient for first-order accuracy of rare-event probabilities
and the associated scaled conditional laws. In particular, for every admissible
regular collection of loss functionals \(L_1,\ldots,L_K\), every nonempty
\(\mathcal J\subseteq[K]\), and every \(c_j>0\),
\[
    \mathbb P\bigl(L_j(\XX_u^{\rm gen})\geq c_j u,\ j\in\mathcal J\bigr)
    \sim
    \mathbb P\bigl(L_j(\XX)\geq c_j u,\ j\in\mathcal J\bigr),
\]
whenever \(\XX_u^{\rm gen}\) and the data-generating law have the same limiting
tail measure. The associated scaled conditional laws, together with the induced laws of
regular diagnostics describing severity and allocation, inherit the same
first-order limit. Conversely, failure to
preserve the limiting tail structure distorts some regular joint-stress
probability. Thus validity is not marginal-tail calibration alone: it requires
preserving extremal dependence.

\noindent \textbf{ii) SSGEN and data-driven rates:}
We answer \textbf{(Q2)} by constructing SSGEN, which learns a finite-threshold
angular law from intermediate exceedances and combines it with an explicit Pareto
radial component. Thus generation reduces to compact-domain angular learning plus
radial tail modelling.

We analyze the empirically hard regime in which the target event is effectively
unobserved. For stress levels \(u_n\) satisfying
\[
    \mathbb P\bigl(L_j(\XX)\geq c_j u_n,\ j\in\mathcal J\bigr)
    \asymp n^{-r}, \qquad r\geq 1,
\]
the expected number of target observations is a constant $c\in [0,\infty)$, while an intermediate level,
chosen to retain \(n^{1-q}\) exceedances with \(q\in(0,1)\), remains available for
learning. Under second-order tail regularity and angular-learning assumptions,
the generated scaled conditional law converges in total variation at rate
\[
    O_P\!\left(n^{-(1-q)(m\wedge\ell)} + n^{-q\gamma/s}\right),
\]
where the first term is statistical learning error from the \(n^{1-q}\)
exceedances and the second is finite-threshold bias. Increasing \(q\) trains
deeper in the tail, reducing bias but leaving fewer exceedances for learning; the
exponents make this tradeoff explicit and identify the balancing choice of \(q\).
By contraction of total variation under measurable mappings, the same rate transfers to the induced laws of diagnostics describing severity, allocation, concentration, distress breadth, and dominant contributors. 

\noindent \textbf{iii) Reverse stress testing as perturbed optimization:}
Reverse stress testing asks which risk-factor configurations most plausibly
generate a given stress event. Following
\cite{breuer2013systematic,glasserman2015stress}, we formalize plausibility as
constrained likelihood maximization over the stress region, and show that under
tail scaling this finite-level problem converges to a limiting optimization
problem determined by the limiting tail density. This casts reverse stress
testing as a perturbed optimization problem and yields a transfer principle:
\emph{any} density-preserving generator preserves the limiting argmax, and hence
the set of asymptotically most plausible stress mechanisms. We show that SSGEN
is one such generator. When the limiting reverse-stress problem has a unique
maximizer, both the true and SSGEN reverse-stress solutions converge to it. When
the limiting problem has several maximizers, the generated problem may select a
different limiting maximizer, but every such selection has the same limiting
likelihood and therefore does not represent an asymptotically less plausible
stress mechanism. A local identifiability condition further converts the
learning error of the tail model into an explicit convergence rate from the
data-driven generated solution set to the true finite-level reverse-stress
solution set.

\subsection{Literature Review and Positioning}
We position the paper relative to four strands of work: financial stress
analysis, multivariate extremes, tail-risk estimation, and
deep generative modeling for heavy tails. A preliminary conference version of this work
\citep{gupta2026extremevalueperspectivelearning} proposed the SSGEN algorithm
and studied the estimation of tail risk measures such as VaR and CVaR; the
present paper develops the distributional theory (tail-measure and
density-level preservation as validity criteria), treats reverse stress testing, and provides the data-driven
estimation guarantees of Section~\ref{sec:dd_rst}. None of these results
appear in the conference version. 

\paragraph{\textbf{i) Stress testing, reverse stress testing, and systemic risk.}}
Stress testing and reverse stress testing emphasize severe but plausible scenarios:
adverse outcomes together with risk-factor configurations that can credibly
generate them
\cite{glasserman2015stress,breuer2009find,breuer2013systematic}. In systemic-risk
settings, this concern becomes distributional: common shocks can affect several
institutions or portfolios simultaneously, making allocation, concentration, and
spillover objects of interest. Such effects arise in bipartite exposure models,
clearing systems, and financial-network models
\cite{eisenberg2001systemic,elliott2014financial,glasserman2015likely,chen2016optimization}.
These models motivate our multi-loss events, but typically do not take the
conditional risk-factor law as primitive.

\paragraph{\textbf{ii) Multivariate heavy tails and tail dependence.}}
Our mathematical foundation is multivariate regular variation and heavy-tailed
extreme-value theory. Classical work characterizes the limiting tail measure, which
describes first-order tail behavior and separates magnitude from extremal
dependence
\cite{de1984domains,resnick2008multivariate,rootzen2006multivariate,Rootzen}.
The radial part governs severity, while the angular part determines which risk
factors, and hence which losses, become large together.

A substantial statistical literature estimates angular tail structure from
intermediate exceedances
\cite{einmahl2008method,einmahl2012m}, while related work tests
multivariate heavy-tail assumptions
\cite{einmahl2025empirical}. We use the tail measure differently:
as the benchmark for validating generated tail models, not only as an estimand or
testing target. Accordingly, SSGEN imposes radial scaling explicitly and learns a
finite-threshold angular law as a data-driven surrogate for limiting extremal
dependence.

\paragraph{\textbf{iii) Tail-risk estimation and downstream functionals.}}
A large body of work estimates tail probabilities and tail-sensitive risk measures,
including VaR, CVaR, CoVaR, CoES, and related systemic-risk functionals
(see \cite{nolde2020conditional,nolde2026tail,Mainik,DeoOpt,deo2025achieving}). Such work typically fixes an event or
functional and targets a scalar or low-dimensional quantity. We instead approximate
the conditional law under the rare event, so that these quantities are evaluated as
functionals of a common rare-event distribution.

In multi-loss systems, the same conditioning event may raise questions about
aggregate severity, allocation, concentration, spillovers, and conditional tail
risk. A law-level approximation treats these questions as projections of the same
conditional law, rather than as separate estimation problems.

\paragraph{\textbf{iv) Learning and generating tail distributions.}}
Recent work incorporates extreme-value structure into generative models through
tail-aware architectures and objectives
(see \cite{allouche2022ev,allouche2026exceedgan,cont2022tail}), heavy-tailed latent
variables, priors, base noise, and diffusion mechanisms
(see \cite{huster2021pareto,liu2024learning,pandey2024heavytaileddiffusionmodels}), and tail-specific generation schemes
(see \cite{bhatia2021exgan,boulaguiem2022modeling,lhaut2025wasserstein}).
Generative models have also been used to estimate multivariate angular
distributions (see \cite{wessel2025generative}).

SSGEN differs in its inferential target rather than its architecture. Existing
methods primarily target unconditional extreme-value laws, threshold
exceedances, or selected tail-risk functionals. Tail-GAN
\citep{cont2022tail}, for example, preserves VaR and ES across classes of
portfolio strategies, while ExceedGAN \citep{allouche2026exceedgan} generates
multivariate threshold exceedances. At the angular-estimation stage,
\citet{lhaut2025wasserstein} and \citet{wessel2025generative} use generative
models to learn extremal dependence. SSGEN instead treats this stage as modular:
any angular estimator satisfying the required error bound may be used.

Our question is \textit{when} filtering a generated tail law through a rare multi-loss
event yields an asymptotically valid conditional stress law, particularly when
the conditioning event is substantially rarer than the exceedances used for
training. We show that preservation of the limiting tail measure is the sharp
first-order criterion across regular multi-loss systems, and derive
total-variation rates for the resulting conditional laws and their downstream
summaries. 
Reverse stress testing is governed instead by the limiting tail density, and the criterion for it is developed in Section~\ref{sec:RST_pop}.

SSGEN operationalizes these criteria through an architecture-agnostic EVT
reduction: it imposes Pareto radial scaling and learns the finite-threshold
angular law. Its contribution lies not in a new generative architecture or
angular estimator, but in the law-level validity characterization, the
data-driven conditional-law guarantees, and the reverse-stress transfer
principle. For ease of reference, the main results summarized below are stated under
Assumption~\ref{assume:ht_data} (heavy-tailed data). Their additional
conditions are given in Table~\ref{tab:assumption_map}.

\begin{table}[t]
    \centering
    \small
    \caption{Assumptions used in the main results.}
    \label{tab:assumption_map}
    \renewcommand{\arraystretch}{1.25}
    \begin{tabularx}{\textwidth}{
        @{}>{\raggedright\arraybackslash}X
        >{\raggedright\arraybackslash}p{0.34\textwidth}@{}
    }
        \toprule
        \textbf{Result} & \textbf{Additional conditions} \\
        \midrule

        Stress-law validity and tail-measure characterization
        (Proposition~\ref{prop:asymp_tails});
        population SSGEN validity
        (Theorem~\ref{thm:heavy_vanishing_RE});
        stress-summary convergence
        (Proposition~\ref{prop:stress_summaries})
        &
        Definition~\ref{def:hom_loss}
        \newline
        \emph{Asymptotically Homogeneous losses}
        \\[0.5em]

        Population reverse-stress characterization
        (Theorem~\ref{thm:rst_convergence})
        &
        Definition~\ref{def:density_class}
        \newline
        \emph{Density-preserving family}
        \\[0.5em]

        Data-driven SSGEN rates
        (Theorem~\ref{thm:rate_TV},
        Corollary~\ref{cor:stress_summary_dd})
        &
        Assumptions~\ref{assume:data_driven}
        and~\ref{assume:2nd_order}
        \newline
        \emph{Learning rates; second-order rate}
        \\[0.5em]

        Data-driven reverse-stress solution rate
        (Theorem~\ref{thm:rst_solution_rate})
        &
        Assumptions~\ref{assume:data_driven},
        \ref{assume:2nd_order}
        and~\ref{assume:rst_margin}, and \eqref{eqn:unif_condition}
        \newline
        \emph{Learning rates; second-order rate; local identifiability, uniform learning rate for density}
        \\
        \bottomrule
    \end{tabularx}
\end{table}

\paragraph{\textbf{Paper Outline}}
Section~\ref{sec:assump_model} introduces the modelling framework, conditional
stress law, and reverse-stress formulation.
Section~\ref{sec:characterisation} answers \textbf{(Q1)} by characterizing
first-order validity and establishing population-level correctness of SSGEN.
Section~\ref{sec:data_driven_guarantees} answers \textbf{(Q2)} through
data-driven convergence rates and statistical tradeoffs.
Sections~\ref{sec:RST_pop} and \ref{sec:dd_rst} develop the population and
data-driven reverse-stress theory,  and
Section~\ref{sec:numericals} presents the numerical experiments.
Proofs and implementation details are provided in the Appendix.

\paragraph{\textbf{Notation}} Vectors are written in boldface, e.g., \(\zz=(z_1,\ldots,z_d)\). Random variables
are denoted by upper-case letters and their realizations by lower-case letters. We
write \(\Rightarrow\) for convergence in distribution. For \(S\subset\mathcal E\),
let $S^\delta
    :=
    \{\yy\in\mathcal E:\exists\,\xx\in S \text{ such that } d(\xx,\yy)\le\delta\}$
denote its \(\delta\)-expansion. We write \(\iota_S\) for the extended-valued
indicator of \(S\), namely \(\iota_S(\zz)=0\) if \(\zz\in S\) and
\(\iota_S(\zz)=+\infty\) otherwise. For \(f:\mathsf Y\to\mathbb R\), define
$\Lev_t[f]:=\{\zz\in\mathsf Y:f(\zz)\le t\}$ and $ 
    \Lev_t^+[f]:=\{\zz\in\mathsf Y:f(\zz)\ge t\}$.  For nonnegative deterministic
sequences, \(a_n=O(b_n)\) means that \(a_n/b_n\) is eventually bounded, and
\(a_n=o(b_n)\) means that \(a_n/b_n\to0\). We use \(O_P(\cdot)\) and \(o_P(\cdot)\)
for the corresponding stochastic order notation. For a set $S$, we will denote $\Theta_S  = \{\zz\in S: \|\zz\|= 1\}$. Finally, for $r>0$, we let $B_r = \{\zz: \|\zz\| \leq r\}$.
\section{Problem Setup and Modeling Assumptions}
\label{sec:assump_model}
Let \(\XX\sim P\) be a vector of underlying risk factors with density
\(f_{\XX}\) supported on a closed cone \(\mathcal E\subset\mathbb R^d\). We represent the
system through loss functionals
\(L_1,\ldots,L_K:\mathcal E\to\mathbb R\), corresponding to institutions,
portfolios, or system components, and call \((L_1,\ldots,L_K)\) a multi-loss
system. For a non-empty \(\mathcal J\subseteq [K]\) and stress level \(u>0\), define the
joint stress region
\begin{equation}\label{eqn:rst_feasible}
    \mathcal S(u; \mathcal J)
    :=
    \{\zz\in\mathcal E:L_j(\zz)\geq u \text{ for all } j\in\mathcal J\}.
\end{equation}
The common threshold is only notational: proportional thresholds \(u_j=c_j u\),
\(c_j>0\), are covered by replacing \(L_j\) with \(L_j/c_j\), and the homogeneity
assumptions below are unchanged under such re-scalings\footnote{Non-proportional scalings, such as \(u_j=u^{c_j}\), lead to a different
asymptotic regime: depending on the relative exponents, some constraints may
dominate while others become asymptotically slack. We restrict attention to
proportional thresholds \(u_j=c_j u\), which keep the event on a common scale.}.

Our objective is to approximate the conditional law of \(\XX\) given
\(\XX\in\mathcal S(u;\mathcal J)\) as \(u\to\infty\). Since this law drifts outward
with \(u\), we study the scaled vector \(u^{-1}\XX\). In this regime, the stressed
conditional law is determined by the tail behavior of \(\XX\) and the asymptotic
geometry of the loss functionals, rather than by the body of the risk-factor
distribution. This motivates the assumption below.

\begin{assumption}[\textbf{Density Level Regular Variation}]
\label{assume:ht_data}
The density of \(\XX\) is continuous on $B_{r_0}^c \cap \mathcal E$ for some $r_0>0$, and  satisfies, uniformly over compact subsets of
\(\mathcal E\setminus\{\mv 0\}\),
\begin{equation}\label{eqn:ht_data}
    t^{s+d} f_{\XX}(t\zz) \to \varphi^\star(\zz)
    \qquad \text{as } t\to\infty,
\end{equation}
for some \(s>0\) and some function
\(\varphi^\star:\mathcal E\setminus\{\mv 0\}\to(0,\infty)\).
\end{assumption}

Assumption~\ref{assume:ht_data} is a standard condition in multivariate
heavy-tail analysis; see, for example, \citet{resnick2008multivariate}.
It holds for commonly used models with regularly varying densities, including
multivariate \(t\), Cauchy, and absolutely continuous Pareto-type models with a strictly positive limiting density. It is also preserved under finite mixtures: the components with the smallest tail index determine the limit, with their
limiting densities combined according to the mixture weights.  The limiting function \(\varphi^\star\) describes the first-order
asymptotic density of extreme risk-factor configurations. We record below
some properties that will be used throughout the paper.

\begin{lemma}\label{lem:phi_star_props}
Let $\varphi^\star$ be as in \eqref{eqn:ht_data}. Then:
\begin{enumerate}
    \item[(i)] $\varphi^\star$ is homogeneous of degree $-(s+d)$; that is,
    \[
        \varphi^\star(r\zz)
        =  r^{-(s+d)}\varphi^\star(\zz), \qquad r>0,\quad
        \zz\in\mathcal E\setminus\{\mv 0\}.
    \]
    \item[(ii)] $\varphi^\star$ is continuous on
    $\mathcal E\setminus\{\mv 0\}$. Thus,    
    $0<  \min_{\zz\in\Theta_{\mathcal E}}\varphi^\star(\zz)  \leq
        \max_{\zz\in\Theta_{\mathcal E}}\varphi^\star(\zz)
        <\infty$.
\end{enumerate}
\end{lemma}

We next impose the corresponding large-threshold structure on the loss functionals.
The condition below says that, after scaling by the stress level, the finite-level
loss mechanism converges to a homogeneous limiting loss map.

\begin{definition}[\textbf{Asymptotically homogeneous multi-loss system}]
\label{def:hom_loss}
A multi-loss system \((L_1,\ldots,L_K)\) is \emph{asymptotically homogeneous} if,
for each \(j=1,\ldots,K\), there exists a nondegenerate function
\(L_j^\star:\mathcal E\to\mathbb R\) such that
\begin{equation}\label{eqn:loss}
    t^{-1}L_j(t\zz_t)\to L_j^\star(\zz)
\end{equation}
whenever \(t\to\infty\) and \(\zz_t\to\zz\). Nondegeneracy means
\(\{\zz\in\mathcal E:L_j^\star(\zz)>0\}\neq\varnothing\).
We write
\((L_1,\ldots,L_K)\in
    \mathcal H_1(L_1^\star,\ldots,L_K^\star)\),
or simply \((L_1,\ldots,L_K)\in\mathcal H_1\) when the limits are clear.
\end{definition}
Condition~\eqref{eqn:loss} is a form of continuous convergence for the rescaled
maps \(\zz\mapsto t^{-1}L_j(t\zz)\). Consequently, each \(L_j^\star\) is
continuous and homogeneous of degree one, and by \cite{RockafellarWets1998}, Theorem 7.14, the convergence is uniform on compact
subsets of \(\mathcal E\). For \(\mathcal J\subseteq[K]\), define the limiting stress region
\begin{equation}\label{eqn:asymptotic_stress_region}
   \mathcal S^\star(\mathcal J)
    :=
    \{\zz\in\mathcal E:L_j^\star(\zz)\geq 1 \text{ for all } j\in\mathcal J\}.
\end{equation}
We call the system \emph{regular} if, for every nonempty
\(\mathcal J\subseteq[K]\) considered below,
\(\mathcal S^\star(\mathcal J)\) has nonzero Lebesgue measure. 

The limiting losses determine the first-order geometry of the joint stress region.
Regularity ensures that this limiting region remains nondegenerate after scaling even though its probability
vanishes as \(u\to\infty\). The following examples show that asymptotic homogeneity is not limited to linear
portfolio losses. It also arises in networked systems with shared exposures,
clearing mechanisms, or fire-sale amplification.

\begin{example}[\textbf{Bipartite reinsurance network}]\em
\label{eg:reinsurance}
Consider a bipartite agent--object exposure model, as in large-claims reinsurance
\cite{kley2016risk} and operational-risk networks \cite{kley2020modelling}. Let
\(\XX=(\xi_1,\ldots,\xi_d)\in\mathbb R_+^d\) denote object losses. Agent \(i\) has
exposure vector \(\mv a_i\in\mathbb R_+^d\) and capital \(c_i>0\), yielding the net
loss
\begin{equation}\label{eqn:loss_reinsurance}
    L_i(\XX):=\bigl(\mv a_i^\intercal\XX-c_i\bigr)_+,
    \qquad i=1,\ldots,q .
\end{equation}

\begin{lemma}\label{lem:verify_reinsurance}
If \(\mv a_i\neq \mv 0\) for each agent under consideration, then the reinsurance
loss system is regular and asymptotically homogeneous with
\[
    L_i^\star(\zz)= (\mv a_i^\intercal \zz)_{+},
    \qquad i=1,\ldots,q .
\]
\end{lemma}
\end{example}
The next example shows that asymptotic homogeneity is preserved under an endogenous
clearing mechanism.

\begin{example}[\textbf{Clearing with stress-contingent obligations}]\em
\label{eg:clearing}
Consider a stylized clearing network with \(m\) firms whose settlement
obligations are driven by exogenous losses \(\XX\in\mathbb R_+^d\). We use the clearing mechanism of \citet{eisenberg2001systemic}, but allow
nominal obligations to vary with the underlying stress rather than remain fixed
at face value. In particular, firm
\(j\) has exposure vector \(\mv a_j\in\mathbb R_+^d\) and nominal obligation
\[
    \bar p_j(\XX)
    :=
    \bigl(\mv a_j^\intercal\XX-r_j\bigr)_+,
    \qquad r_j\geq0.
\]
These obligations may represent derivatives settlements, variation-margin
calls, or other loss-triggered payments. The formulation describes a stress
region over which the relative directions of settlements remain fixed while
their magnitudes vary with market losses.

Let \(\lambda_{ji}\geq0\) be the fraction of firm \(j\)'s obligation owed to
firm \(i\), with \(\sum_{i=1}^m\lambda_{ji}\leq1\); the remaining fraction is
owed to an outside or society node. Let \(x_j(\XX)\geq0\) be the external
resources available to firm \(j\), and assume that \(x_j(\cdot)\) is measurable.
These may include liquid assets, collateral, or state-dependent recoveries, as
in network-valuation models with state-dependent endowments
\cite{barucca2020network}. A clearing vector satisfies
\[
    p_j(\XX)
    =
    \min\left\{
        \bar p_j(\XX),\,
        x_j(\XX)+\sum_{k=1}^m\lambda_{kj}p_k(\XX)
    \right\},
    \qquad j=1,\ldots,m.
\]
For each \(\XX\), define the clearing map
\(\mathcal C_{\XX}:[0,\bar{\mathbf p}(\XX)]
\to[0,\bar{\mathbf p}(\XX)]\) by
\[
    \bigl[\mathcal C_{\XX}(\mathbf q)\bigr]_j
    :=
    \min\left\{
        \bar p_j(\XX),\,
        x_j(\XX)+\sum_{k=1}^m\lambda_{kj}q_k
    \right\},
    \qquad j=1,\ldots,m.
\]
Since the payment made by each firm depends on the payments received from the
other firms, a clearing vector is a fixed point of this map: $\mathbf p(\XX)=\mathcal C_{\XX}\bigl(\mathbf p(\XX)\bigr)$.
The map \(\mathcal C_{\XX}\) is monotone and continuous on
\([0,\bar{\mathbf p}(\XX)]\), and therefore admits least and greatest fixed
points
\citep[Theorem~1]{eisenberg2001systemic}. Denote the collection of all fixed points of $\mathcal C_{\XX}$ as $\textsf{Fix}(\mathcal C_{\XX})$. We adopt the least clearing fixed
point \(\underline{\mathbf p}(\XX)\), characterized by
\[
    \mathcal C_{\XX}\bigl(\underline{\mathbf p}(\XX)\bigr)
    =
    \underline{\mathbf p}(\XX),
    \qquad
    \underline{\mathbf p}(\XX)\leq\mathbf p 
\]
for all $\mathbf p\in \textsf{Fix}(\mathcal C_{\XX}$), where the ordering is
component-wise. Since unpaid obligations decrease with clearing payments, this
gives the conservative maximal-loss selection:
\[
    L_j(\XX)
    :=
    \bar p_j(\XX)-\underline p_j(\XX),
    \qquad j=1,\ldots,m.
\]
\begin{lemma}\label{lem:eisenberg_clearing}
Assume that, for each \(j\), $t^{-1}x_j(t\zz_t)\to x_j^\star(\zz)$
whenever \(t\to\infty\) and \(\zz_t\to\zz\), and define $\bar p_j^\star(\zz):=\mv a_j^\intercal\zz$. Suppose that, for every \(\zz\in\mathbb R_+^d\), the limiting clearing system
\[
    p_j^\star(\zz)
    =
    \min\left\{
        \bar p_j^\star(\zz),\,
        x_j^\star(\zz)
        +\sum_{k=1}^m\lambda_{kj}p_k^\star(\zz)
    \right\},
    \qquad j=1,\ldots,m,
\]
has a unique solution \(\mathbf p^\star(\zz)\). Then, for any measurable selection
\(\widetilde{\mathbf p}(\cdot)\) of the finite-level clearing fixed points,
\[
    t^{-1}\widetilde p_j(t\zz_t)
    \longrightarrow
    p_j^\star(\zz),
    \qquad j=1,\ldots,m,
\]
whenever \(t\to\infty\) and \(\zz_t\to\zz\). In particular, taking
\(\widetilde{\mathbf p}=\underline{\mathbf p}\),
\[
    t^{-1}L_j(t\zz_t)
    \longrightarrow
    L_j^\star(\zz)
    :=
    \bar p_j^\star(\zz)-p_j^\star(\zz).
\]
\end{lemma}
\end{example}

The same framework can accommodate other systemic-risk mechanisms, including
central-counterparty clearing systems (see \cite{amini2014systemic,veraart2025systemic})
and liquidation-based amplification. The next example illustrates the latter in a
stylized proportional deleveraging model.

\begin{example}[\textbf{Fire-sale amplification network}]\em
\label{eg:firesale}
Consider \(N\) banks with exposures to \(d\) asset classes. Let
\(M\in\mathbb R_+^{N\times d}\) be the reference mark-to-market exposure
matrix, so that the direct losses generated by an exogenous asset-loss shock
\(\XX\in\mathbb R_+^d\) are
\[
    R(\XX):=M\XX.
\]
To capture fire-sale feedback through overlapping portfolios, let
\(G\in\mathbb R_+^{d\times N}\) map bank losses into the additional asset
depreciation caused by deleveraging and price impact (see for instance, 
\cite{greenwood2015vulnerable}). The resulting additional bank losses are
\(MG\mathbf y\) for an initial bank-loss vector \(\mathbf y\). Writing \(B:=MG\), the stable cumulative loss satisfies
\[
    \mathbf Y(\XX)=R(\XX)+B\mathbf Y(\XX).
\]
Hence, if \(\rho(B)<1\),
\[
    \mathbf Y(\XX)=(I-B)^{-1}R(\XX),
    \qquad
    L_i(\XX):=Y_i(\XX).
\]

\begin{lemma}\label{lem:firesale}
If \(\rho(B)<1\), and $M_i\neq \mv 0$  for every $i$, then the fire-sale loss system is asymptotically homogeneous,
with
\[
    L_i^\star(\zz)
    =
    \left[(I-B)^{-1}M\zz\right]_i,
    \qquad i=1,\ldots,N.
\]
\end{lemma}
\end{example}

These examples illustrate the role of asymptotic homogeneity: finite capital levels,
retentions, and activation thresholds affect moderate-loss behavior, but vanish at
the first order under large stress scaling. 







\subsection{Reverse stress testing}
The stressed law describes the distribution of risk-factor configurations
conditional on a joint stress event. Reverse stress testing asks a complementary
question: which configurations in the stress region are most plausible under the
original model? We formalize plausibility by constrained likelihood. For
\(\mathcal S(u;\mathcal J)\) defined in \eqref{eqn:rst_feasible}, define
\begin{equation}\label{eqn:rst_problem}
    \mathcal R(u;\mathcal J;P)
    :=
    \arg\max_{\zz\in\mathcal S(u;\mathcal J)}
    f_{\XX}(\zz).
\end{equation}
When the maximum is attained but not unique, \(\mathcal R(u;\mathcal J;P)\)
denotes the full maximizer set. Since the stressed conditional density is
proportional to \(f_{\XX}\) on \(\mathcal S(u;\mathcal J)\), this set is precisely
the set of conditional modes. Thus reverse stress testing selects, among all
configurations producing the prescribed joint stress event, those most likely under
the original risk-factor model.

This likelihood-based notion of scenario plausibility is standard in reverse
stress testing
\cite{glasserman2015stress,breuer2009find,breuer2013systematic}. For SSGEN, the
key question is whether this likelihood ordering is preserved when the true tail
density is replaced by the generated one. We show later that, after scaling by
\(u^{-1}\), the generated maximizer set converges to the true limiting
reverse-stress solution set under a local identifiability condition.

\section{Tail-Preserving Stress Generation}
\label{sec:characterisation}

We first recall the measure-convergence notion used throughout the
asymptotic analysis. After scaling, tail events are represented by sets bounded
away from the origin, while neighborhoods of the origin correspond to non-extreme
outcomes. An appropriate convergence notion to capture this phenomenon is
\(\mv M_0\)-convergence from \cite{hult2006regular}.  In what follows, let \((\mv S,d)\) be a complete,
separable metric space (in our case, \(\mathcal E \subset \mathbb R^d\), equipped with the usual norm)
and let \(\mathscr F\) be the Borel \(\sigma\)-algebra on \(\mv S\). Let
\(\mv S_0:=\mv S\setminus\{\mv 0\}\), and let \(\mv M_0\) be the class of
nonnegative measures on \(\mv S_0\) whose restriction to \(\mv S_0\setminus B_r\)
is finite for every \(r>0\).

\begin{definition}[\textbf{\(\mv M_0\)-convergence}]
\label{def:convergence_concept}
Let \(\{\mu_n\}_{n\geq1},\mu\in \mv M_0\). We say that
\(\mu_n\to\mu\) in \(\mv M_0\) if 
\[
    \lim_{n\to\infty}\mu_n(B)=\mu(B)
\]
for every \(B\in\mathscr F\) such that $\mu(\partial B)=0$ and $\mv 0\notin \operatorname{cl}(B)$\footnote{$\mv M_0$ convergence has the following alternate characterization: $\liminf \mu_n(G) \geq \mu(G)$ and $\limsup \mu_n(F)\leq \mu(F)$ for all open $G\subset \mv S$ and closed $F\subset \mv S$, such that $\mv 0\not\in\operatorname{cl}[G]\cup F$. This is equivalent to Definition~\ref{def:convergence_concept}, and will only be used in the proofs.}. We denote this by $\mu_n \xrightarrow{\mv M_0} \mu$.
\end{definition}
\subsection{Tail preservation as a validity criterion}
The two tasks introduced in Section~\ref{sec:assump_model} - approximating
stressed conditional laws and preserving reverse-stress plausibility - both depend
on the tail behavior of \(P\). The following lemma is a
direct consequence of Assumption~\ref{assume:ht_data}.

\begin{lemma}\label{lem:tail_asymp}
Suppose Assumption~\ref{assume:ht_data} holds. Then
\[
    t^s \mathbb P(t^{-1}\XX\in \cdot)\xrightarrow{\mv M_0} \nu^\star(\cdot)
\]
In polar coordinates
\(\zz=r\pphi\), $\nu^\star$ admits the representation
\[
    d\nu^\star(r,\pphi)
    =
    r^{-(s+1)}\varphi^\star(\pphi)\,dr\,d\pphi.
\]
\end{lemma}
Lemma~\ref{lem:tail_asymp} identifies \(\nu^\star\) as the first-order tail law of \(\XX\). Its radial component captures the Pareto tail scaling, while its angular component captures limiting extremal dependence. For stress generation, \(\nu^\star\) is therefore not merely descriptive: it is the benchmark that an approximation must preserve in order to match rare joint-stress probabilities and the associated stressed conditional laws. This leads to the following notion of tail fidelity.

\begin{definition}[\textbf{Tail-preserving approximation}]
\label{def:preserve_stress}
A family of random vectors \(\{\widetilde \XX_u\}_{u>0}\) is said to be
\emph{tail preserving} for \(\XX\) if
\[
    u^s \mathbb P(u^{-1}\widetilde\XX_u\in\cdot)\xrightarrow{\mv M_0} \nu^\star(\cdot)
\]
 where $\nu^\star$ is the
limiting tail measure associated with \(\XX\) in Lemma~\ref{lem:tail_asymp}.
We denote the collection of tail-preserving approximations of \(\XX\) by
\(\mathcal T(P)\).
\end{definition}

The relevance of \(\mv M_0\)-convergence for stress generation can be seen from
the simplest homogeneous case. Suppose, for example, that the losses
\(\{L_j:j\in\mathcal J\}\) are positively homogeneous of degree one. Then, for a
common stress level \(u\),
\[
    \mathcal S(u;\mathcal J)=u\,\mathcal S(1;\mathcal J),
\]
and hence the scaled stressed law is
\[
    \Law\!\left(
        u^{-1}\XX
        \,\middle|\,
        \XX\in\mathcal S(u;\mathcal J)
    \right)
    =
    \Law\!\left(
        u^{-1}\XX
        \,\middle|\,
        u^{-1}\XX\in\mathcal S(1;\mathcal J)
    \right).
\]
Thus, after scaling, the rare stress event is represented by a set bounded away
from the origin. The convergence in Lemma~\ref{lem:tail_asymp} therefore implies
that both the probability of the stress event and the limiting conditional law are
determined by the restriction of \(\nu^\star\) to this stress region.

This simple case captures the role of \(\nu^\star\): it is the first-order object
governing rare joint-stress probabilities and the associated scaled stressed laws.
The results below extend this intuition from exactly homogeneous losses to the
regular asymptotically homogeneous systems considered in the paper.

\begin{proposition}[\textbf{A Sharp Characterization of Tail Preservation}]\label{prop:asymp_tails}
Suppose Assumption~\ref{assume:ht_data} holds, and let \(\XX\sim P\).

\begin{enumerate}
    \item[(i)] If \(\{\XX_u\}_{u>0}\in\mathcal T(P)\), then for every regular
    asymptotically homogeneous multi-loss system
    \((L_1,\ldots,L_K)\in\mathcal H_1\) and every nonempty subset
    \(\mathcal J\subseteq[K]\) considered in the analysis,
    \[
        \mathbb P\!\left(\XX_u\in \mathcal S(u;\mathcal J)\right)
        \sim
        \mathbb P\!\left(\XX\in \mathcal S(u;\mathcal J)\right),
        \qquad u\to\infty .
    \]

    \item[(ii)] Conversely, if
    \(\{\XX_u\}_{u>0}\not\in \mathcal T(P)\), then there exist
    \(K\ge 1\), regular  loss functionals
    \((L_1,\ldots,L_K)\in\mathcal H_1\), and a nonempty subset
    \(\mathcal J\subseteq[K]\) such that
    \[
        \mathbb P\!\left(\XX_u\in \mathcal S(u;\mathcal J)\right)
        \not\sim
        \mathbb P\!\left(\XX\in \mathcal S(u;\mathcal J)\right),
        \qquad u\to\infty .
    \]
\end{enumerate}
\end{proposition}

Proposition~\ref{prop:asymp_tails} identifies tail preservation as the sharp
first-order criterion for recovering probabilities of regular asymptotically
homogeneous joint-stress events. The same criterion also yields convergence of the
corresponding stressed conditional laws.

\begin{corollary}\label{cor:stress_law_preserve}
Suppose \(\{\XX_u\}_{u>0}\in\mathcal T(P)\), and let
\((L_1,\ldots,L_K)\in\mathcal H_1\) be regular. Then, for every nonempty
\(\mathcal J\subseteq[K]\) considered in the analysis,
\[
\Law\!\left(
u^{-1}\XX_u
\,\middle|\,
\XX_u\in \mathcal S(u;\mathcal J)
\right)
\Rightarrow
\frac{\nu^\star\!\left(\,\cdot \cap \mathcal S^\star(\mathcal J)\right)}
     {\nu^\star\!\left(\mathcal S^\star(\mathcal J)\right)}
\qquad \text{as } u\to\infty,
\]
where \(\mathcal S^\star(\mathcal J)\) is defined in
\eqref{eqn:asymptotic_stress_region}.
\end{corollary}
Together, Proposition~\ref{prop:asymp_tails} and
Corollary~\ref{cor:stress_law_preserve} establish
Definition~\ref{def:preserve_stress} as the natural first-order fidelity
criterion for heavy-tailed stress generation. Preserving the limiting tail
measure ensures first-order accuracy for rare joint-stress probabilities and
convergence of the corresponding scaled stressed laws; failure to preserve it
distorts some regular asymptotically homogeneous joint-stress event.

This characterization also exposes the statistical bottleneck. The limiting
tail measure \(\nu^\star\) depends on the asymptotic angular intensity
\(\varphi^\star\), which is not directly observable from finite samples. Write
\(R=\|\XX\|\) and \(\Pphi=\XX/\|\XX\|\), and let \(\varphi^{(t)}\) denote the
conditional density of \(\Pphi\) given \(\{R>t\}\). SSGEN replaces
\(\varphi^\star\) by its finite-threshold counterpart \(\varphi^{(t)}\), while
modeling the radial tail separately. An intermediate threshold makes
\(\varphi^{(t)}\) estimable from observed exceedances while permitting
extrapolation beyond the observed radial range. The next subsection formalizes
this population construction.

\subsection{Learning Stress Laws via Tail Extrapolation}\label{sec:SSGEN} 

We now turn the tail-preservation criterion into a population construction.
Above a high radial threshold, the conditional
tail density admits an asymptotically self-similar factorization: the radial
component is Pareto to first order, while the angular component is described
by the finite-threshold density \(\varphi^{(t)}\).

\begin{proposition}\label{prop:conditional_tail_factorization}
Suppose Assumption~\ref{assume:ht_data} holds. Then
\[
    \sup_{r>t,\;\pphi}
    \left|
        \frac{
            f_{R,\Pphi}(r,\pphi\mid R>t)
        }{
            s t^s r^{-(s+1)}\varphi^{(t)}(\pphi)
        }
        -1
    \right|
    \to 0
    \qquad\text{as }t\to\infty.
\]
\end{proposition}
Let \(p_t=\mathbb P(R>t)\). 
Motivated by Proposition~\ref{prop:conditional_tail_factorization}, define
the population SSGEN distribution at threshold \(t\) through the
radial--angular density
\begin{equation}\label{eqn:f_tgen}
    \fgen^{(t)}(r,\pphi)
    =
    f_{R,\Pphi}(r,\pphi)\mathbf 1(r<t)
    +
    p_t s t^s r^{-(s+1)}
    \varphi^{(t)}(\pphi)\mathbf 1(r>t).
\end{equation}
Thus, the body is left unchanged, while the tail is replaced by a Pareto
radial law coupled with the finite-threshold angular law. The factor \(p_t\)
preserves the probability mass beyond \(t\).

The following result shows that this splice is tail-preserving. Its relative
density error vanishes uniformly beyond the learning threshold, and hence it
preserves the limiting tail measure whenever that threshold is intermediate
relative to the target stress level.

\begin{theorem}\label{thm:heavy_vanishing_RE}
Suppose Assumption~\ref{assume:ht_data} holds, and let \(\fgen^{(t)}\) be
defined by \eqref{eqn:f_tgen}. Then
\begin{equation}\label{eqn:heavy_vanishing_RE}
    \sup_{r>t,\;\pphi}
    \left|
        \frac{f_{R,\Pphi}(r,\pphi)}
             {\fgen^{(t)}(r,\pphi)}
        -1
    \right|
    \to 0
    \qquad\text{as }t\to\infty.
\end{equation}
Moreover, let \(t(u)\to\infty\) satisfy \(t(u)=o(u)\), and suppose that
\(\XX_u^{\rm gen}\) has radial--angular density
\(\fgen^{(t(u))}\). Then $\{\XX_u^{\rm gen}\}_{u>0}\in\mathcal T(P)$.
In particular, this holds for \(t(u)=u^\theta\) with
\(\theta\in(0,1)\).
\end{theorem}

Theorem~\ref{thm:heavy_vanishing_RE} shows that the population splice is
asymptotically exact throughout the tail region \(\{r>t\}\), and hence does more
than preserve the limiting tail measure: it matches the tail density with
vanishing relative error. This permits the threshold at which the model is
constructed to differ from the rarer scale at which it is used.

Specifically, \(u\) is the target scale defining
\(\mathcal S(u;\mathcal J)\), while \(t(u)\) is the radial threshold used to
construct the splice. The conditions \(t(u)\to\infty\) and \(t(u)/u\to0\)
ensure that the construction is tail-based but uses a less rare exceedance
region. For example, \(t(u)=u^\theta\), \(\theta\in(0,1)\), uses directions
observed above \(t(u)\) together with Pareto radial scaling to extrapolate to
the stress scale \(u\).

At the sample level, the empirical distribution represents the body, while
SSGEN constructs the tail beyond an intermediate threshold \(t\). Motivated by
Theorem~\ref{thm:heavy_vanishing_RE}, it learns the angular law from observed
exceedances, estimates the radial tail index, and generates tail observations
by pairing a sampled direction with an independent Pareto radius.

\begin{algorithm}[h]
\small
\caption{\(\SSGEN\): Tail generation for regularly varying risk factors}
\label{algo:generate_extremes}
\KwIn{Data samples \(\XX_1,\ldots,\XX_n\), intermediate threshold \(t>0\)}
\KwOut{A generated tail observation \(H_t\)}

\textbf{I) Training step}\;
\begin{enumerate}
    \item[(i)] Compute \(R_i=\|\XX_i\|\) and form the exceedance index set
    \[
        \mathcal Y_t=\{i:R_i>t\}.
    \]

    \item[(ii)] Form the corresponding directions
    \[
        \Theta_{\mathcal Y_t}
        =
        \left\{
            \Pphi_i=\frac{\XX_i}{R_i}:i\in\mathcal Y_t
        \right\},
    \]
    and train a nonparametric angular model
    \(\widehat Q_t\) on \(\Theta_{\mathcal Y_t}\).

    \item[(iii)] Estimate \(s\) from the exceedance radii
    \(\{R_i:i\in\mathcal Y_t\}\), obtaining \(\widehat s_t\).
\end{enumerate}

\noindent\textbf{II) Generation step}\;
\begin{enumerate}
    \item[(i)] Sample \(\Pphi^\star\sim\widehat Q_t\).

    \item[(ii)] Independently sample \(R_t\) from the density
    \[
        r\mapsto
        \widehat s_t\,t^{\widehat s_t}
        r^{-(\widehat s_t+1)}\mathbf 1(r>t).
    \]

    \item[(iii)] Return \(H_t=R_t\Pphi^\star\).
\end{enumerate}
\end{algorithm}

Algorithm~\ref{algo:generate_extremes} is the sample-level counterpart of the
tail component in \eqref{eqn:f_tgen}. If a full spliced distribution is
required, the empirical body and the generated tail are combined using the
estimated tail mass
\[
    \widehat p_t=\frac{|\mathcal Y_t|}{n}.
\]
Section~\ref{sec:data_driven_guarantees} quantifies the additional error from
estimating the angular law and radial tail index.

\subsection{Asymptotic Distribution of Stress Summaries} \label{sec:stress_asympt}
The stressed law is the primitive object in our framework, but its practical
value comes from the downstream quantities it determines. Since the
distributional asymptotics are stated for the scaled factors \(u^{-1}\XX\), we
evaluate such diagnostics through the scaled loss vector
\(\mv L_u(\zz):=\bigl(u^{-1}L_1(u\zz),\ldots,u^{-1}L_K(u\zz)\bigr)\), and
represent a diagnostic \(W\) on the scaled state space as
\(W_u(\zz):=W(\mv L_u(\zz))\). Write
\(\mv L^\star(\zz):=\bigl(L_1^\star(\zz),\ldots,L_K^\star(\zz)\bigr)\). Let
\(\ZZ_u\sim\Law(u^{-1}\XX\mid\XX\in\mathcal S(u;\mathcal J))\),
\(\widetilde\ZZ_u\sim\Law(u^{-1}\widetilde\XX_u\mid\widetilde\XX_u\in
\mathcal S(u;\mathcal J))\), and
\[
    \Law(\ZZ^\star)
    :=
    \frac{\nu^\star(\cdot\cap\mathcal S^\star(\mathcal J))}
         {\nu^\star(\mathcal S^\star(\mathcal J))}.
\]
\begin{proposition}[\textbf{Convergence of stress summaries}]
\label{prop:stress_summaries}
Let \((L_1,\ldots,L_K)\in\mathcal H_1(L_1^\star,\ldots,L_K^\star)\) be regular, and suppose $\{\widetilde \XX_u\}_{u> 0} \in \mathcal T(P)$. 
Let \(W:\mathbb R^K\to\mathsf Y\) be a measurable diagnostic taking values in a
metric space \(\mathsf Y\), and set
\[
    W_u(\zz):=W(\mv L_u(\zz)),
    \qquad
    W^\star(\zz):=W(\mv L^\star(\zz)).
\]
Suppose that \(W\) is continuous at \(\mv L^\star(\zz)\) for all
\(\zz\in D_0\), where \(D_0\subset\mathcal E\) is measurable and
\(\mathbb P(\ZZ^\star\in D_0)=1\). Then
\begin{equation}\label{eqn:stress_summaries}
    W_u(\ZZ_u)\Rightarrow W^\star(\ZZ^\star),
    \qquad
    W_u(\widetilde \ZZ_u)\Rightarrow W^\star(\ZZ^\star).
\end{equation}
In particular, the conclusions  of \eqref{eqn:stress_summaries} hold when $\widetilde\XX_u = \XX_u^{\rm gen}$. 
\end{proposition}
Any diagnostic that is continuous at the limiting scaled-loss configurations
thus inherits the same first-order limit under the true and the generated law;
indicator and threshold-based diagnostics are covered whenever the limiting law
assigns no mass to the relevant discontinuity boundary. The examples below
apply this to aggregate stressed loss, loss shares, concentration, and breadth
of distress.
\begin{example}\label{eg:stress_summaries}\em
For \(\boldsymbol\ell=(\ell_1,\ldots,\ell_K)\), aggregate stressed loss, the
loss-share vector, and the HHI concentration index are
\[
    W^{\rm agg}(\boldsymbol\ell):=\sum_{j=1}^K \ell_j,
    \qquad
    W^{\rm share}(\boldsymbol\ell)
    :=\left(\frac{\ell_i}{\sum_{j}\ell_j}\right)_{i=1}^K,
    \qquad
    W^{\rm HHI}(\boldsymbol\ell):=\sum_{i=1}^K
    \bigl(W_i^{\rm share}(\boldsymbol\ell)\bigr)^2,
\]
the share vector being defined arbitrarily when the denominator vanishes
\cite{acharya2017measuring,brownlees2017srisk,herfindahl1950concentration}.
For \(\eta\in(0,1)\), the breadth statistic
\[
    W_\eta^{\rm breadth}(\boldsymbol\ell)
    :=\sum_{i=1}^K\mathbf 1\!\left\{W_i^{\rm share}(\boldsymbol\ell)>\eta\right\}
\]
counts components with economically significant loss shares
\cite{greenwood2015vulnerable}. As the loss shares are scale-free, breadth
depends on \(\ZZ^\star\) only through its direction, and hence on the angular
part of \(\nu^\star\) rather than radial severity: it is a diagnostic of
extremal dependence. Collect
\(W:=(W^{\rm agg},W^{\rm share},W^{\rm HHI},W_\eta^{\rm breadth})\). If the
limiting total loss is positive \(\ZZ^\star\)-almost surely and no limiting
share has an atom at \(\eta\), the discontinuity set of \(W\) is
\(\ZZ^\star\)-null, and Proposition~\ref{prop:stress_summaries} applies.
\end{example}
The diagnostics above concern the distribution of stress outcomes, rather than
the plausibility of individual generated scenarios. For every continuity set
\(B\), Proposition~\ref{prop:stress_summaries} gives
$\mathbb P\bigl(W_u(\ZZ_u)\in B\bigr)
    - 
    \mathbb P\bigl(W_u(\widetilde\ZZ_u)\in B\bigr)
    \to 0.$
Thus, SSGEN recovers the asymptotic probabilities of large aggregate loss, high
concentration, broad distress, and related stress outcomes. This requires the joint tail law. In particular, the probability of broad
distress under system-wide stress depends on the conditional angular measure
induced by restricting \(\nu^\star\) to
\(\mathcal S^\star(\mathcal J)\). Marginal tails and pairwise co-risk summaries  do not generally determine this measure.
Appendix~\ref{app:matched} illustrates the distinction using two reinsurance
networks with identical marginal tails and pairwise co-stress probabilities,
but markedly different probabilities of system-wide distress. By estimating
the angular law from intermediate exceedances, SSGEN recovers this probability
at the rate established in Corollary~\ref{cor:stress_summary_dd}, even when the
joint-stress event is too rare to estimate by direct conditioning.

\section{Asymptotics for Reverse Stress Testing}\label{sec:RST_pop}
Unlike the summaries above, reverse stress testing is determined by likelihood
maximizers and is therefore not controlled by the limiting tail measure alone.
Tail preservation is a measure-level property: it neither requires the
approximating laws to be absolutely continuous nor controls their densities
when they exist. Thus, for a family in \(\mathcal T(P)\), the reverse-stress
objective need not even be defined. Even when it is defined, convergence of
the stressed laws leaves the densities unconstrained pointwise, and their
maximizers need not converge.

We first study this requirement at the population level. At large stress levels,
the feasible region lies in the far tail, and it is therefore natural to formulate
the problem on a scaled state space. Throughout this subsection, \(f_Q\)
denotes the density of the risk-factor vector under \(Q\). Let
\begin{equation}\label{eqn:scaled_versions}
    \Gamma_u(\mathcal J)
    :=
    u^{-1}\mathcal S(u;\mathcal J),
    \qquad
    \Psi_u^{(Q)}(\zz)
    :=
    u^{s+d}f_Q(u\zz),
\end{equation}
and
\[
    \mathcal Y_u(\mathcal J;Q)
    :=
    \arg\max_{\zz\in\Gamma_u(\mathcal J)}
    \Psi_u^{(Q)}(\zz).
\]
The map \(\zz\mapsto u\zz\) is a bijection from
\(\Gamma_u(\mathcal J)\) to \(\mathcal S(u;\mathcal J)\), and multiplication of
the objective by \(u^{s+d}>0\) does not change its maximizers. Hence
\[
    \mathcal Y_u(\mathcal J;Q)
    =
    u^{-1}\mathcal R(u;\mathcal J;Q).
\]
Thus, \(\mathcal Y_u(\mathcal J;Q)\) is the reverse-stress solution set viewed on
the stress scale. We assume throughout this section that the finite-level
reverse-stress optimization problems under consideration attain their maxima
for all sufficiently large \(u\).

Replacing the feasible region \(\Gamma_u(\mathcal J)\) and the scaled objective
\(\Psi_u^{(P)}\) by their respective limits suggests the candidate limiting
problem
\begin{equation}\label{eqn:Y_star}
    \mathcal Y^\star(\mathcal J)
    :=
    \arg\max_{\zz\in\mathcal S^\star(\mathcal J)}
    \varphi^\star(\zz).
\end{equation}
This formulation identifies the separate density-level requirement imposed by
reverse stress testing. Accordingly, rather than imposing tail preservation, we
require an approximation \(\XX_u\sim P_u\) to reproduce the scaled density
objective on the region relevant to the optimization. Together with stability
of the limiting problem, this yields convergence of the approximate
reverse-stress solutions to \(\mathcal Y^\star(\mathcal J)\). The class below
formalizes this density-level criterion; membership in \(\mathcal T(P)\) is not
required.
\begin{definition}[\textbf{Density-preserving families}]\label{def:density_class}
    We say that approximations $\XX_u\sim P_u$ on $\mathcal E$ with scaled densities $\Psi_u^{(P_u)}$ defined in \eqref{eqn:scaled_versions} are density preserving with respect to $P$ if
    \begin{enumerate}
        \item[\textup{(D1)}] $\Psi_u^{(P_u)}(\zz_u) \to  \varphi^\star(\zz)$ whenever $\zz_u\to \zz\in \mathcal E\setminus \{\mv 0\}$ and 
        \item[\textup{(D2)}] $\displaystyle\lim_{M\to\infty}\limsup_{u\to\infty}\,
    \sup_{\|\zz\|\ge M} \Psi_u^{(P_u)}(\zz)=0$ 
    \end{enumerate}
    We denote the collection of all density preserving approximations of $P$ by $\mathcal D(P)$.
\end{definition}

Although (D1) yields uniform convergence of the scaled objectives on compact
subsets of \(\mathcal E\setminus\{\mathbf 0\}\), the feasible regions
\(\Gamma_u(\mathcal J)\) are not compact. We therefore first localize the
finite-level solution sets to a common compact region, on which standard
perturbation arguments apply; see \cite[Chapter~4]{bonnans2013perturbation}.
This requires ruling out convergence of the scaled maximizers to the origin
and their escape to infinity. The former is excluded by the geometry of the
feasible regions, whereas the latter requires (D2).

Indeed, (D1) alone permits the approximate density to develop peaks at radii
growing faster than the stress level. Such peaks may carry vanishing probability
mass but still attract the maximizers beyond every compact set on which (D1)
controls the objective. Condition (D2) excludes these remote peaks. It holds for
the data-generating law by Assumption~\ref{assume:ht_data}, and for the population
SSGEN splice because its radial density is exactly Pareto above the splice
threshold and its angular densities are uniformly bounded. The next lemma
establishes localization for any family in \(\mathcal D(P)\).
\begin{lemma}[\textbf{Localization of maximizers}]
\label{lem:containment} Let
\((L_1,\ldots,L_K)\in
\mathcal H_1(L_1^\star,\ldots,L_K^\star)\) be regular. 
Let $\{\XX_u\}\in\mathcal D(P)$. There exist a compact set
$\mathcal K_{\mathcal J}\subset\mathcal E\setminus\{\mathbf 0\}$ and
$u_0<\infty$ such that, for all $u\ge u_0$,
\[
    \mathcal Y_u(\mathcal J;P)
    \cup \mathcal Y_u(\mathcal J;P_u)\subseteq\mathcal K_{\mathcal J}.
\]
\end{lemma}
With the maximizers confined to \(\mathcal K_{\mathcal J}\), the uniform
convergence furnished by (D1) now applies on the whole region relevant to the
optimization. Moreover, asymptotic homogeneity of the losses implies that the
scaled feasible regions \(\Gamma_u(\mathcal J)\) converge to the limiting
stress region \(\mathcal S^\star(\mathcal J)\). The localized constrained
likelihood problems associated with any family in \(\mathcal D(P)\) therefore
admit a well-defined first-order limit. The next theorem shows that this limiting problem governs the
reverse-stress solutions of every density-preserving family. 
\begin{theorem}\label{thm:rst_convergence}
Suppose Assumption~\ref{assume:ht_data} holds, and let
\((L_1,\ldots,L_K)\in
\mathcal H_1(L_1^\star,\ldots,L_K^\star)\) be regular. Then, for any $\{\XX_u\}\in \mathcal D(P)$ with $\XX_u\sim P_u$,
\begin{subequations}\label{eqn:rst_main}
\begin{align}
    \limsup_{u\to\infty}\mathcal Y_u(\mathcal J;P)
    &\subseteq \mathcal Y^\star(\mathcal J),
    \qquad
    \limsup_{u\to\infty}\mathcal Y_u(\mathcal J;P_{u})
    \subseteq \mathcal Y^\star(\mathcal J),
    \label{eqn:rst_main_gen}\\
    \frac{
    \sup_{\zz\in \Gamma_u(\mathcal J)}
    f_{P_{u}}(u\zz)
    }{
    \sup_{\zz\in \Gamma_u(\mathcal J)}
    f_P(u\zz)
    }
    &\to 1 \quad \text{as } u\to\infty.
    \label{eqn:rst_main_value}
\end{align}
\end{subequations}
If, in addition, \(\mathcal Y^\star(\mathcal J)=\{\zz^\star\}\), then any selections
\[
\zz_u\in \mathcal Y_u(\mathcal J;P),
\qquad
\tilde \zz_{u}\in
\mathcal Y_u(\mathcal J;P_{u})
\]
satisfy
\[
    \zz_u\to\zz^\star,
    \qquad
    \tilde \zz_{u}\to\zz^\star .
\]
\end{theorem}
Lemma~\ref{lem:ssgen_in_D} below demonstrates that the SSGEN approximation is
density-preserving, and therefore inherits the reverse-stress guarantees of
Theorem~\ref{thm:rst_convergence}.
\begin{lemma}[\textbf{SSGEN is density-preserving}]\label{lem:ssgen_in_D}
Suppose Assumption~\ref{assume:ht_data} holds, and let \(t(u)=u^\theta\) for
some \(\theta\in(0,1)\). Then \(\XX_u^{\rm gen}\in\mathcal D(P)\).
\end{lemma}
By Lemma~\ref{lem:ssgen_in_D}, Theorem~\ref{thm:rst_convergence} applies to
SSGEN. At the population level, the splice therefore preserves the first-order
likelihood geometry of reverse stress testing, and any accumulation point of
generated reverse-stress solutions lies in the limiting set of most plausible
stress mechanisms under the true model. The value convergence in
\eqref{eqn:rst_main_value} provides the corresponding likelihood-level
guarantee. When several limiting mechanisms coexist, the generated problem may
select a different maximizer, but such a maximizer attains the same
first-order likelihood and is therefore not an implausible one. If the
limiting mechanism is unique, both the true and generated solutions converge
to it.

\subsection{Numerical Illustration} 
\label{sec:population_numerics}
We illustrate the population-level results above in a simple reinsurance network.
Consider the setting of Example~\ref{eg:reinsurance} with \(K=5\) agents and
\(d=2\) object-loss risk factors. Let
\[
A^\intercal=
\begin{pmatrix}
1.0 & 0.7 & 0.5 & 0.3 & 0\\
0 & 0.3 & 0.5 & 0.7 & 1.0
\end{pmatrix},
\qquad
\mathbf c=(0.5,\,0.3,\,0.4,\,0.3,\,0.5).
\]
We take \(\XX\) to have radial-angular density
\[
f_{R,\Pphi}(r,\pphi)
=
\alpha r^{-\alpha-1}
\left[
(1-\varepsilon r^{-\gamma})\,g^\star(\pphi)
+
\varepsilon r^{-\gamma}p(\pphi)
\right],
\qquad r>1,
\]
where \(g^\star\) and \(p\) are Beta densities on \([0,\pi/2]\), and
\[
    \varepsilon=0.5,\qquad \alpha=2.5,\qquad \gamma=3.
\]
In this model, Assumption~\ref{assume:ht_data} holds with \(s=\alpha\), and the
conditional angular law converges to \(g^\star\). The term involving \(p\) creates
a finite-threshold angular perturbation, so the example is designed to show how the
population SSGEN splice improves as the target stress level moves deeper into the
tail.  We focus on the event in which all agents are simultaneously stressed:
\[
\mathcal J=[5], \qquad
\mathcal S(u; {\mathcal J})
=
\{\zz: L_j(\zz)\ge u \text{ for all } j\in[5]\}.
\]
The population SSGEN splice is formed at the intermediate threshold \(t=3\), for
which \(p_t=P(R>t)\approx 0.06\). We compare the true conditional law and the
SSGEN conditional law at three stress levels \(u=5,15,25\), chosen so that
\(P(\XX\in\mathcal S(u;\mathcal J))\) is approximately of order \(10^{-2}\),
\(10^{-3}\), and \(10^{-4}\), respectively. At these finite levels, the
separation between the intermediate and target scales may be parameterized by
an effective exponent \(\theta\in(0,1)\) through \(t=u^{\theta}\). The three
stress levels correspond to \(\theta\approx 0.68,0.41,\) and \(0.34\),
respectively. Thus, the experiment illustrates increasingly deep extrapolation
from a fixed intermediate threshold, paralleling the separation
\(t(u)=u^\theta=o(u)\) used in Theorem~\ref{thm:heavy_vanishing_RE}.

Figure~\ref{fig:stress_comparison} compares scaled samples from the true stressed
law and from the population SSGEN approximation. As \(u\) increases, the scaled
stress region approaches its limiting shape, and the two conditional sample clouds
align more closely. This illustrates the conditional-law convergence in
Corollary~\ref{cor:stress_law_preserve}: the intermediate-threshold splice captures
the relevant joint-tail geometry even when the target stress event is substantially
rarer than the exceedance level used to construct the splice.

\begin{figure}[htbp]
    \centering
    \begin{subfigure}[b]{0.32\textwidth}
        \centering
        \includegraphics[width=\textwidth]{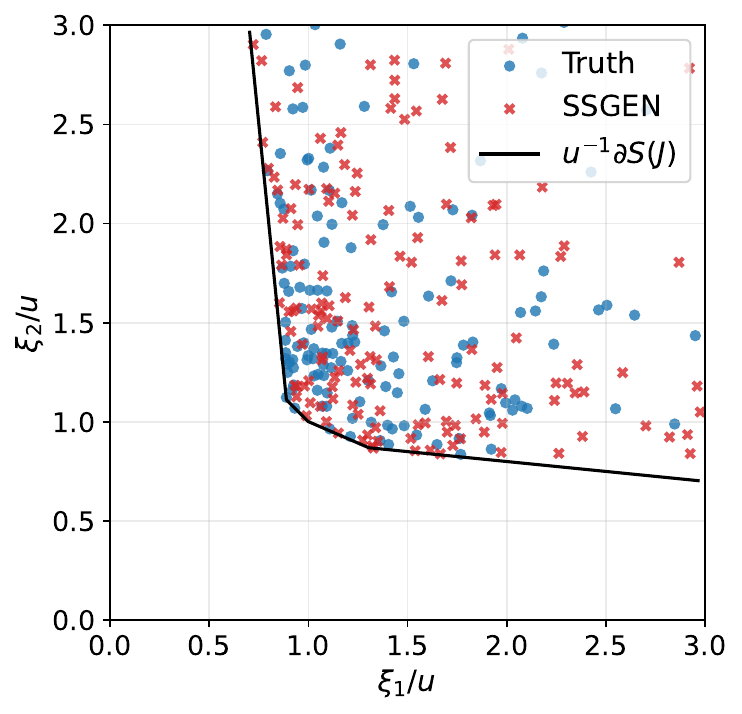}
        \caption{\(u=5\)}
        \label{fig:stress_u5}
    \end{subfigure}
    \hfill
    \begin{subfigure}[b]{0.32\textwidth}
        \centering
        \includegraphics[width=\textwidth]{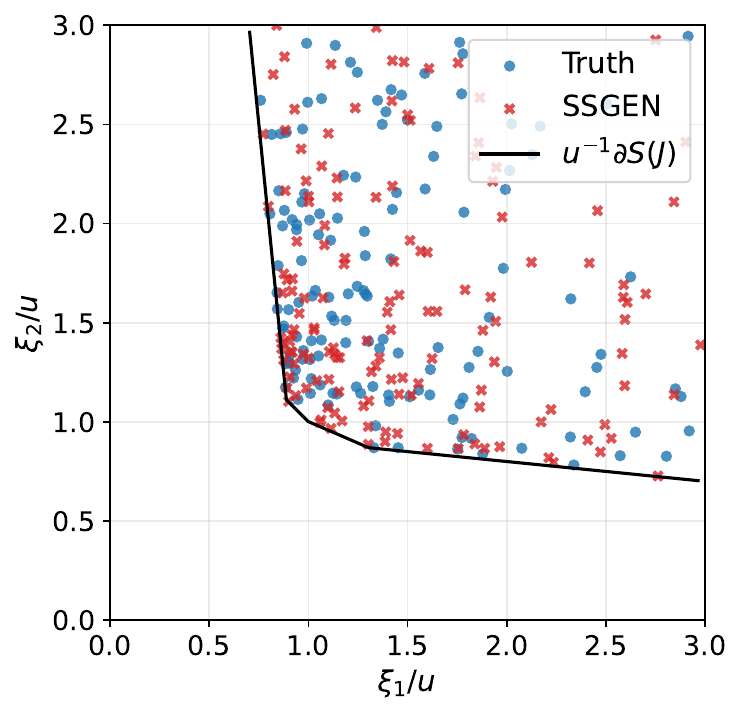}
        \caption{\(u=15\)}
        \label{fig:stress_u15}
    \end{subfigure}
    \hfill
    \begin{subfigure}[b]{0.32\textwidth}
        \centering
        \includegraphics[width=\textwidth]{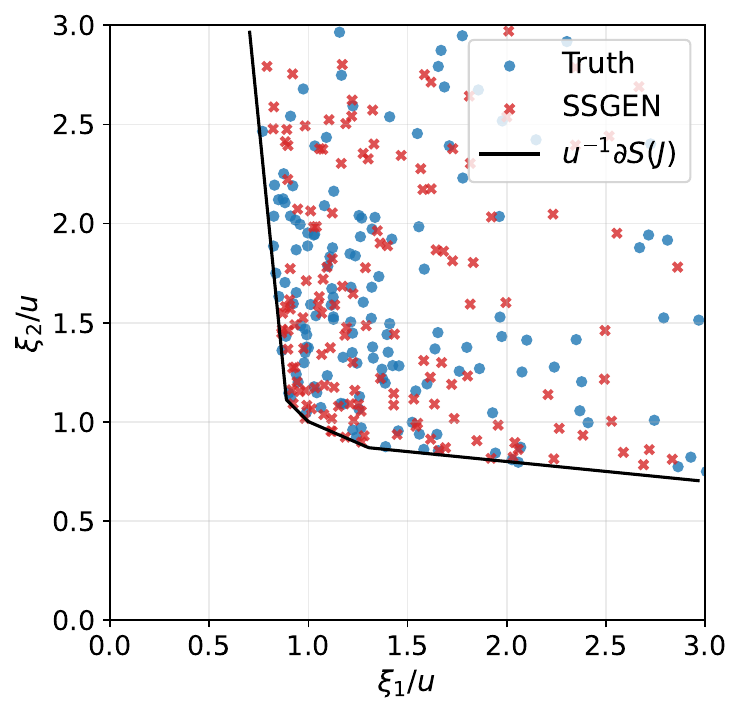}
        \caption{\(u=25\)}
        \label{fig:stress_u25}
    \end{subfigure}
    \caption{Scaled stressed samples under the true model and the population SSGEN
    approximation, together with the scaled stress boundary
    \(u^{-1}\partial\mathcal S(u; {\mathcal J})\). The SSGEN splice is formed
    at the intermediate threshold \(t=3\), where \(p_t\approx0.06\), while the target
    stress probabilities are approximately \(10^{-2}\), \(10^{-3}\), and \(10^{-4}\)
    across the three panels.}
    \label{fig:stress_comparison}
\end{figure}

Next, we examine downstream summaries of the stressed system state. We use the
aggregate stressed loss \(W^{\rm agg}\), the breadth statistic
\(W_{\eta}^{\rm breadth}\) with \(\eta=0.25\), and the HHI-type concentration
index \(W^{\rm HHI}\), as defined in Example~\ref{eg:stress_summaries}. At
\(u=15\), the joint stress event has probability approximately \(10^{-3}\).
Figure~\ref{fig:stress_summary_cdfs} compares the corresponding conditional
distributions under the true model and under the population SSGEN approximation.
The close agreement across the three panels illustrates the transfer principle in
Proposition~\ref{prop:stress_summaries}: once the stressed law is approximated
correctly, the induced laws of severity, breadth, and concentration summaries are
also preserved. Equivalently, SSGEN preserves the plausibilities of downstream
stress patterns, such as unusually large aggregate loss, broad distress across
agents, or high loss concentration.

\begin{figure}[htbp]
    \centering
    \begin{subfigure}[b]{0.32\textwidth}
        \centering
        \includegraphics[width=\textwidth]{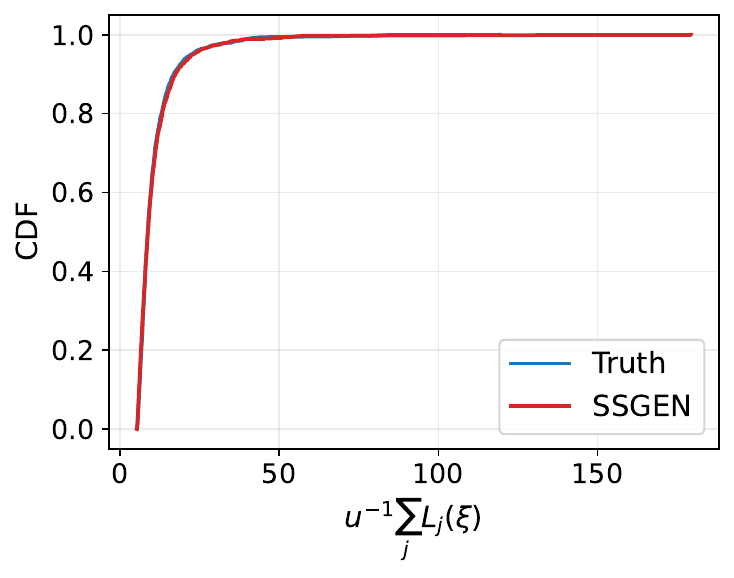}
        \caption{Aggregate stressed loss}
        \label{fig:cdf_aggregate_u15}
    \end{subfigure}
    \hfill
    \begin{subfigure}[b]{0.32\textwidth}
        \centering
        \includegraphics[width=\textwidth]{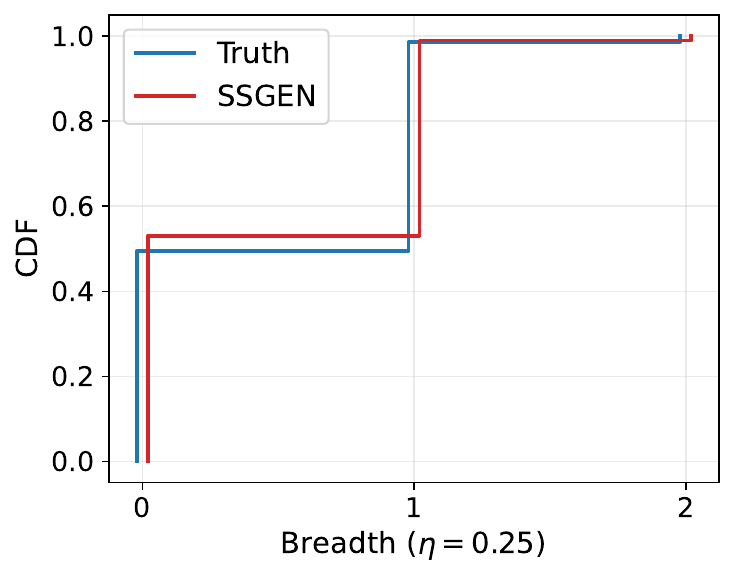}
        \caption{Breadth at \(\eta=0.25\)}
        \label{fig:cdf_breadth_u15}
    \end{subfigure}
    \hfill
    \begin{subfigure}[b]{0.32\textwidth}
        \centering
        \includegraphics[width=\textwidth]{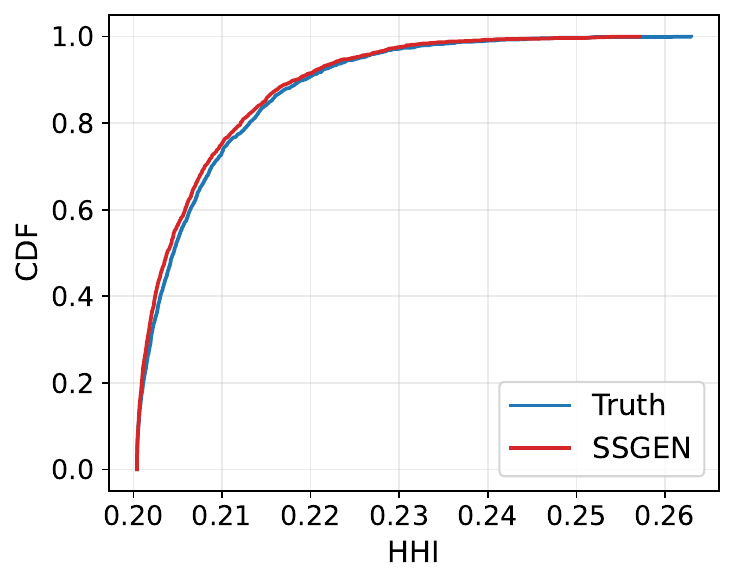}
        \caption{HHI concentration}
        \label{fig:cdf_hhi_u15}
    \end{subfigure}
    \caption{Empirical CDFs of downstream stress summaries under the true stressed
    law and the population SSGEN approximation at \(u=15\). From left to right:
    aggregate stressed loss, breadth of materially affected institutions, and
    HHI-type concentration of stressed losses.}
    \label{fig:stress_summary_cdfs}
\end{figure}

Finally, we illustrate the reverse-stress result. For each stress level, we compare
the scaled reverse-stress solutions under the true model and under the population
SSGEN approximation. Figure~\ref{fig:rst_boundary_comparison} shows that the
selected reverse-stress points remain close in scaled coordinates, with nearby
contours of the scaled likelihood objective also well aligned. This illustrates
Theorem~\ref{thm:rst_convergence}: the population splice preserves not only the
stressed conditional law, but also the first-order plausibility geometry of the
reverse-stress problem. Thus SSGEN recovers stress mechanisms that are
asymptotically as plausible as those selected under the true tail model.

\begin{figure}[htbp]
    \centering
    \begin{subfigure}[b]{0.32\textwidth}
        \centering
        \includegraphics[width=\textwidth]{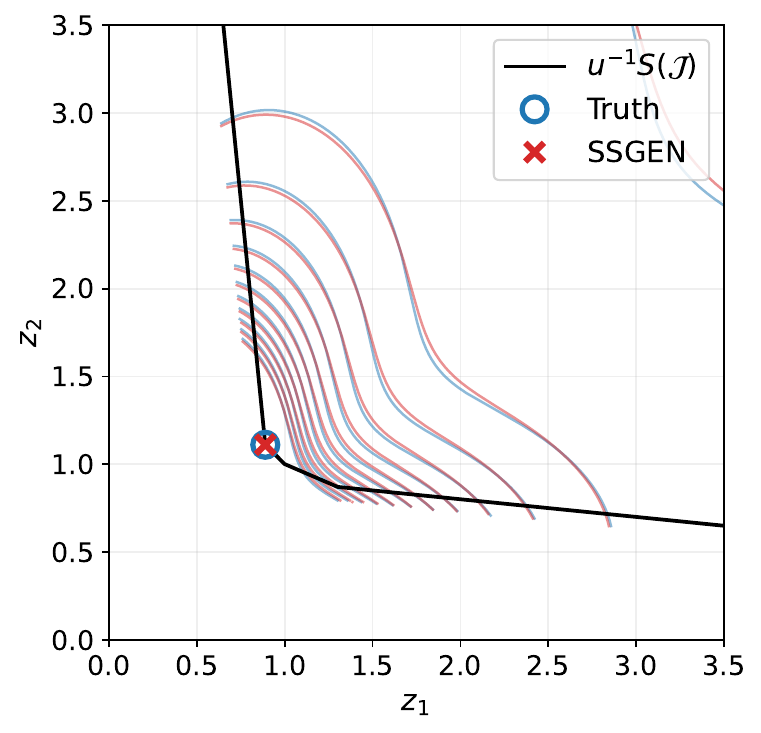}
        \caption{\(u=5\)}
        \label{fig:rst_u5}
    \end{subfigure}
    \hfill
    \begin{subfigure}[b]{0.32\textwidth}
        \centering
        \includegraphics[width=\textwidth]{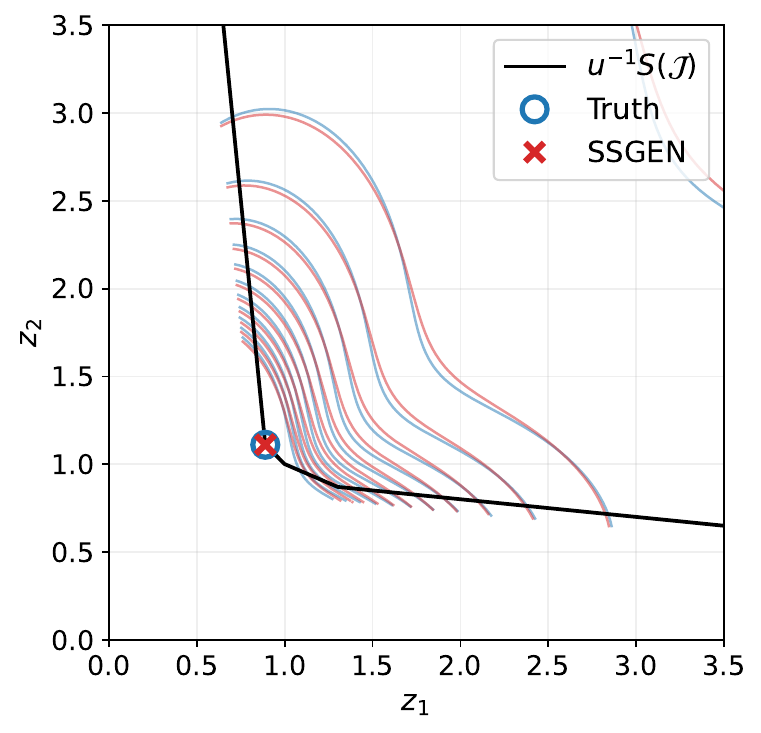}
        \caption{\(u=15\)}
        \label{fig:rst_u15}
    \end{subfigure}
    \hfill
    \begin{subfigure}[b]{0.32\textwidth}
        \centering
        \includegraphics[width=\textwidth]{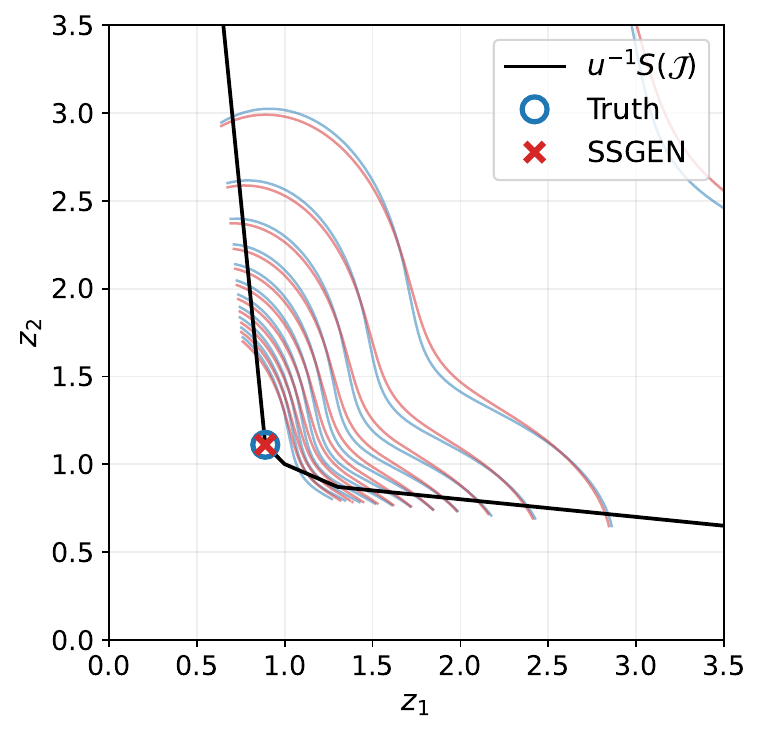}
        \caption{\(u=25\)}
        \label{fig:rst_u25}
    \end{subfigure}
    \caption{Reverse stress testing in the reinsurance example for
    \(u=5,15,25\). The contours are level sets of the scaled likelihood objective
    over the scaled state space, and the markers indicate the reverse-stress
    solutions under the true model and under the population SSGEN approximation.}
    \label{fig:rst_boundary_comparison}
\end{figure}

\section{Data-Driven Convergence Guarantees}\label{sec:data_driven_guarantees}
The population splice in \eqref{eqn:f_tgen} is tail preserving, but it is not
available from data. Both the tail index \(s\) and the angular
law must be estimated from a finite threshold,
whereas the guarantee of Theorem~\ref{thm:heavy_vanishing_RE} is on the limiting
tail law. Two errors therefore arise: (i) a statistical error from estimating
with finitely many exceedances, and (ii) a finite-threshold bias. This section precisely quantifies these to arrive at rates of convergence for SSGEN.

Let $\{\XX_1,\ldots,\XX_n\}\sim P$ be independent observations. Set
\(R_i=\|\XX_i\|\), and let \(R_{(k)}\) denote the \(k\)th largest value among
\(R_1,\ldots,R_n\). We train SSGEN at an intermediate radial threshold
\[
    k_n:=\lfloor n^{1-q}\rfloor,\qquad
    t_n:=R_{(k_n)},\qquad q\in(0,1),
\]
so that \(k_n\to\infty\) and \(k_n/n\to0\). The angular learner is trained on
the exceedance directions
\[
    \Theta_{\mathcal Y}
    =
    \left\{\frac{\XX_i}{\|\XX_i\|}:R_i\ge t_n\right\}.
\]
From these exceedances we estimate the tail index and the conditional angular
law above \(t_n\), obtaining \(\widehat s_{t_n}\) and
\(\widehat\Phi^{(t_n)}\), and we write \(\widehat P_{\rm gen}^{(t_n)}\) for the
resulting data-driven SSGEN law.
Assumption~\ref{assume:data_driven} below controls statistical estimation from the
\(k_n\) exceedances; Assumption~\ref{assume:2nd_order} controls the bias
incurred by replacing the pre-asymptotic tail with its self-similar limit.

\begin{assumption}[\textbf{Estimation rates}]
\label{assume:data_driven}
Suppose the following hold:
\begin{enumerate}
    \item[(i)] There exists \(m\in(0,1/2]\) such that
    \[
        |\widehat s_{t_n}-s|=O_P(k_n^{-m}).
    \]
    \item[(ii)] The population angular law \(\Phi^{(t_n)}\) and its learned
    approximation \(\widehat\Phi^{(t_n)}\) admit densities
    \(\varphi^{(t_n)}\) and \(\widehat\varphi^{(t_n)}\) on
    \(\Theta_{\mathcal E}\), and there exists \(\ell\in(0,1/2]\) such that
    \begin{equation}\label{eqn:learning_rates}
           \mathsf{TV}\!\left(
        \Phi^{(t_n)},\widehat\Phi^{(t_n)}
        \right)
        =
        O_P(k_n^{-\ell}).
    \end{equation}
\end{enumerate}
\end{assumption}

Assumption~\ref{assume:data_driven}(i) is the radial tail-index rate. Several standard estimators for tail index may be used to estimate $s$. For example, under standard second-order regular-variation conditions and an
appropriate choice of the intermediate sequence, the reciprocal Hill estimator
satisfies the above condition with $m=1/2$; see \citealp[Theorem~3.2.5]{deHaan}. We point out that as with the rest of the analysis, our assumption is at a general level. 

Assumption~\ref{assume:data_driven}(ii) is likewise learner-agnostic.
For example, let \(\Phi\) be an angular law with a Sobolev-smooth density of
order \(\alpha\), and let \(\widehat\Phi_N\) denote an estimator of \(\Phi\)
based on \(N\) independent observations. Then \cite{liang2021well} provides
estimators satisfying
\[
    \mathbb{E}\!\left[
        \mathsf{TV}\!\left(\Phi,\widehat\Phi_N\right)
    \right]
    =
    O\!\left(
        N^{-\left(\frac{\alpha}{2\alpha+d-1}\wedge\frac12\right)}
    \right).
\]
By Markov's inequality, the same rate holds in probability. Taking
\(N=k_n\), this benchmark corresponds to $\ell=\frac{\alpha}{2\alpha+d-1}\wedge\frac12$
in Assumption~\ref{assume:data_driven}(ii).

\begin{assumption}[\textbf{Second-order Error}]
\label{assume:2nd_order}
Let $f_{\XX}$ be continuous on $B_{r_0}^c \cap \mathcal E$ for some $r_{0} > 0$. Further suppose
there exists \(\gamma>0\) such that
\[
\sup_{r>t,\;\pphi\in \Theta_{\mathcal E}}
\left|
r^{s+d}f_{\XX}(r\pphi)-\varphi^\star(\pphi)
\right|
=
O(t^{-\gamma}),
\qquad t\to\infty .
\]
\end{assumption}

Assumption~\ref{assume:2nd_order} gives the tail-bias rate. It controls the
residual pre-asymptotic error in replacing the finite-threshold tail by the
Pareto radial-angular form used by SSGEN. It is the density-level analogue
of the second-order conditions standard in extreme-value statistics; see, for
example, \cite{einmahl1993estimating,einmahl2021testing} and \citealp[Assumption~7.4.3]{deHaan}. 

\subsection{Convergence Rates for Stress Laws}
We now consider a target stress level that is rarely, if at all, observed in
data. Let \(\beta(n)=n^{-r}\) with \(r\geq 1\), and set
\[
    u_n:=v_{1-\beta(n)}(R),
\]
where \(v_{1-\beta(n)}(R)\) is the \((1-\beta(n))\)-quantile of \(R=\|\XX\|\).
Under the regularity conditions on the limiting stress region, the target
stress event \(\mathcal S(u_n;\mathcal J)\) has probability of order
\(\beta(n)\), so the expected number of observations in it is
\(n\beta(n)=n^{1-r}\to c<\infty\). SSGEN is trained instead at the intermediate level
\(t_n=R_{(k_n)}\), from \(k_n=\lfloor n^{1-q}\rfloor\) exceedances: the
angular law is learned on an observable exceedance region, and Pareto radial
extrapolation carries the model to the target scale \(u_n\).

For a deterministic threshold \(t\), let \(P_{\rm gen}^{(t)}\) denote the
population SSGEN splice in \eqref{eqn:f_tgen}; at the empirical threshold
\(t_n\), the data-driven law \(\widehat P_{\rm gen}^{(t_n)}\) is obtained by
replacing \(s\) with \(\widehat s_{t_n}\) and \(\Phi^{(t_n)}\) with
\(\widehat\Phi^{(t_n)}\). Define the true scaled stressed law
\[
\Pi_n
:=
\Law\!\left(
u_n^{-1}\XX
\,\middle|\,
\XX\in \mathcal S(u_n;\mathcal J)
\right),
\qquad \XX\sim P,
\]
and its data-driven SSGEN counterpart
\[
\widehat \Pi_n
:=
\Law\!\left(
u_n^{-1}\widehat \XX_n
\,\middle|\,
\widehat \XX_n\in \mathcal S(u_n;\mathcal J)
\right),
\qquad
\widehat \XX_n\sim \widehat P_{\rm gen}^{(t_n)} .
\]
For laws \(P\) and \(Q\) with densities \(p\) and \(q\), write
\[
    \textsf{TV}(P,Q)
    :=
    \sup_B |P(B)-Q(B)|
    =
    \frac12\int |p-q|,
\]
so that a total-variation bound controls the error uniformly over stress
events \(B\). Theorem~\ref{thm:rate_TV} shows that SSGEN controls the full
stressed law in this metric, not merely selected bounded summaries.

\begin{theorem}\label{thm:rate_TV}
Suppose that \((L_1,\ldots,L_K)\in\mathcal H_1\) is regular, and that
Assumptions~\ref{assume:data_driven}--\ref{assume:2nd_order} hold. Then
\begin{equation}\label{eqn:rate_TV}
    \textsf{TV}\!\left(\widehat\Pi_n,\Pi_n\right)
    =
    O_P\!\left(n^{-(1-q)(m\wedge \ell)} + n^{-q\gamma/s}\right).
\end{equation}
\end{theorem}

Theorem~\ref{thm:rate_TV} separates the two sources of error. The term
\(n^{-(1-q)(m\wedge\ell)}\) is statistical: it comes from estimating the tail
index and the angular law from the \(k_n\) exceedances. The exponent \(\ell\)
captures the difficulty of angular learning; when \(\ell<1/2\), this
nonparametric step is slower than the parametric tail-index rate. The term
\(n^{-q\gamma/s}\) is the finite-threshold bias from replacing the
pre-asymptotic tail by its self-similar limit. The two terms move in opposite
directions in \(q\): a larger \(q\) trains deeper in the tail, shrinking the
bias, but leaves fewer exceedances for learning, inflating the statistical
error. The choice of \(q\) balances the two. 

\begin{corollary}\label{cor:opt_learning}
Under the assumptions of Theorem~\ref{thm:rate_TV}, the oracle choice
\[
    q^\star
    =
    \frac{s(m\wedge \ell)}{\gamma+s(m\wedge \ell)}
\]
balances statistical error and finite-threshold bias. For this choice,
$\textsf{TV}\!\left(\widehat\Pi_n,\Pi_n\right)
=
O_P\!\left(
n^{-\frac{\gamma(m\wedge \ell)}{\gamma+s(m\wedge \ell)}}
\right)$. Any other fixed \(q\neq q^\star\) gives a smaller polynomial exponent in the upper
bound.
\end{corollary}

\begin{remark}\label{rem:diagnostic}\em
Corollary~\ref{cor:opt_learning} is an oracle statement, not a data-driven
threshold rule: the quantities \(\gamma\) and \(\ell\) are generally unknown. Its
purpose is to identify the best possible balance, within this upper bound, between
estimation error and finite-threshold bias.
\end{remark}
Theorem~\ref{thm:rate_TV} controls the full stressed conditional law; the
rate transfers without loss to downstream stress diagnostics. As in
Example~\ref{eg:stress_summaries}, define
\[
    \mv L_{u_n}(\zz)
    :=
    \left(
        u_n^{-1}L_1(u_n\zz),\ldots,u_n^{-1}L_K(u_n\zz)
    \right),
    \qquad
    W_{u_n}(\zz):=W(\mv L_{u_n}(\zz)),
\]
where \(W:\mathbb R^K\to\mathsf Y\subset\mathbb R^M\) is measurable, and let
\(G_n:=\Pi_n\circ W_{u_n}^{-1}\) and
\(\widehat G_n:=\widehat\Pi_n\circ W_{u_n}^{-1}\) denote the laws of the
diagnostic \(W_{u_n}\) under the true and generated stressed laws.

\begin{corollary}[\textbf{Stress-summary distributions}]
\label{cor:stress_summary_dd}
Under the conditions of Theorem~\ref{thm:rate_TV}, for every measurable
stress summary \(W_{u_n}:\mathcal E\to\mathsf Y\),
\[
    \textsf{TV}\!\left(\widehat G_n,G_n\right)
    =
    O_P\!\left(
        n^{-(1-q)(m\wedge \ell)}
        +
        n^{-q\gamma/s}
    \right).
\]
\end{corollary}

Corollary~\ref{cor:stress_summary_dd} is the data-driven counterpart of
Proposition~\ref{prop:stress_summaries}. The proposition identifies the
limiting law of stress summaries as \(u\to\infty\), while the corollary
controls the distance between the true and generated diagnostic laws at the
target level \(u_n\). The transfer is lossless, since total variation
contracts under measurable maps; no smoothness, continuity, or
dimension-dependent term is incurred.

\subsection{Convergence Rates for Reverse Stress Testing}\label{sec:dd_rst}
We next study rates of convergence for the reverse-stress solutions produced by
the data-driven version of SSGEN. Passing from objective approximation to
optimizer approximation requires additional regularity: small perturbations of an
objective need not imply small perturbations of its maximizers unless the true
problem is locally well identified. We therefore impose a quantitative separation
condition on the finite-level reverse-stress objective and combine it with the
objective approximation rate delivered by SSGEN.

Let \(\widehat P_{\rm gen}^{(t_n)}\) denote the data-driven SSGEN law constructed
at the intermediate threshold \(t_n\). For the target stress level \(u_n\), define
\[
    \Psi_n(\zz)
    :=
    u_n^{s+d}f_P(u_n\zz),
    \qquad
    \widehat\Psi_n(\zz)
    :=
    u_n^{s+d}f_{\widehat P_{\rm gen}^{(t_n)}}(u_n\zz).
\]
The true and data-driven reverse-stress solution sets are
\[
    Y_n
    :=
    \arg\max_{\zz\in\Gamma_{u_n}(\mathcal J)}
    \Psi_n(\zz),
    \qquad
    \widehat Y_n
    :=
    \arg\max_{\zz\in\Gamma_{u_n}(\mathcal J)}
    \widehat\Psi_n(\zz).
\]
Our goal is to bound the distance from \(\widehat Y_n\) to \(Y_n\). Let
\(\mathcal K_{\mathcal J}\subset\mathcal E\setminus\{\mathbf 0\}\) be the compact
localization set from Lemma~\ref{lem:containment}; for all sufficiently large
\(n\), the relevant true optimizers lie in this set. In what follows, we will assume that the arg-max in $\hat Y_n$ is attained.

\begin{assumption}[\textbf{Local identifiability of the RST problem}]
\label{assume:rst_margin}
Let
\[
    V_n
    :=
    \sup_{\zz\in\Gamma_{u_n}(\mathcal J)}
    \Psi_n(\zz).
\]
There exist constants \(c>0\), \(\kappa\ge1\), and \(\delta_0>0\) such that, for
all sufficiently large \(n\),
\begin{equation}\label{eqn:A4}
    V_n-\Psi_n(\zz)
    \ge
    c\,
    \left\{
    d(\zz,Y_n)\wedge \delta_0
    \right\}^{\kappa},
    \qquad
    \zz\in \Gamma_{u_n}(\mathcal J)\cap\mathcal K_{\mathcal J}.
\end{equation}
\end{assumption}

Assumption~\ref{assume:rst_margin} is a quantitative identification condition for
the finite-level reverse-stress problem. Within a \(\delta_0\)-neighborhood of
\(Y_n\), it requires the objective loss to grow at least polynomially with the
distance from the solution set; outside this neighborhood, it imposes the
uniform separation \(c\delta_0^\kappa\). Such margin conditions are standard in
perturbation analyses of optimization problems; see, for example,
\citealp[Section~4.4.1]{bonnans2013perturbation}. They convert uniform objective
error into a bound on $\sup_{\widehat\zz\in\widehat Y_n}
    d(\widehat\zz,Y_n)$. Since distance is measured to the full solution set \(Y_n\), \eqref{eqn:A4} allows
multiple competing reverse-stress mechanisms. Because the data-driven analysis targets the finite-level solution set \(Y_n\),
we compare \(\widehat\Psi_n\) directly with \(\Psi_n\) on the localized region
where maximizers can occur.

\begin{lemma}[\textbf{Data-driven objective approximation}]
\label{lem:rst_dd_local_objective}
Suppose Assumptions~\ref{assume:ht_data}, \ref{assume:data_driven}(i), and
\ref{assume:2nd_order} hold and let $(L_1,\ldots, L_K) \in \mathcal H_1$ be regular. Suppose also that
\begin{equation}\label{eqn:unif_condition}
      \|\widehat\varphi^{(t_n)}-\varphi^{(t_n)}\|_\infty
    =
    O_P(k_n^{-\ell}).
\end{equation}
Then, for $\mathcal K_{\mathcal J} $ compact with $\mv 0\not\in \mathcal K_{\mathcal J}$
\[
    \Delta_n
    :=
    \sup_{\zz\in\Gamma_{u_n}(\mathcal J)\cap\mathcal K_{\mathcal J}}
    \left|
    \widehat\Psi_n(\zz)-\Psi_n(\zz)
    \right|
    =
    O_P(a_n), \quad \text{where}\quad 
    a_n
    :=
    k_n^{-\ell}
    +
    k_n^{-m}\log(u_n/t_n)
    +
    t_n^{-\gamma}.
\]
\end{lemma}

Lemma~\ref{lem:rst_dd_local_objective} gives uniform control of the
statistical error in the data-driven RST objective over the localized feasible
set. The first and third terms in \(a_n\) are, respectively, the angular
estimation error and the finite-threshold bias appearing in
Theorem~\ref{thm:rate_TV}. The tail-index error carries an additional factor
\(\log(u_n/t_n)\), because \(\widehat s_{t_n}\) enters the extrapolated density
through \((u_n/t_n)^{-\widehat s_{t_n}}\). Perturbing the exponent over the
extrapolation range therefore yields
\(k_n^{-m}\log(u_n/t_n)\), with
\(\log(u_n/t_n)\asymp\log n\) under the adopted threshold sequences.

The localization is asymptotically without loss. Indeed,
Lemma~\ref{lem:rst_dd_localization}, proved in the appendix, shows that
\(\widehat Y_n\subseteq\mathcal K_{\mathcal J}\) with probability tending to
one, while Lemma~\ref{lem:containment} gives the corresponding containment for
\(Y_n\). Thus, with probability tending to one, the comparison of the two
argmax problems takes place entirely on
\(\Gamma_{u_n}(\mathcal J)\cap\mathcal K_{\mathcal J}\), where
Lemma~\ref{lem:rst_dd_local_objective} provides uniform objective control.
Theorem~\ref{thm:rst_solution_rate} then converts this control into a rate for
the distance between the estimated and population solution sets.

\begin{theorem}[\textbf{Rate for data-driven reverse-stress solutions}]
\label{thm:rst_solution_rate}
Suppose the conditions of Lemma~\ref{lem:rst_dd_local_objective} and
Assumption~\ref{assume:rst_margin} hold. Then with $a_n$ as in Lemma~\ref{lem:rst_dd_local_objective},
\[
    \sup_{\widehat\zz\in\widehat Y_n}
    d(\widehat\zz,Y_n)
    =
    O_P(a_n^{1/\kappa}),
\]
In particular, if \(Y_n=\{\zz_n^\star\}\) is a singleton, then every
\(\widehat\zz_n\in\widehat Y_n\) satisfies $\|\widehat\zz_n-\zz_n^\star\|
    =
    O_P\!\left(a_n^{1/\kappa}
    \right)$.
\end{theorem}
Theorem~\ref{thm:rst_solution_rate} is stated in terms of \(k_n\), the number of intermediate
exceedances used for tail learning. Note that $k_n \sim n^{1-q}$ and thus, $
    t_n=\Theta_P\!\left(n^{q/s}\right)$.
Moreover, the target regime $\mathbb P\bigl(L_j(\XX)\geq c_j u_n,\ j\in\mathcal J\bigr)
    \asymp n^{-r}$ implies \(u_n=\Theta(n^{r/s})\). Consequently, $\log(u_n/t_n)
    =
    O_P(\log n)$.
It follows that
\[
    a_n
    =
    O_P\!\left(
        n^{-(1-q)\ell}
        +
        n^{-(1-q)m}\log n
        +
        n^{-q\gamma/s}
    \right),
\]
and hence
\[
    \sup_{\widehat\zz\in\widehat Y_n}
    d(\widehat\zz,Y_n)
    =
    O_P\!\left(
        \left\{
            n^{-(1-q)\ell}
            +
            n^{-(1-q)m}\log n
            +
            n^{-q\gamma/s}
        \right\}^{1/\kappa}
    \right).
\]

This rate has the same  structure as
Theorem~\ref{thm:rate_TV}. Angular learning contributes
\(n^{-(1-q)\ell}\), while finite-threshold bias contributes
\(n^{-q\gamma/s}\). Tail-index estimation incurs the additional factor
\(\log(u_n/t_n)\), since the estimated index is used to extrapolate from
\(t_n\) to \(u_n\), yielding the term \(n^{-(1-q)m}\log n\). Such logarithmic
amplification is standard in extreme-value theory based tail estimation; see
\citealp[Section~3]{deHaan}.

Importantly, convergence persists in the limited-sample regime \(r\geq1\).
While the expected number of target observations is of order \(n^{1-r}\), and
is therefore bounded or vanishing, \(k_n\sim n^{1-q}\to\infty\). Since all
three terms in the displayed rate vanish, SSGEN remains consistent even when
the target stress region contains only finitely many, or asymptotically no,
observations from the original sample.

\section{Numerical Experiments}\label{sec:numericals}
\subsection{Estimation of the stressed law}
\label{sec:num_stressed_law}

We illustrate the practical utility of the theory using a moderately
high-dimensional reinsurance network with \(d=10\) risk factors and \(K=20\)
agents. Let \(\mv A\in\mathbb R_+^{K\times d}\) denote the agent--object exposure
matrix, and take
\[
    \mv L(\xx)=\mv A\xx .
\]
The matrix \(\mv A\) is fixed across replications and is chosen to create one
broadly exposed hub agent together with heterogeneous locally exposed agents.
Specifically, agent \(1\) is exposed equally to all risk factors, so
\[
    A_{1\ell}=\frac1d,\qquad \ell=1,\ldots,d.
\]
Each remaining agent is exposed to two risk factors only. If agent \(i\) is
assigned the pair \((p_i,q_i)\), then
\[
    A_{i,p_i}=\alpha_i,\qquad
    A_{i,q_i}=1-\alpha_i,\qquad
    A_{i\ell}=0\quad\text{for }\ell\notin\{p_i,q_i\},
\]
where \(\alpha_i\in[0,1]\). Thus agent \(2\) is exposed to risk factors \(1\)
and \(2\), agent \(3\) to risk factors \(3\) and \(4\), agent \(4\) to risk
factors \(4\) and \(5\), and so on. This construction yields a network with one
hub-like agent and several local exposure channels, some overlapping, without
requiring the full \(20\times10\) exposure matrix to be displayed.

\noindent \textbf{Risk Factor Distribution: }Writing \(\XX=R\Pphi\), with \(R\ge1\) and \(\Pphi\) on the positive unit
sphere, we generate the risk factors from
\begin{equation}\label{eqn:numerical_density}
    f_{R,\Pphi}(r,\pphi)
    =
    C_{s,\varepsilon,\gamma}\,
    r^{-(s+1)}
    \left\{
        g_1(\pphi)
        +
        \varepsilon r^{-\gamma}g_2(\pphi)
    \right\},
    \qquad r\ge1 ,
\end{equation}
where \(g_1\) and \(g_2\) are angular densities. We set \(s=2.5\),
\(\varepsilon=0.1\), and \(\gamma=2.5\), with the angular densities specified
in our implementation. The leading term \(g_1\) determines the limiting
angular law, while the perturbation decays at rate \(r^{-\gamma}\), so that
Assumption~\ref{assume:2nd_order} holds. The law is deliberately
non-elliptical, making a fitted \(t\)-copula an informative heavy-tailed
benchmark whose elliptical dependence structure may fail to recover the
limiting angular law.

\noindent \textbf{Metrics of Comparison: }For each \(\mathcal J\) and each target stress probability \(\beta\),
we choose the stress level \(u=u_\beta(\mathcal J)\) so that
\[
    P\bigl(\XX\in\mathcal S(u;\mathcal J)\bigr)=\beta.
\]
The values of \(\beta\) are displayed in the respective figures. For each
configuration \((u,\mathcal J)\), we compute Monte Carlo reference values using
\(10^6\) independent samples from the true law in
\eqref{eqn:numerical_density}. These values approximate the target quantities
under the true stressed law
\[
    \Law\left(\XX\mid \XX\in\mathcal S(u;\mathcal J)\right).
\]
The target quantities are the following stressed risk measures: the agent-level CoVaR (see \cite{tobias2016covar}),
\(\mathrm{CoVaR}_\alpha(L_i) = \textsf{VaR}_{1-\alpha}[L_i(\XX) \mid \XX\in \mathcal S(u;\mathcal J)]\) and conditional expected shortfall
\(\mathrm{CoES}_\alpha(L_i) = \textsf{CVaR}_{1-\alpha}[L_i(\XX) \mid \XX\in \mathcal S(u;\mathcal J)]\) under the stressed law, evaluated at
\(\alpha=0.1\); and the stressed aggregated loss
\[
    E\left[
        \sum_{i=1}^K L_i
        \,\middle|\,
        \XX\in\mathcal S(u;\mathcal J)
    \right];
\]
and the HHI of the stressed expected losses, which measures how concentrated
the stress impact is across agents. The first two are agent-level tail
functionals, the third is a network-level scale summary, and the fourth is a
scale-free concentration summary.

\noindent\textbf{Experiment Setup:} We evaluate finite-sample performance over \(M=1000\) independent
replications. In every replication the training sample size is \(n=1000\), and
the number of intermediate exceedances used for SSGEN is fixed at \(k_n=100\),
corresponding to the empirical \(90\%\) radial threshold and to the choice
\(k_n=n^{1-q}\) with \(q=1/3\). For replication \(m=1,\ldots,M\), we proceed as
follows:
\begin{enumerate}
    \item[(i)] Draw an independent training sample $\XX_1^{(m)},\ldots,\XX_n^{(m)}$
    from \eqref{eqn:numerical_density}.
    \item[(ii)] Fit an angular generative model to the directions of the
    \(k_n=100\) intermediate exceedances, estimate the radial tail index, and
    construct the corresponding SSGEN tail generator. 
    \item[(iii)] Generate stressed samples from the trained SSGEN model at the
    target level \(u\), retain those falling in
    \(\mathcal S(u;\mathcal J)\), and use the retained samples to estimate the
    target risk measures and downstream stress summaries.
\end{enumerate}
In the implementation reported here, the finite-threshold angular density
\(\widehat\varphi^{(t)}\) is learned using a diffusion model, which is
re-trained independently in each replication. The diffusion model serves as
a flexible computational implementation of the angular-learning step; the EVT
radial extrapolation is then applied as specified in
Algorithm~\ref{algo:generate_extremes}.

We compare SSGEN with two benchmarks fitted to the same training samples.
\textsf{Naive} directly conditions the empirical sample on
\(\mathcal S(u;\mathcal J)\), and therefore becomes increasingly variable as the
stress event becomes rarer. Replications containing no stress observation are treated as undefined, and the box plots use only replications for which the estimator is defined. \textsf{MVT} fits a
\(t\)-copula, together with marginal models, and passes the generated scenarios
through the same conditioning and estimation pipeline. It accommodates
heavy-tailed dependence but imposes an elliptical copula and therefore need not
capture the non-elliptical angular structure in
\eqref{eqn:numerical_density}.

For a scalar target quantity \(\Theta\), the reported accuracy measure is the
absolute relative error
\[
    \operatorname{ARE}(\widehat\Theta,\Theta^{\rm true})
    :=
    100\left|
        \frac{\widehat\Theta-\Theta^{\rm true}}
             {\Theta^{\rm true}}
    \right|,
\]
where \(\Theta^{\rm true}\) denotes the corresponding Monte Carlo reference
value. For the agent-level CoVaR and CoES quantities, the error reported in
each replication is the maximum of the agent-specific relative errors:
\[
   {100} \times  \max_{i=1}^K  
    \left|
        \frac{\widehat\Theta_i-\Theta_i^{\rm true}}
             {\Theta_i^{\rm true}}
    \right|.
\]

\noindent \textbf{Results:} Across the \(M\) replications, we report box-plots of these errors for the
three methods. Training details for the diffusion implementation are provided
in Appendix~\ref{app:implement}. Figures~\ref{fig:all20_numerical_boxplots}
and~\ref{fig:hub_numerical_boxplots} report the results for full-network and
hub-only stress, respectively ($\mathcal J= [20]$ and $\mathcal J = \{1\},$ respectively). Across both stress sets, SSGEN attains the
lowest median errors for CoVaR, CoES, and aggregated loss at every value of
\(\beta\), with its advantage generally widening as the stress event becomes
rarer. For example, at \(\beta=0.005\) under full-network stress, the median
CoVaR errors are approximately \(15\%\), \(25\%\), and \(47\%\) for SSGEN,
empirical conditioning, and the fitted \(t\)-copula, respectively. Empirical
conditioning deteriorates as the number of observations in the target stress
region shrinks, whereas SSGEN uses all \(k_n\) intermediate exceedances and
extrapolates radially. The persistently large \(t\)-copula errors are
consistent with its elliptical restriction failing to capture the relevant
extremal angular structure.

The HHI results are different: all three methods perform well, with median
errors below \(1\%\) under full-network stress, and the fitted \(t\)-copula is
competitive with SSGEN. Since HHI is scale-free, errors in the overall
magnitude of stressed losses largely cancel. Thus the benefits of learning the
extremal angular law appear most clearly for scale-sensitive stress
functionals.
\begin{figure}[t]
    \centering

    \begin{subfigure}{0.49\textwidth}
        \centering
        \includegraphics[width=\linewidth]{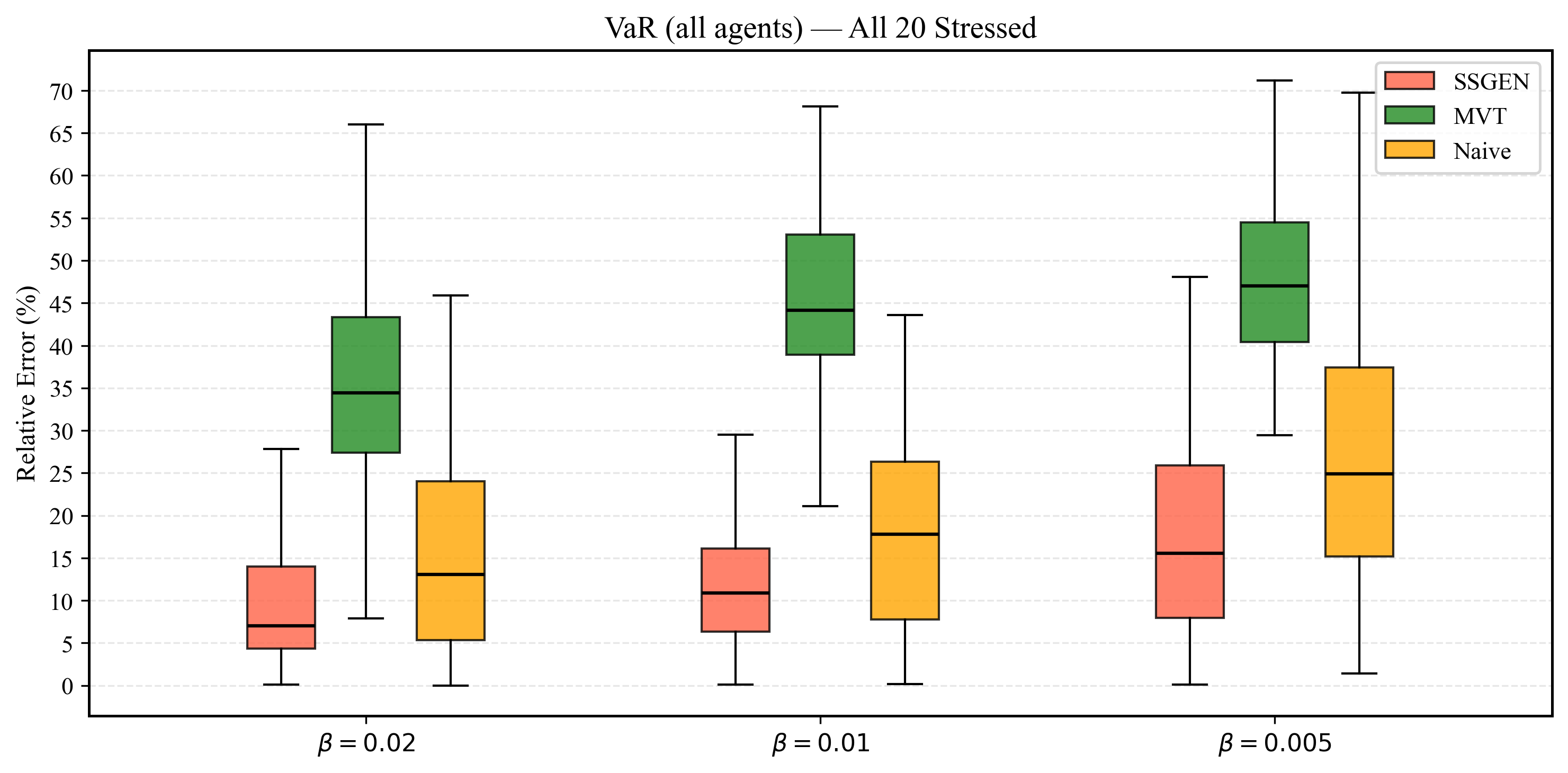}
        \caption{CoVaR}
        \label{fig:all20_var}
    \end{subfigure}
    \hfill
    \begin{subfigure}{0.49\textwidth}
        \centering
        \includegraphics[width=\linewidth]{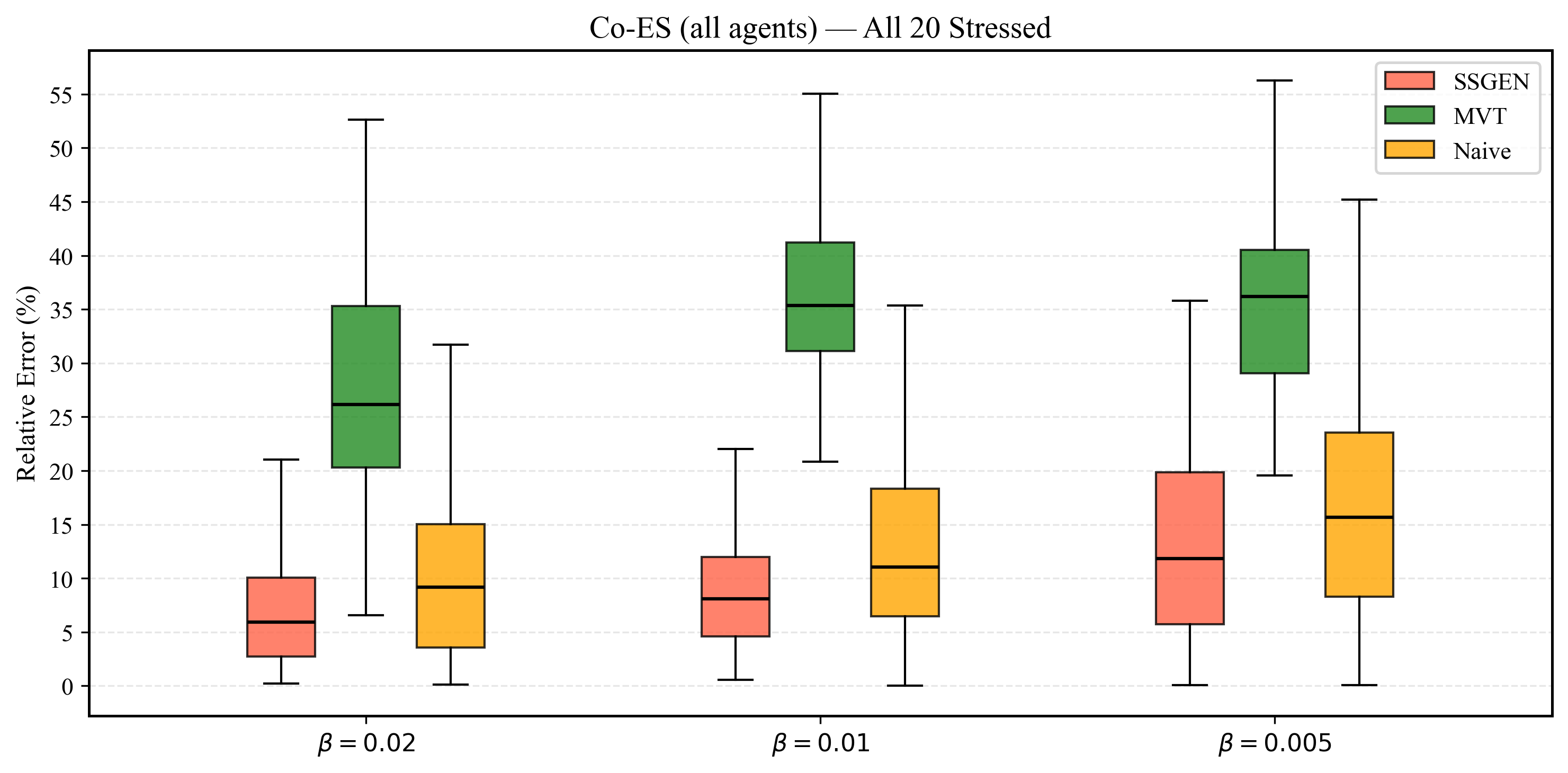}
        \caption{CoES}
        \label{fig:all20_coes}
    \end{subfigure}

    \medskip

    \begin{subfigure}{0.49\textwidth}
        \centering
        \includegraphics[width=\linewidth]{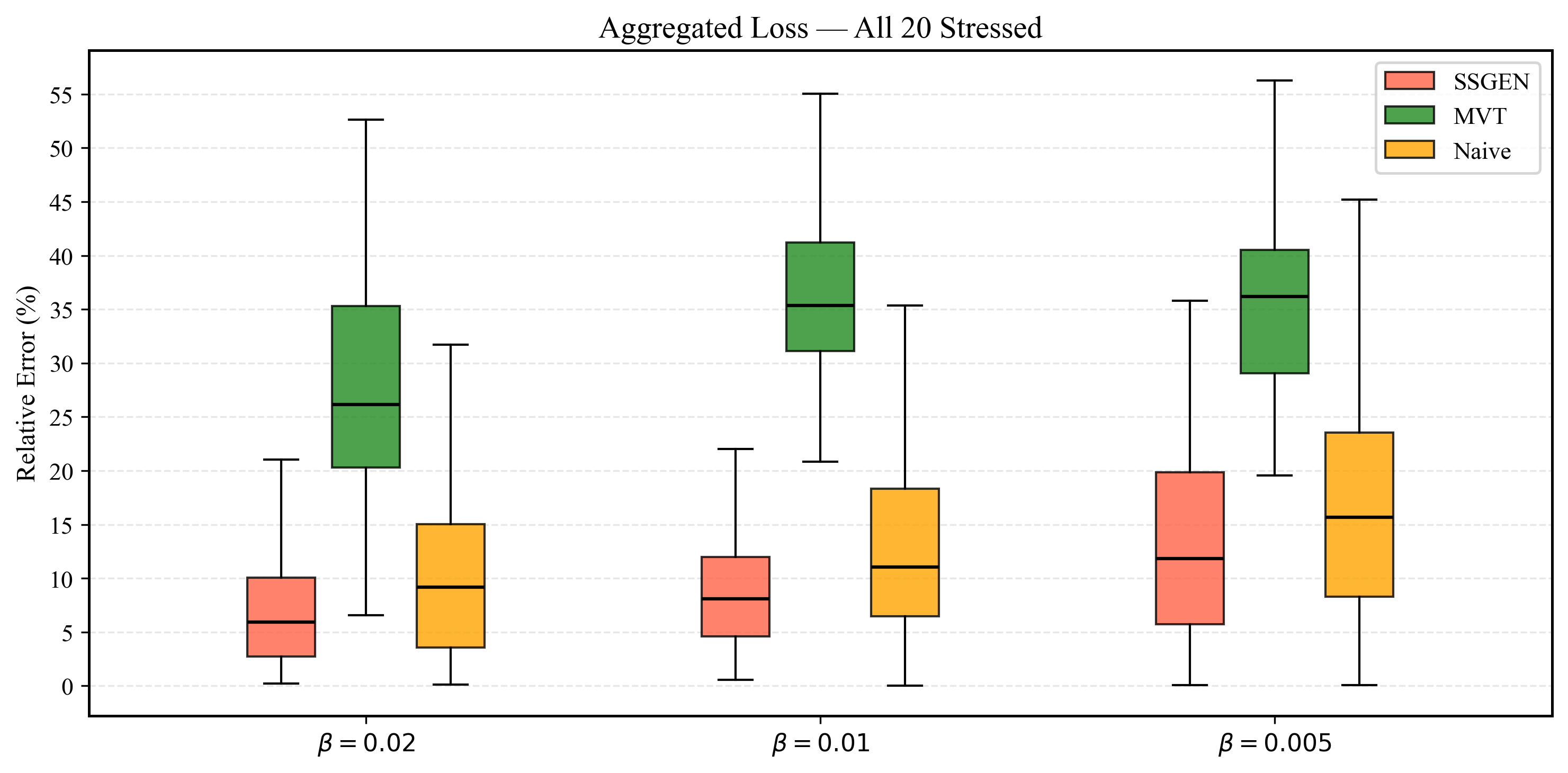}
        \caption{Aggregate loss}
        \label{fig:all20_agg_loss}
    \end{subfigure}
    \hfill
    \begin{subfigure}{0.49\textwidth}
        \centering
        \includegraphics[width=\linewidth]{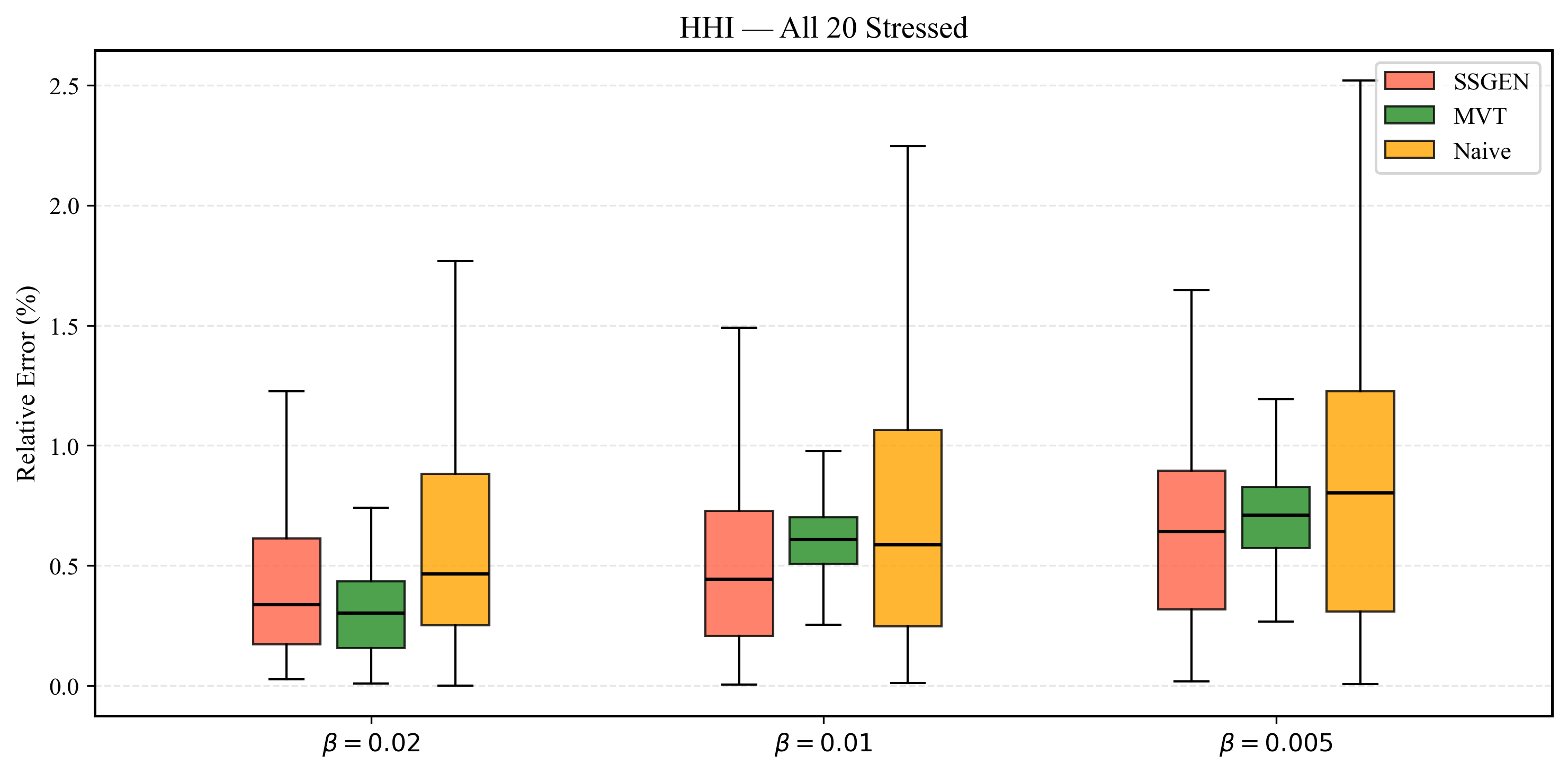}
        \caption{HHI}
        \label{fig:all20_hhi}
    \end{subfigure}

    \caption{
    Relative estimation errors for stress-law functionals in the \(d=10\),
    \(K=20\) reinsurance network, conditional on all \(20\) agents being
    stressed. Boxplots are computed across \(M\) independent training-sample
    replications and compare SSGEN with direct empirical conditioning
    \emph{(Naive)} and a fitted multivariate t-copula benchmark \emph{(MVT)}.
    }
    \label{fig:all20_numerical_boxplots}
\end{figure}

\begin{figure}[t]
    \centering

    \begin{subfigure}{0.49\textwidth}
        \centering
        \includegraphics[width=\linewidth]{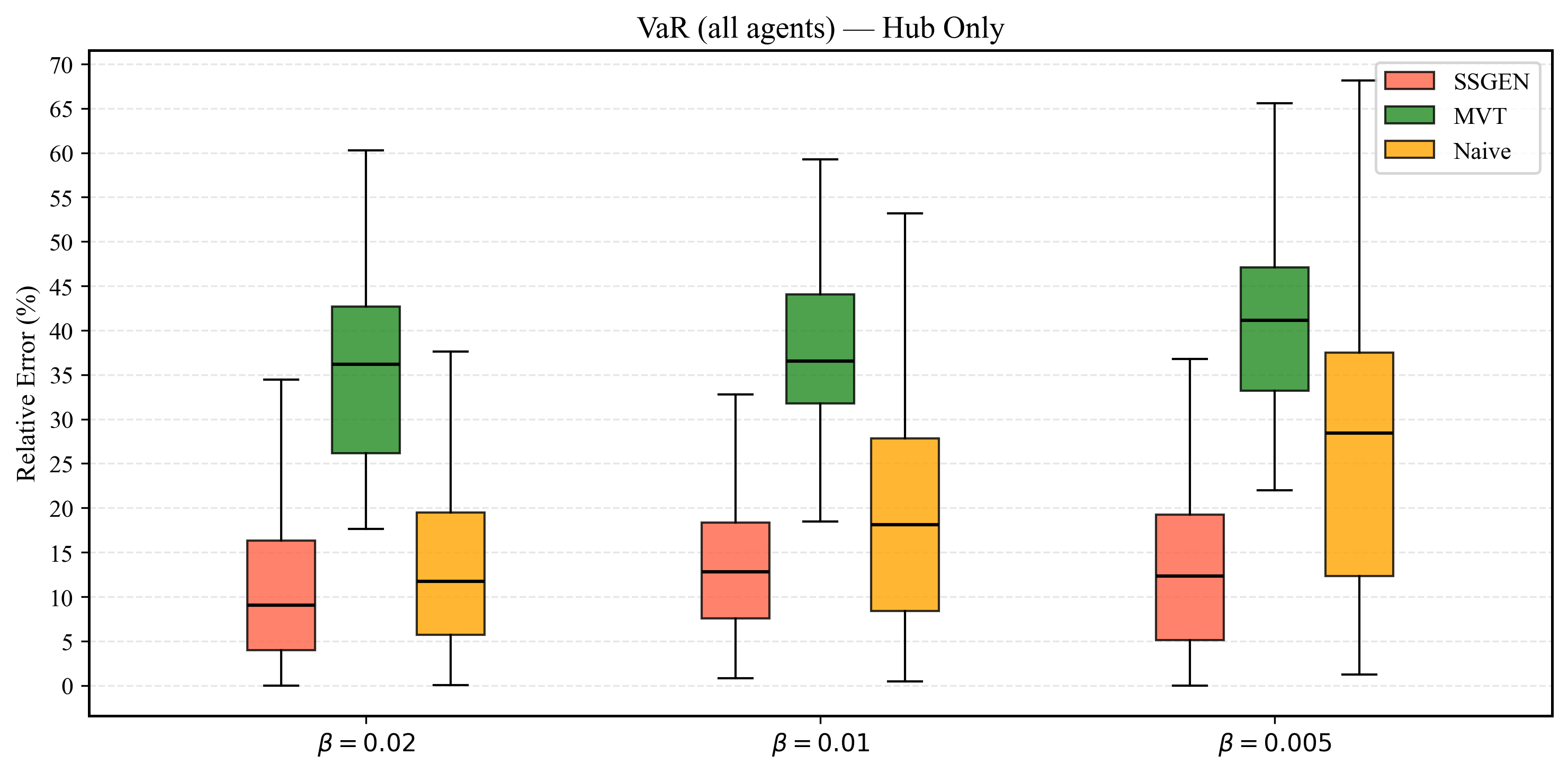}
        \caption{CoVaR}
        \label{fig:hub_var}
    \end{subfigure}
    \hfill
    \begin{subfigure}{0.49\textwidth}
        \centering
        \includegraphics[width=\linewidth]{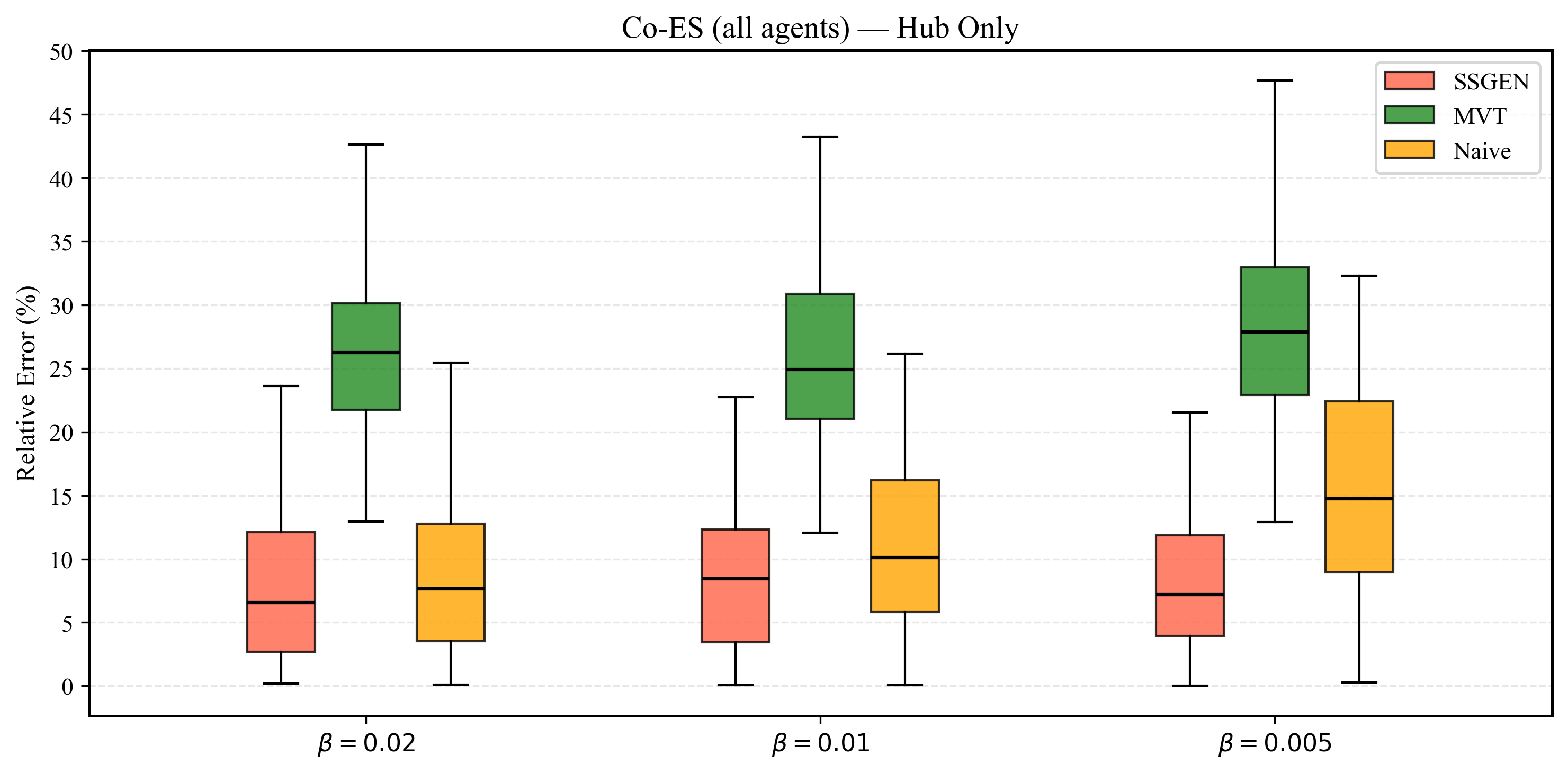}
        \caption{CoES}
        \label{fig:hub_coes}
    \end{subfigure}

    \medskip

    \begin{subfigure}{0.49\textwidth}
        \centering
        \includegraphics[width=\linewidth]{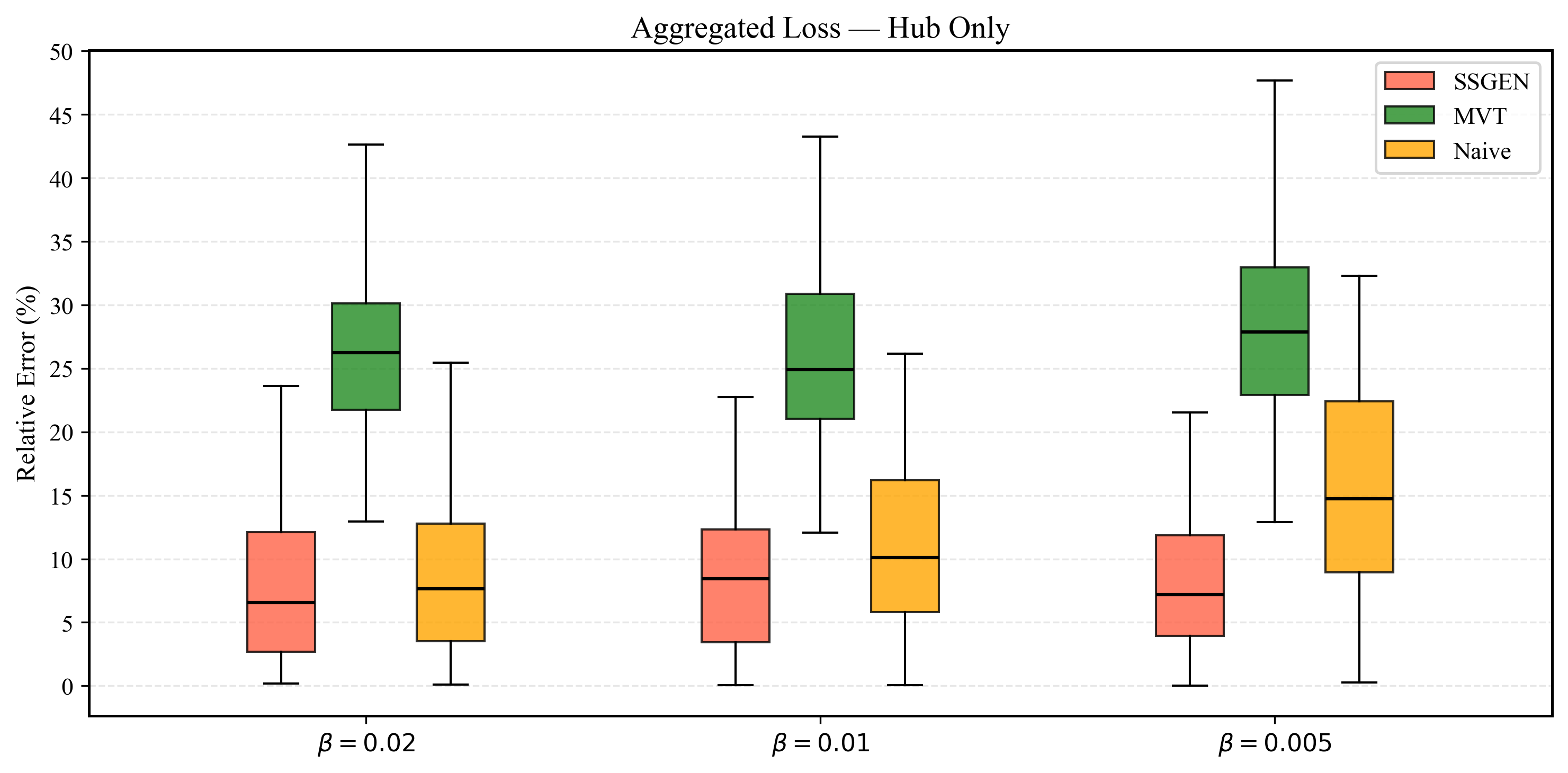}
        \caption{Aggregate loss}
        \label{fig:hub_agg_loss}
    \end{subfigure}
    \hfill
    \begin{subfigure}{0.49\textwidth}
        \centering
        \includegraphics[width=\linewidth]{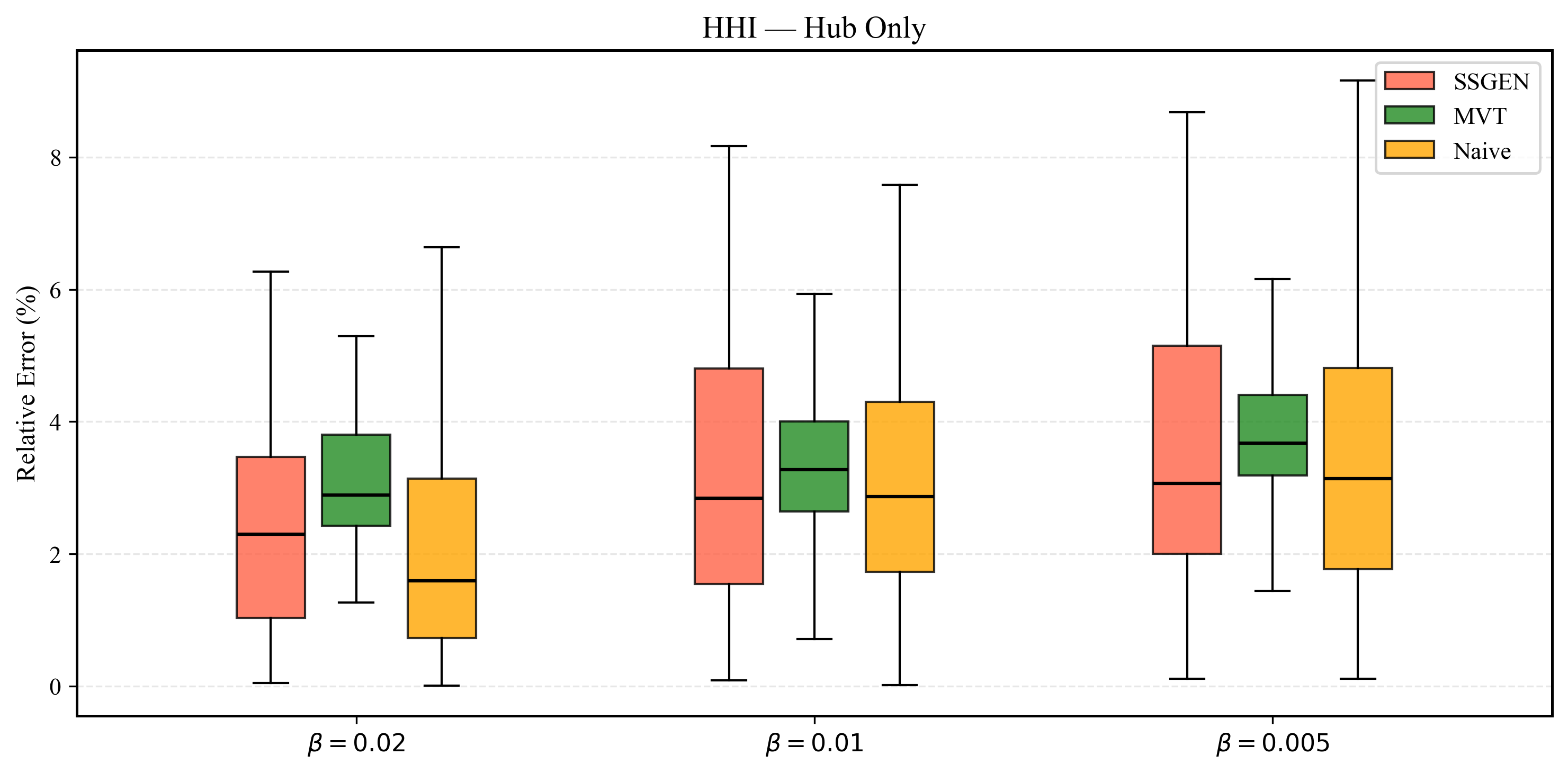}
        \caption{HHI}
        \label{fig:hub_hhi}
    \end{subfigure}

    \caption{
    Relative estimation errors for stress-law functionals in the \(d=10\),
    \(K=20\) reinsurance network, conditional on the hub agent being stressed.
    Boxplots are computed across \(M\) independent training-sample replications
    and compare SSGEN with direct empirical conditioning \emph{(Naive)} and a
    fitted multivariate-\(t\) benchmark \emph{(MVT)}.
    }
    \label{fig:hub_numerical_boxplots}
\end{figure}

\subsection{Data-driven reverse stress testing}
\label{sec:num_rst}

We evaluate data-driven reverse stress testing in a \(d=5\), \(K=5\) linear
reinsurance system of the form of Example~1. The exposure matrix
\(\mv A\in\mathbb R^{5\times 5}\) is row-stochastic, and the
capital vector is \(c=(0.5,\,0.3,\,0.4,\,0.3,\,0.5)\). The risk factors follow
the law \eqref{eqn:numerical_density} in $\R^5_+$, with tail index
\(s=2.5\) and the same non-elliptical angular structure. We condition on the
full system being stressed, \(\mathcal J=[5]\), at a common level \(u\)
calibrated so that \(P\bigl(\XX\in\mathcal S(u;\mathcal J)\bigr)=\beta\) for
\(\beta\in\{0.02,\,0.01\}\). 

Since the joint event coincides with
\(\{\min_{j\in\mathcal J}L_j(\XX)\ge u\}\), the level \(u\) is determined as
the \((1-\beta)\)-quantile of the minimum loss on a benchmark sample of
\(10^6\) observations. The reference solution \(\zz^{\mathrm{true}}\) is obtained by
maximising the exact conditional density of
\eqref{eqn:numerical_density} over \(\Gamma_u(\mathcal J)\), using
a multi-start optimisation. Each of \(M=50\) independent replications draws a training sample
of size \(n=1000\). Algorithm~\ref{algo:generate_extremes} is trained on the \(k_n=100\)
intermediate exceedance directions, and the radial tail index is estimated using
the Hill estimator. We then generate proposals from the fitted SSGEN tail model
until \(N_{\mathrm{KDE}}=5000\) scenarios satisfying
\(\mathcal S(u;\mathcal J)\) have been retained. Following the reverse-stress
formulation in \eqref{eqn:rst_problem}, a Gaussian kernel density estimate
\(\widehat g_h\), using Scott's bandwidth, is fitted to the scaled stressed
sample. The reverse-stress solution \(\widehat{\zz}\) is obtained by maximizing
\(\log\widehat g_h\) over \(\Gamma_u(\mathcal J)\) via sequential quadratic
programming under the linear stress constraints. The precise KDE-based procedure
is given in Algorithm~\ref{alg:rst_dd}.

Theorem~\ref{thm:rst_solution_rate} analyzes the reverse-stress problem associated with the fitted SSGEN density. Algorithm~\ref{alg:rst_dd} adds downstream Monte Carlo, KDE, and numerical-optimization errors. Under standard KDE mode-consistency conditions, with the bandwidth updated as the retained sample size increases, these errors can be controlled independently through the SSGEN simulation budget and solver tolerance. The numerical results below evaluate the complete implemented pipeline.

The benchmark is the \textsf{MVT} fit of
Section~\ref{sec:num_stressed_law}. Since the fitted \(t\) model has an explicit density, its
reverse-stress solution is obtained by maximising that density
directly over \(\Gamma_u(\mathcal J)\), using the same
optimisation as for SSGEN. The \textsf{MVT} benchmark therefore carries no additional smoothing error.

For each replication we record the mode coordinates \(\widehat{\zz}\), the
squared error \(\|\widehat{\zz}-\zz^{\mathrm{true}}\|^2\), and the relative
likelihood ratio
\(\widehat g_h^{\mathrm{true}}(\widehat{\zz})\,/\,
\widehat g_h^{\mathrm{true}}(\zz^{\mathrm{true}})\).
Table~\ref{tab:rst} reports the per-coordinate results, with the squared error
and likelihood ratio averaged across replications. Table~\ref{tab:rst} and Figure~\ref{fig:rst} show that SSGEN recovers the finite-level
reverse-stress solution more accurately than \textsf{MVT} at both stress levels.
For $\beta=0.02$, its mean squared solution error is $0.059$, compared
with $0.193$ for MVT, while the mean likelihood ratio is $0.875$,
compared with $0.706$. For $\beta=0.01$, the corresponding values are
$0.066$ versus $0.095$ and $0.837$ versus $0.791$. Although the
advantage narrows at the rarer level, SSGEN remains more accurate under
both diagnostics. Since these quantities are computed after scenario
generation, KDE fitting, and constrained numerical maximization, they
evaluate the complete implemented reverse-stress pipeline.

\begin{figure}[t]
    \centering
\caption{Results of the RST experiments. Panel (a) compares squared errors of the estimated reverse-stress solutions, while panel (b) compares their likelihoods relative to the true reverse-stress solution.}\label{fig:rst}
    \begin{subfigure}{0.46\textwidth}
        \centering
        \includegraphics[width=\linewidth]{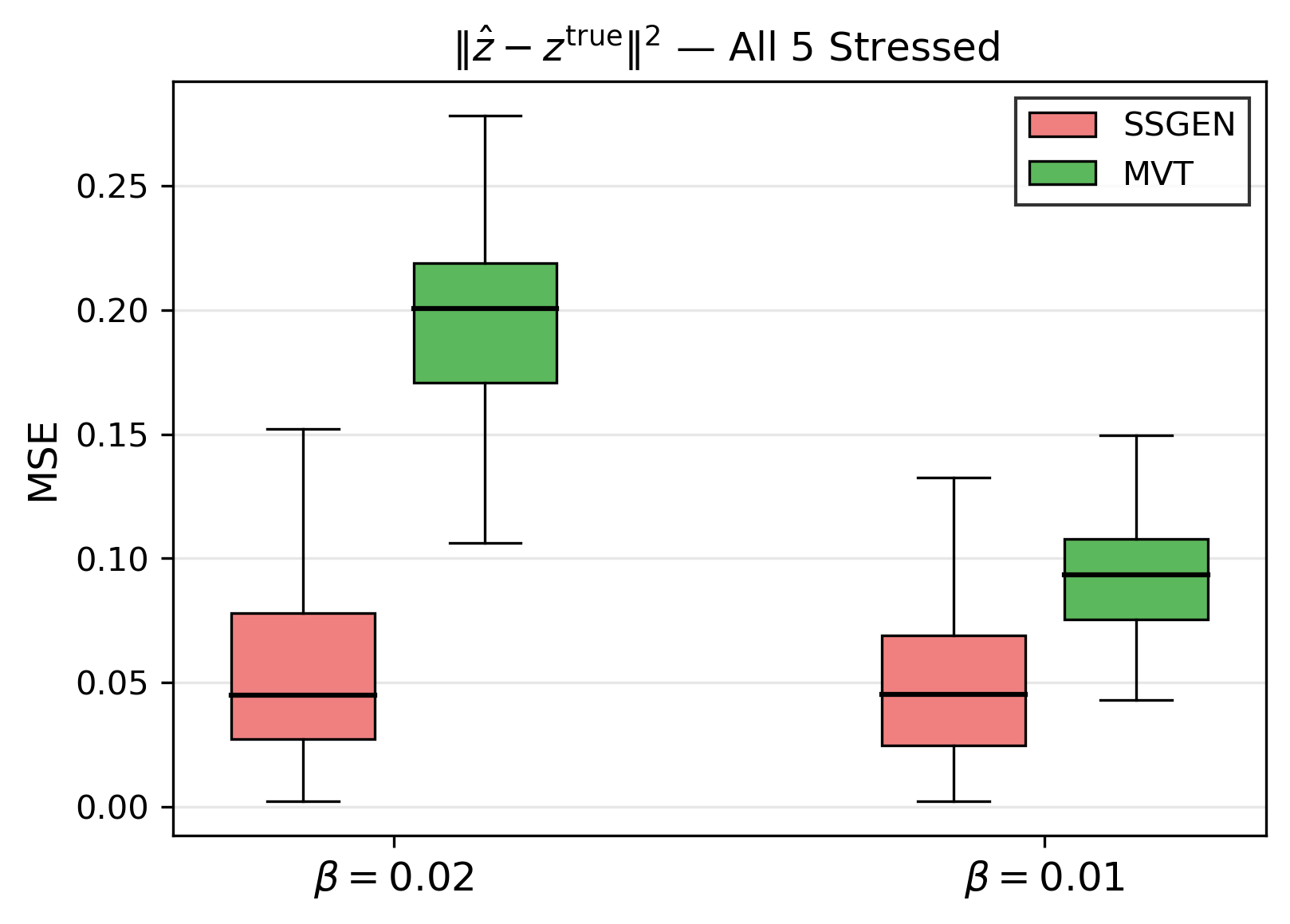}
        \caption{Squared error comparison}
        \label{fig:rst_mse}
    \end{subfigure}
    \hfill
    \begin{subfigure}{0.46\textwidth}
        \centering
        \includegraphics[width=\linewidth]{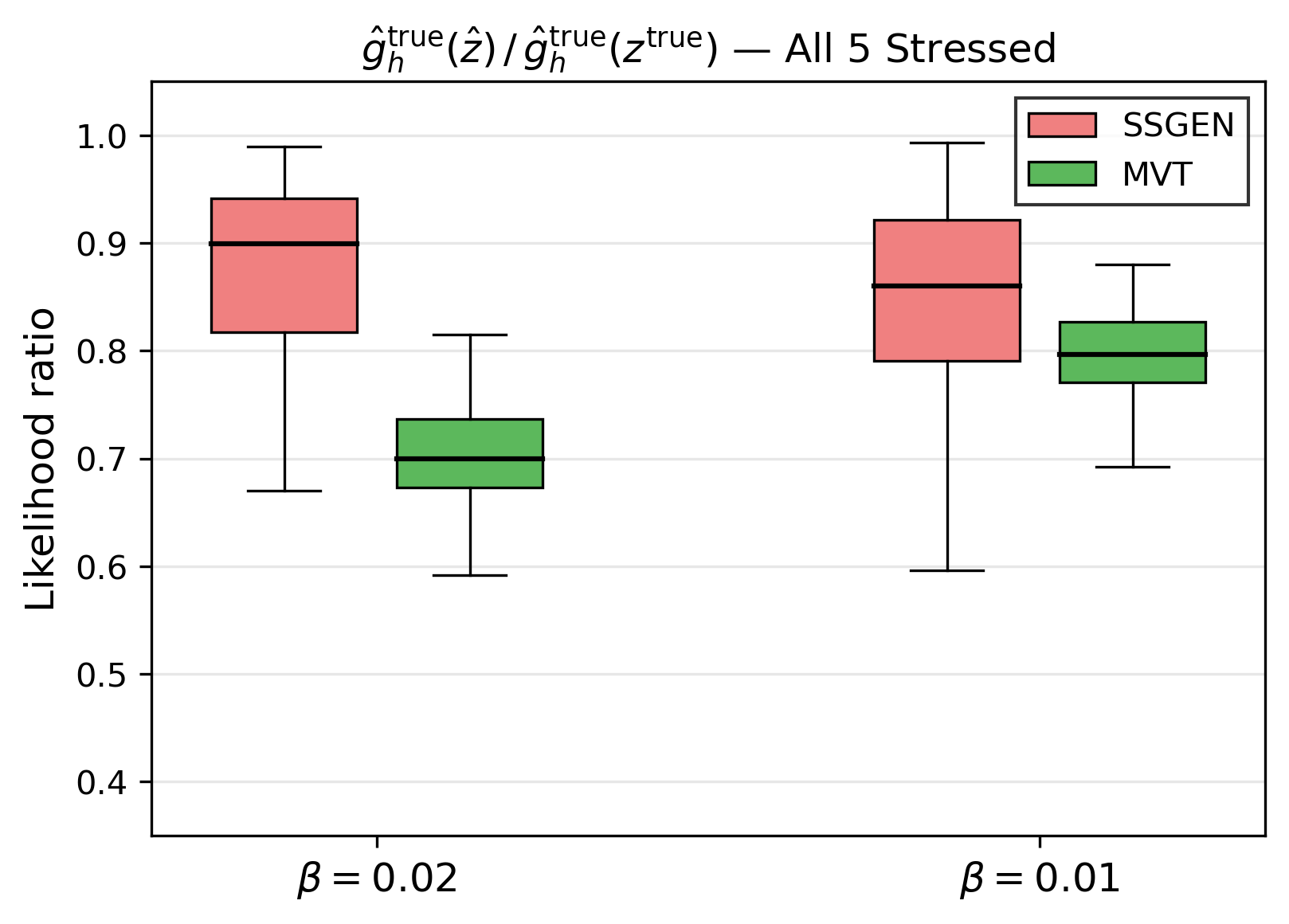}
        \caption{Likelihood ratio comparison}
        \label{fig:rst_lik}
    \end{subfigure}
\end{figure}

\begin{table}[htbp]
\centering
\caption{Reverse-stress solutions over $M$ independent replications. ``Mean'' is the
average mode coordinate, ``95\% CI'' the confidence interval for the mean,
``MSE'' the mean of
$\|\widehat{\boldsymbol{z}}-\boldsymbol{z}^{\mathrm{true}}\|^2$ across
replications, and ``Lik.\ ratio'' the mean of
$\widehat{g}_h^{\mathrm{true}}(\widehat{\boldsymbol{z}})\,/\,
\widehat{g}_h^{\mathrm{true}}(\boldsymbol{z}^{\mathrm{true}})$.}
\label{tab:rst}
\small
\begin{tabular}{c l r @{\hspace{1.2em}} r r @{\hspace{1.2em}} r r}
\toprule
& & & \multicolumn{2}{c}{SSGEN} & \multicolumn{2}{c}{MVT} \\
\cmidrule(r){4-5} \cmidrule(l){6-7}
$\beta$ & & Truth & Mean & 95\% CI & Mean & 95\% CI \\
\midrule
\multirow{7}{*}{0.02}
& $z_1$ & 1.738 & 1.729 & (1.706, 1.753) & 1.498 & (1.492, 1.504) \\
& $z_2$ & 1.883 & 1.864 & (1.837, 1.892) & 1.635 & (1.629, 1.642) \\
& $z_3$ & 1.538 & 1.538 & (1.518, 1.558) & 1.553 & (1.546, 1.561) \\
& $z_4$ & 1.906 & 1.892 & (1.871, 1.913) & 1.651 & (1.645, 1.658) \\
& $z_5$ & 1.282 & 1.282 & (1.269, 1.295) & 1.343 & (1.337, 1.349) \\
\cmidrule{2-7}
& MSE & & \multicolumn{2}{c}{0.059} & \multicolumn{2}{c}{0.193} \\
& Lik.\ ratio & & 0.875 & (0.857, 0.892) & 0.706 & (0.697, 0.716) \\
\midrule
\multirow{7}{*}{0.01}
& $z_1$ & 1.559 & 1.619 & (1.598, 1.641) & 1.385 & (1.380, 1.390) \\
& $z_2$ & 1.684 & 1.752 & (1.728, 1.776) & 1.514 & (1.508, 1.520) \\
& $z_3$ & 1.461 & 1.462 & (1.443, 1.482) & 1.466 & (1.460, 1.472) \\
& $z_4$ & 1.704 & 1.785 & (1.764, 1.805) & 1.527 & (1.521, 1.533) \\
& $z_5$ & 1.235 & 1.221 & (1.207, 1.234) & 1.270 & (1.265, 1.275) \\
\cmidrule{2-7}
& MSE & & \multicolumn{2}{c}{0.066} & \multicolumn{2}{c}{0.095} \\
& Lik.\ ratio & & 0.837 & (0.815, 0.860) & 0.791 & (0.782, 0.800) \\
\midrule
\end{tabular}
\end{table}

\FloatBarrier
\bibliographystyle{plainnat}
\bibliography{ref}
\newpage
\appendix
\section{Proofs of Main Results}
\subsection{Proofs from Section~\ref{sec:assump_model}:}  
\ \\
\noindent\textbf{Proof of Lemma~\ref{lem:phi_star_props}:} Fix  points $(\zz,\zz_0) \neq \mv 0$. Observe that $\frac{f_{\XX}(t\zz)}{f_{\XX} (t\mv z_0)} \to \frac{\varphi^\star(\zz)}{\varphi^\star(\mv z_0)}$. Therefore, $f_{\XX}$ satisfies \cite{Resnickbook}, Equation 5.32 with $\rho = -(s+d)$. Consequently, $\varphi^\star$ is homogeneous of order $-(s+d)$. To prove part (ii), define $\varphi_t(\zz) = t^{s+d} f_{\XX}(t\zz)$, and note that 
\[
|\varphi^\star(\zz_m) - \varphi^\star(\zz) | \leq |\varphi^\star(\zz_m) - \varphi_t(\zz_m) | + |\varphi_t(\zz) - \varphi_t(\zz_m) |+ |\varphi^\star(\zz) - \varphi_t(\zz) |.
\]
Let $\zz_m\to \zz \neq \mv 0$. Then, $\{\zz_m\}_{m>m_0}$ lie in a compact set for all large enough $m_0$. Consequently, choosing a $t$ large, the first and the third terms above can be made smaller than $\varepsilon/3$. For the second term, note that since $\varphi_t$ are continuous by assumption, the second term is smaller than $\varepsilon/3$ whenever $m>m_1$. Choosing $m>\max\{m_0,m_1\}$ yields that $|\varphi^\star(\zz_m) - \varphi^\star(\zz) |<\varepsilon$. $\varphi^\star$ is therefore continuous, and bounded on $\Theta_{\mathcal E}$ (a compact set). The second conclusion of part (ii) now follows. \qed

\noindent \textbf{Proof of Lemma~\ref{lem:eisenberg_clearing}}:  Fix an arbitrary measurable selection of clearing fixed
points and denote it by \(\pp\) for the purpose of this proof.
Since $\bar p_j(\XX)=\bigl(\mv a_j^\intercal \XX-r_j\bigr)_+,$ we have, for any \(\zz_t\to\zz\),
\[
    t^{-1}\bar p_j(t\zz_t)
    =
    \bigl(\mv a_j^\intercal \zz_t-t^{-1}r_j\bigr)_+
    \to
    (\mv a_j^\intercal \zz)_+
    =
    \mv a_j^\intercal \zz
    =:\bar p_j^\star(\zz),
\]
where the last equality uses \(\mv a_j,\zz\in\mathbb R_+^d\). Define
\[
    \bar p^{(t)}(\zz_t):=t^{-1}\bar p(t\zz_t),
    p^{(t)}(\zz_t):=t^{-1}p(t\zz_t),
    \qquad
    x^{(t)}(\zz_t):=t^{-1}x(t\zz_t).
\]
Dividing the clearing fixed point by \(t\) gives
\[
p_j^{(t)}(\zz_t)
=
\min\Bigl\{
\bar p_j^{(t)}(\zz_t),\,
x_j^{(t)}(\zz_t)+\sum_{k=1}^m\lambda_{kj}p_k^{(t)}(\zz_t)
\Bigr\}
=
[\Phi_t(p^{(t)}(\zz_t);\zz_t)]_j .
\]
By hypothesis, $\bar p^{(t)}(\zz_t)\to \bar p^\star(\zz)$ and $x^{(t)}(\zz_t)\to x^\star(\zz)$.
Hence \(\Phi_t(\cdot;\zz_t)\to\Phi^\star(\cdot;\zz)\) uniformly on every bounded
subset of \(\mathbb R^m\), where
\[
[\Phi^\star(p;\zz)]_j
:=
\min\Bigl\{
\bar p_j^\star(\zz),\,
x_j^\star(\zz)+\sum_{k=1}^m\lambda_{kj}p_k
\Bigr\}.
\]
Moreover, $\mv 0\le p^{(t)}(\zz_t)\le \bar p^{(t)}(\zz_t)$,
so the sequence \(\{p^{(t)}(\zz_t)\}\) is bounded. Let
\(p^{(t_m)}(\zz_{t_m})\to q\) be any convergent subsequence. Then
\[
\begin{aligned}
\|q-\Phi^\star(q;\zz)\|
&\le
\|q-p^{(t_m)}(\zz_{t_m})\|  \\
&\quad
+\|p^{(t_m)}(\zz_{t_m})
      -\Phi^\star(p^{(t_m)}(\zz_{t_m});\zz)\| \\
&\quad
+\|\Phi^\star(p^{(t_m)}(\zz_{t_m});\zz)-\Phi^\star(q;\zz)\|.
\end{aligned}
\]
The first term vanishes by the definition of the subsequence. The third term
vanishes by continuity of \(\Phi^\star(\cdot;\zz)\). For the second term, using
\(p^{(t_m)}(\zz_{t_m})=\Phi_{t_m}(p^{(t_m)}(\zz_{t_m});\zz_{t_m})\), we obtain
\[
\begin{aligned}
&\|p^{(t_m)}(\zz_{t_m})
      -\Phi^\star(p^{(t_m)}(\zz_{t_m});\zz)\| \\
&\qquad =
\|\Phi_{t_m}(p^{(t_m)}(\zz_{t_m});\zz_{t_m})
      -\Phi^\star(p^{(t_m)}(\zz_{t_m});\zz)\|
\to 0,
\end{aligned}
\]
by the uniform convergence of \(\Phi_t\) on bounded sets. Therefore
\(q=\Phi^\star(q;\zz)\). By uniqueness of the limiting clearing vector,
\(q=p^\star(\zz)\). Since every convergent subsequence has the same limit, $t^{-1}p(t\zz_t)\to p^\star(\zz)$.
Finally,
\[
    t^{-1}L_j(t\zz_t)
    =
    t^{-1}\bar p_j(t\zz_t)-t^{-1}p_j(t\zz_t)
    \to
    \bar p_j^\star(\zz)-p_j^\star(\zz)
    =
    L_j^\star(\zz).
\]
This proves the convergence. If each \(L_j^\star\) is
nondegenerate, the clearing loss system is asymptotically homogeneous.
If in addition, the limiting joint-stress regions have positive
Lebesgue measure, the system is regular.\qed
\subsection{Proofs from Section~\ref{sec:characterisation}}
The following technical lemmas are required for  the proofs from this section
\begin{lemma}\label{lem:angle_converges}
    Suppose Assumption~\ref{assume:ht_data} holds. Then, as $t\to\infty$,
    \[
    f_{R,\Pphi}(r,\pphi) = r^{-(s+1)}\varphi^\star(\pphi)(1+\varepsilon(r,\pphi)), \quad \text{ where } \quad \sup_{r>t,\pphi\in \Theta_{\mathcal E}} |\varepsilon(r,\pphi)|\to 0. 
    \]
\end{lemma}
\begin{lemma}\label{lem:homogenous_epi_conv}
    Let $f_n:\mathcal E\to\R$ be a sequence of real functions converging continuously to a homogeneous function, $f^\star$. Then, the characteristic functions $\iota_{\Lev_1^+[f_n]} \to \iota_{\Lev_1^+[f^\star]}$, epigraphically.
\end{lemma}
\begin{lemma}\label{lem:0-not-contained}
    Let $(L_1,\ldots,L_K)\in \mathcal H_1$ be regular. Then, there exists a $\delta>0$ such that $B_{\delta}\cap  \mathcal S^\star(\mathcal J) =\varnothing$. There exist constants $(c, u_1)$, such that for all $u>u_1$, $\zz \in \Gamma_u(\mathcal J) $ implies that $\|\zz\|>c$. Consequently, with $u\to\infty$ and $t = u^\theta$ for $\theta\in (0,1)$, there exists $u_0$ such that for all $u>u_0$, $\mathcal S(u;\mathcal J) \subset \{r>t\}$. 
\end{lemma}
\begin{lemma}\label{lem:homogeneous_sets_separate}
Let \(\mu\) and \(\nu\) be Radon measures on 
\((0,\infty)\times\Theta_{\mathcal E}\), finite on sets bounded away from the
origin. If \(\mu\neq\nu\), then there exists a
strictly positive continuous function \(h:\Theta_{\mathcal E}\to(0,\infty)\) such
that
\[
    A_h:=\{(r,\pphi):r\ge h(\pphi)\}
\]
is a continuity set for both \(\mu\) and \(\nu\), and $\mu(A_h)\neq\nu(A_h).$
Equivalently, there exist a strictly positive continuous function
\(\psi:\Theta_{\mathcal E}\to\mathbb R_+\) and \(c>0\) such that
\[
    A_{\psi,c}:=\{(r,\pphi):r\psi(\pphi)^{1/s}\ge c\}
\]
separates \(\mu\) and \(\nu\).
\end{lemma}
\noindent\textbf{Proof of Lemma~\ref{lem:tail_asymp}: } Let \(R=\|\XX\|\), \(\Pphi=\XX/\|\XX\|\), and let $\Theta_{\mathcal E}:=\{\pphi\in\mathcal E:\|\pphi\|=1\}$ 
denote the angular domain. We use the polar representation $\zz=r\pphi$, where  $r>0,\ \pphi\in\Theta_{\mathcal E},$
with Lebesgue measure written as $d\zz=r^{d-1}\,dr\,d\pphi$.
By a change of variables, if \(f_{R,\Pphi}\) denotes the joint density of \((R,\Pphi)\), then $f_{R,\Pphi}(r,\pphi)=r^{d-1}f_{\XX}(r\pphi)$. We first identify the density of the scaled polar vector
$\left({R}/{t},\Pphi\right)$.
By the one-dimensional change of variables \(R=t\rho\),
\[
    f_{R/t,\Pphi}(\rho,\pphi)
    =
    t\,f_{R,\Pphi}(t\rho,\pphi)
    =
    t\,(t\rho)^{d-1}f_{\XX}(t\rho\pphi)
    =
    t^d \rho^{d-1}f_{\XX}(t\rho\pphi).
\]
For a Borel set \(B\subset\mathcal E\setminus\{\mathbf 0\}\),
write \(B\) in polar coordinates. Then,
\[
\begin{aligned}
    t^s\mathbb P(t^{-1}\XX\in B)
    &=
    t^s\mathbb P\left(\left(\frac Rt,\Pphi\right)\in B\right) =
    \int_{B}
        t^{s+d}\rho^{d-1}f_{\XX}(t\rho\pphi)\,d\rho\,d\pphi .
\end{aligned}
\]
Now let \(B\subset\mathcal E\setminus\{\mathbf 0\}\) be a Borel set
such that \(\mathbf 0\notin\operatorname{cl}(B)\). Since \(B\) is
bounded away from the origin, there exists \(a>0\) such that
\(\rho>a\) for every \((\rho,\pphi)\in B\). 
Then, we have $B \subset \{\rho>a\}$, and therefore $t \rho > t a$ uniformly over $B$. 
Since $\pphi$ lies in the compact set $\Theta_\mathcal E,$
Lemma~\ref{lem:angle_converges}, implies that uniformly over $B$
\[
    t^{s+d}f_{\XX}(t\rho\pphi)
    =  \rho^{-(s+d)}
    \varphi^\star(\pphi) (1+o(1)).
\]
It follows from dominated convergence, that
\[
    t^s\mathbb P(t^{-1}\XX\in B)
    \to
    \int_{B}
        \rho^{-(s+1)}\varphi^\star(\pphi)\,d\rho\,d\pphi .
\]
Define the Radon measure \(\nu^\star\) on \(\mathcal E\setminus\{\mathbf 0\}\) by
\begin{equation}\label{eqn:polar_limit}
     \nu^\star(B)
    :=
    \int_{B}
        \rho^{-(s+1)}\varphi^\star(\pphi)\,d\rho\,d\pphi .
\end{equation}
It follows that, for every Borel set
\(B\subset\mathcal E\setminus\{\mathbf 0\}\) satisfying
\(\mathbf 0\notin\operatorname{cl}(B)\), $t^s\mathbb P(t^{-1}\XX\in B)
    \longrightarrow
    \nu^\star(B)$. In particular, this convergence holds whenever, additionally,
\(\nu^\star(\partial B)=0\). Hence, by
Definition~\ref{def:convergence_concept}, 
\[
    t^s\mathbb P(t^{-1}\XX\in\cdot)
    \xrightarrow{\mv M_0}
    \nu^\star(\cdot)
    \qquad\text{on }\mathcal E\setminus\{\mathbf 0\}.
\]
Finally, since \eqref{eqn:polar_limit} holds for every $B$, 
\[
    d\nu^\star(r,\pphi)
    =
    r^{-(s+1)}\varphi^\star(\pphi)\,dr\,d\pphi.\qed
    \]
\noindent \textbf{Proof of Proposition~\ref{prop:asymp_tails}:}  \textbf{Part i) Forward direction:} To complete the proof we demonstrate that 
\begin{equation}\label{eqn:show_sim}
    u^{s}\mathbb P(\XX_u\in \mathcal S(u;\mathcal J)) \sim \nu^\star(\Gamma_{\mathcal J}^\star) \quad \text{ and } \quad  u^{s}\mathbb P(\XX\in \mathcal S(u;\mathcal J)) \sim \nu^\star(\Gamma_{\mathcal J}^\star).
\end{equation}
To this end, define $\ZZ_u:=u^{-1}\XX$, and $\widetilde \ZZ_u:=u^{-1}\widetilde \XX_u$.
Then $\mathbb P\!\left(\widetilde \XX_u\in\mathcal S(u;\mathcal J)\right)
    =
    \mathbb P\!\left(\widetilde \ZZ_u\in\Gamma_u(\mathcal J)\right)$,
where $\Gamma_u(\mathcal J)
    :=
    u^{-1}\mathcal S(u;\mathcal J)$.
For \(j\in\mathcal J\), define
\begin{equation}\label{eqn:Ljs}
     L_{j,u}(\zz):=u^{-1}L_j(u\zz),
    \quad
    L_{\mathcal J,u}(\zz):=\min_{j\in\mathcal J}L_{j,u}(\zz), \quad \text{ and }\quad 
    L_{\mathcal J}^\star(\zz):=\min_{j\in\mathcal J}L_j^\star(\zz).
\end{equation}
Then
\[
    \Gamma_u(\mathcal J)
    =
    \{\zz\in\mathcal E:L_{\mathcal J,u}(\zz)\ge 1\},
    \qquad
    \Gamma_{\mathcal J}^\star
    :=
    \{\zz\in\mathcal E:L_{\mathcal J}^\star(\zz)\ge 1\}.
\]
Note that since $(L_1,\ldots, L_K) \in \mathcal H_1$,  \(L_{j,u}\to L_j^\star\) continuously for each
\(j\in\mathcal J\). Since \(\mathcal J\) is finite, the minimum also converges
continuously:
\[
    L_{\mathcal J,u}\to L_{\mathcal J}^\star .
\]
\noindent \textbf{Upper Bound: }Let $N>0$. From \cite{RockafellarWets1998}, Proposition 7.7(i)
for every \(\delta>0\) and every sufficiently large \(u\),
\[
    \Gamma_u(\mathcal J) \cap B_N\subseteq [\Gamma_{\mathcal J}^\star]^\delta,
\]
Choose \(\delta>0\) small enough that
\([\Gamma_{\mathcal J}^\star]^\delta\) remains bounded away from the origin; this is possible using Lemma~\ref{lem:0-not-contained}. Then for all such $u$,
\[
    \mathbb P\!\left(\widetilde \ZZ_u\in\Gamma_u(\mathcal J)\right) 
    \leq
    \mathbb P\!\left(\widetilde \ZZ_u\in[\Gamma_{\mathcal J}^\star]^\delta\right)+ \mathbb P (\tilde \ZZ_u\in B_N^c).
\]
Since \(\{\widetilde \XX_u\}_{u>0}\in\mathcal T(P)\), $\mv M_0$ convergence gives
\[
    \limsup_{u\to\infty}
    u^s
    \mathbb P\!\left(\widetilde \ZZ_u\in\Gamma_u(\mathcal J)\right)
    \le
    \nu^\star\!\left([\Gamma_{\mathcal J}^\star]^\delta\right) + \nu^\star(B_N^c).
\]
Letting \(\delta\downarrow0\) and using continuity from above, $\nu^\star\!\left([\Gamma_{\mathcal J}^\star]^\delta\right)
    \downarrow
    \nu^\star(\Gamma_{\mathcal J}^\star)$. It remains to control the second term. Observe that from Lemma~\ref{lem:tail_asymp},  
\[
\begin{aligned}
    \nu^\star(B_N^c)
    &=
    \int_{\Theta_{\mathcal E}}\int_N^\infty
        r^{-(s+1)}\varphi^\star(\pphi)\,dr\,d\pphi =
    N^{-s}
    \int_{\Theta_{\mathcal E}}\int_1^\infty
        r^{-(s+1)}\varphi^\star(\pphi)\,dr\,d\pphi =    N^{-s}\nu^\star(B_1^c).
\end{aligned}
\]
Thus \(\nu^\star(B_N^c)\to0\) as \(N\to\infty\). Putting these two together.
\[
    \limsup_{u\to\infty}
    u^s
    \mathbb P\!\left(\widetilde \ZZ_u\in\Gamma_u(\mathcal J)\right)
    \le
    \nu^\star(\Gamma_{\mathcal J}^\star) \implies  \limsup_{u\to\infty}
    u^s
    \mathbb P\!\left(\widetilde \XX_u\in\mathcal S (u; \mathcal J)\right)
    \le
    \nu^\star(\Gamma_{\mathcal J}^\star)
\]
Since $\XX_u\in \mathcal T(\mathcal P)$, $\ZZ_u$ has the same $\mv M_0$ limit as $\widetilde \ZZ_u$. Then,
\[
 \limsup_{u\to\infty}
    u^s
    \mathbb P\!\left(\XX\in\mathcal S(u;\mathcal J)\right)
    \le
    \nu^\star(\Gamma_{\mathcal J}^\star).
\]

\noindent \textbf{Lower Bound: }We next prove the matching lower bound. For \(N\ge1\), define the open inner
approximation
\begin{equation}\label{eqn:K_N}
     K_N
    :=
    \left\{\zz\in\mathcal E:
        L_{\mathcal J}^\star(\zz)> 1+\frac1N, \quad \|\zz\|<N
    \right\}
\end{equation}
Since \(L_{\mathcal J,u}\to L_{\mathcal J}^\star\) continuously, the convergence is
uniform on \(B_N\). Hence, for all sufficiently large \(u\),
\[
    \sup_{\zz\in B_N}
    \left|
        L_{\mathcal J,u}(\zz)-L_{\mathcal J}^\star(\zz)
    \right|
    <\frac1N .
\]
Therefore, if \(\zz\in K_N\), then $L_{\mathcal J,u}(\zz)
    \ge
    L_{\mathcal J}^\star(\zz)-\frac1N
    \ge 1$,
and hence $K_N\subseteq \Gamma_u(\mathcal J)$
for all sufficiently large \(u\). Consequently,
$\mathbb P\!\left(\widetilde \ZZ_u\in\Gamma_u(\mathcal J)\right)
    \ge
    \mathbb P\!\left(\widetilde \ZZ_u\in K_N\right)$.
Multiplying by \(u^s\), taking \(\liminf\), and using tail preservation of
\(\widetilde\XX_u\), we obtain
\[
    \liminf_{u\to\infty}
    u^s\mathbb P\!\left(\widetilde \ZZ_u\in\Gamma_u(\mathcal J)\right)
    \ge
    \nu^\star(K_N).
\]
Now let \(N\to\infty\). The sets \(K_N \uparrow \{L_{\mathcal J}^\star >1\}\) continuity from below implies that
\[
    \nu^\star(K_N)
    \uparrow
    \nu^\star\!\left(\{\zz:L_{\mathcal J}^\star(\zz)>1\}\right).
\]
Since \(L_{\mathcal J}^\star\) is continuous and homogeneous of degree one, for every $\pphi$, there is at most one $r$ that solves $L^\star_\mathcal J(r\pphi)=1$. Since $r$ has a density  the
boundary $\{\zz:L_{\mathcal J}^\star(\zz)=1\}$
is consequently \(\nu^\star\)-null. Hence
\[
    \nu^\star\!\left(\{\zz:L_{\mathcal J}^\star(\zz)>1\}\right)
    =
    \nu^\star\!\left(\{\zz:L_{\mathcal J}^\star(\zz)\ge1\}\right)
    =
    \nu^\star(\Gamma_{\mathcal J}^\star).
\]
Therefore,
\[
    \liminf_{u\to\infty}
    u^s\mathbb P\!\left(\widetilde \ZZ_u\in\Gamma_u(\mathcal J)\right)
    \ge
    \nu^\star(\Gamma_{\mathcal J}^\star) \implies  \liminf_{u\to\infty}
    u^s\mathbb P\!\left(\widetilde \XX_u\in\mathcal S(u;\mathcal J)\right)
    \ge
    \nu^\star\!\left(\Gamma_{\mathcal J}^{\star}\right).
\]
The same argument, with \(\ZZ_u=u^{-1}\XX\) in place of \(\widetilde\ZZ_u\), gives $\liminf_{u\to\infty}
    u^s\mathbb P\!\left(\XX\in\mathcal S(u;\mathcal J)\right)
    \ge
    \nu^\star\!\left(\Gamma_{\mathcal J}^{\star}\right)$.
Combining the upper and lower bounds yields
\eqref{eqn:show_sim}. \qed 
\medskip

\noindent\textbf{Part ii) Converse:} Let
\[
    \mu_u(B):=u^s\mathbb P(u^{-1}\XX_u\in B),
    \qquad
    \bar\mu_u(B):=u^s\mathbb P(u^{-1}\XX\in B).
\]
By Lemma~\ref{lem:tail_asymp}, \(\bar\mu_u\to\nu^\star\) in $\mv M_0$. Since
\(\{\XX_u\}_{u>0}\notin\mathcal T(P)\), we have \(\mu_u\not\to\nu^\star\).

\noindent\textbf{Case 1: }We first handle the case in which \(\{\mu_u\}\) is not locally bounded on
 $\mathcal E_0$, that is, there exists $a>0$ such that $\limsup_u \mu_u(B_a^c)=\infty.$ Hence $\limsup_u u^s\mathbb P(\|\XX_u\|\ge au)=\infty$.
Set \(K=1\), \(\mathcal J=\{1\}\), and \(L_1(\zz):=\|\zz\|/a\). Then
\(L_1\in\mathcal H_1\), with \(L_1^\star=L_1\), and
\[
    \mathcal S(u;\mathcal J)
    =
    \{\zz:\|\zz\|\ge au\}.
\]
For the true law,
\[
    u^s\mathbb P(\XX\in\mathcal S(u;\mathcal J))
    =
    u^s\mathbb P(\|\XX\|\ge au)
    \to
    \nu^\star(\{\zz:\|\zz\|\ge a\})<\infty.
\]
For the approximation, the corresponding scaled probabilities are unbounded along
a subsequence. Therefore,
\[
    \mathbb P\!\left(\XX_u\in\mathcal S(u;\mathcal J)\right)
    \not\sim
    \mathbb P\!\left(\XX\in\mathcal S(u;\mathcal J)\right).
\]

\noindent\textbf{Case 2:} Next suppose that the masses are locally bounded away from the origin but are not tight at infinity. Then there exist \(\varepsilon>0\), 
\[ 
\limsup_{u\to\infty} \mu_u(\{\zz:\|\zz\|\ge M\}) \ge \varepsilon \qquad \text{for every }M>m_0. 
\] 
Since $ 
\nu^\star(\{\zz:\|\zz\|\ge M\})\to0 $ as $M\to\infty$, choose \(M\) such that  $\nu^\star(\{\zz:\|\zz\|\ge M\})<\varepsilon/4 .$
Along a subsequence \(u_n\to\infty\), 
$ 
\mu_{u_n}(\{\zz:\|\zz\|\ge M\})\ge \varepsilon/2$.
Taking \(L_1(\zz)=\|\zz\|/M\), the associated stress event is 
$\mathcal S(u;\{1\})=\{\zz:\|\zz\|\ge Mu\}$. 
Then 
\[
u^s\mathbb P\!\left(
        \XX\in\mathcal S(u;\{1\})
    \right)
    \to
    \nu^\star\!\left(\{\zz:\|\zz\|\geq M\}\right)
    <\frac{\varepsilon}{4},\quad \text{ and }\quad 
    \limsup_{u\to\infty}
    u^s\mathbb P\!\left(\XX_u\in\mathcal S(u;\{1\})
    \right)
    \geq\frac{\varepsilon}{2}.
\]
The above statement implies that $    \mathbb P\!\left(\XX_u\in\mathcal S(u;\mathcal J)\right)
    \not\sim
    \mathbb P\!\left(\XX\in\mathcal S(u;\mathcal J)\right).
$
\noindent\textbf{Case 3: }It remains to consider the case in which \(\{\mu_u\}\) is locally bounded
away from the origin and asymptotically tight at infinity. Since \(\{\XX_u\}_{u>0}\notin\mathcal T(P)\), the measures $\mu_u
    :=
    u^s\mathbb P\!\left(u^{-1}\XX_u\in\cdot\right)$
do not converge to \(\nu^\star\) in \(\mv M_0\). As the
\(\mv M_0\)-topology is metrizable, there exist \(\varepsilon>0\) and a
sequence \(u_n\to\infty\) such that
\[
    d_{\mv M_0}(\mu_{u_n},\nu^\star)\geq\varepsilon
    \qquad\text{for every }n>n_0,
\]
where \(d_{\mv M_0}\) is any metric generating the
\(\mv M_0\)-topology. Asymptotic radial tightness ensures that
\begin{equation}\label{eqn:hult_1}
     \lim_{M\to\infty}\sup_{n\geq1}
    \mu_{u_n}\!\left(\{\zz:\|\zz\|\geq M\}\right)
    =0.
\end{equation}
Note that local boundedness implies that for every \(m>0\), $\limsup_{n\to\infty}
    \mu_{u_n}\!\left(\{\zz:\|\zz\|\geq m\}\right)<\infty$, or upon taking $n_1$ sufficiently large $\sup_{n\geq n_1}
    \mu_{u_n}\!\left(\{\zz:\|\zz\|\geq m\}\right)<\infty$.
Taking the compact set $C_{m,M} = \{m\leq \|\zz\| \leq M\}$,  
\eqref{eqn:hult_1} shows that for every $m$ and $\varepsilon>0$, there exists an $M$ large enough such that $\sup_{n>n_0}\mu_{u_n}((B_m \cup C_{m,M})^c) \leq \varepsilon$. 
Thus from \citealp[Theorem~2.7]{hult2006regular}, \(\{\mu_{u_n}\}_{n\geq n_0}\) is relatively compact in \(\mv M_0\). Consequently, there exist a further subsequence, again indexed by \(n\),
and some \(\widetilde\nu\in\mv M_0\) such that $\mu_{u_n}\xrightarrow{\mv M_0}\widetilde\nu$.
By continuity of the metric,
\[
    d_{\mv M_0}(\widetilde\nu,\nu^\star)
    =
    \lim_{n\to\infty}
    d_{\mv M_0}(\mu_{u_n},\nu^\star)
    \geq\varepsilon,
\]
and hence \(\widetilde\nu\neq\nu^\star\). Therefore,
Lemma~\ref{lem:homogeneous_sets_separate} yields a strictly positive
continuous function
\(\psi:\Theta_{\mathcal E}\to\mathbb R_+\) and \(c>0\) such that
\(A_{\psi,c}\) is a continuity set for both \(\widetilde\nu\) and
\(\nu^\star\), and
\[
    \widetilde\nu(A_{\psi,c})
    \neq
    \nu^\star(A_{\psi,c}).
\]
Define 
    $L_\psi(r,\pphi)
    :=
    c^{-1}r\psi(\pphi)^{1/s}$.
Then \(L_\psi\) is continuous, strictly positive away from the origin, and
positively homogeneous of degree one; hence \(L_\psi\in\mathcal H_1\).
Moreover, $u^{-1}L_\psi(u\zz)=L_\psi(\zz),$
and therefore $u^{-1}\mathcal S(u;\{1\})
    = \{\zz:L_\psi(\zz)\geq1\}
    = A_{\psi,c}$.
Since $\widetilde\nu(A_{\psi,c})
    \neq
    \nu^\star(A_{\psi,c}),$ the preceding argument yields $\mathbb P\!\left(
        \XX_u\in\mathcal S(u;\{1\})
    \right)
    \not\sim
    \mathbb P\!\left(
        \XX\in\mathcal S(u;\{1\})
    \right)$ as $u\to\infty, $ concluding the proof \qed


\noindent \textbf{Proof of Corollary~\ref{cor:stress_law_preserve}:} Note  that $\XX_u\in \mathcal T(\mathcal P)$ implies $u^s \mathbb P(u^{-1}\XX_u\in \cdot) \xrightarrow{\mv M_0} \nu^\star(\cdot)$. Now fix a $B\subset \mathcal E\setminus\{\mv 0\}$. Then, if $\nu^\star(\partial B) = 0$,
\[
\mathbb P\left(u^{-1}\XX_u \in B \mid \XX_u \in\mathcal S(u;\mathcal J) \right) = \frac{\mathbb P(u^{-1}\XX_u \in B\cap\Gamma_u(\mathcal J))}{\mathbb P(u^{-1}\XX_u\in  \Gamma_u(\mathcal J))}  \to \frac{\nu^\star(B\cap \mathcal S^\star(\mathcal J))}{\nu^\star(\mathcal S^\star(\mathcal J))}
\]
To handle the numerator, intersect the inner and outer approximations
to $\Gamma_u(\mathcal J)$ used in the proof of
Proposition~\ref{prop:asymp_tails} with $B$. Since
$\nu^\star(\partial B)=0$, the same argument yields $u^s\mathbb P\!\left(
u^{-1}\XX_u\in B\cap\Gamma_u(\mathcal J)
\right)
\to
\nu^\star\!\left(B\cap\mathcal S^\star(\mathcal J)\right).$ \qed

\noindent\textbf{Proof of Proposition~\ref{prop:conditional_tail_factorization}:} By Lemma~\ref{lem:angle_converges}, for \(r>t\), $f_{R,\Pphi}(r,\pphi)
=
r^{-(s+1)}\varphi^\star(\pphi)\bigl(1+\varepsilon(r,\pphi)\bigr)$,
where $\eta_t:=\sup_{r>t,\;\pphi}|\varepsilon(r,\pphi)|\to 0 \text{ as } t\to\infty$.
Since $p_t:=\mathbb P(R>t)$,
we have, for \(r>t\),
\[
f_{R,\Pphi}(r,\pphi\mid R>t)
=
\frac{f_{R,\Pphi}(r,\pphi)}{p_t}
=
\frac{r^{-(s+1)}\varphi^\star(\pphi)\bigl(1+\varepsilon(r,\pphi)\bigr)}{p_t}.
\]
Integrating over \(r>t\) gives
\begin{equation}\label{eqn:varphi_is_positive}
    \varphi^{(t)}(\pphi)
=
\int_t^\infty f_{R,\Pphi}(r,\pphi\mid R>t)\,dr
=
\frac{\varphi^\star(\pphi)}{p_t}
\int_t^\infty r^{-(s+1)}\bigl(1+\varepsilon(r,\pphi)\bigr)\,dr.
\end{equation}
Then, \eqref{eqn:varphi_is_positive} implies that for \(r>t\),
\begin{align*}
f_{R,\Pphi}(r,\pphi\mid R>t)
&=
\varphi^{(t)}(\pphi)\,
\frac{r^{-(s+1)}\bigl(1+\varepsilon(r,\pphi)\bigr)}
{\int_t^\infty y^{-(s+1)}\bigl(1+\varepsilon(y,\pphi)\bigr)\,dy} \\
&=
s t^s\,\varphi^{(t)}(\pphi)\,r^{-(s+1)}
\frac{1+\varepsilon(r,\pphi)}
{s t^s\int_t^\infty y^{-(s+1)}\bigl(1+\varepsilon(y,\pphi)\bigr)\,dy}.
\end{align*}
Since
\[
1-\eta_t
\le
1+\varepsilon(y,\pphi)
\le
1+\eta_t
\qquad (y>t),
\]
we obtain, uniformly in \(\pphi\),
\[
1-\eta_t
\le
s t^s\int_t^\infty y^{-(s+1)}\bigl(1+\varepsilon(y,\pphi)\bigr)\,dy
\le
1+\eta_t,
\]
because $s t^s\int_t^\infty y^{-(s+1)}\,dy=1$. Define
\[
1+\varepsilon_t(r,\pphi)
:=
\frac{1+\varepsilon(r,\pphi)}
{s t^s\int_t^\infty y^{-(s+1)}\bigl(1+\varepsilon(y,\pphi)\bigr)\,dy}.
\]
Then
\[
f_{R,\Pphi}(r,\pphi\mid R>t)
=
s t^s\,\varphi^{(t)}(\pphi)\,r^{-(s+1)}
\bigl(1+\varepsilon_t(r,\pphi)\bigr),
\qquad r>t.
\]
Moreover, since both numerator and denominator lie in \([1-\eta_t,1+\eta_t]\),
\[
\sup_{r>t,\;\pphi}\bigl|\varepsilon_t(r,\pphi)\bigr|
\le
\frac{2\eta_t}{1-\eta_t}
\to 0 .
\]
Therefore,
\[
\sup_{r>t,\;\pphi}
\left|
\frac{f_{R,\Pphi}(r,\pphi\mid R>t)}
{s t^s r^{-(s+1)}\varphi^{(t)}(\pphi)}
-1
\right|
\to 0,
\]
which proves the claim.\qed

\noindent\textbf{Proof of Theorem~\ref{thm:heavy_vanishing_RE}: }
For \(r>t\), by \eqref{eqn:f_tgen}, $\fgen^{(t)}(r,\pphi)
    =
    p_t\,s t^s r^{-(s+1)}\varphi^{(t)}(\pphi)$.
Also, $f_{R,\Pphi}(r,\pphi)
    =
    p_t\,f_{R,\Pphi}(r,\pphi\mid R>t)$.
Hence, for \(r>t\),
\begin{align*}
\frac{f_{R,\Pphi}(r,\pphi)}{\fgen^{(t)}(r,\pphi)}-1
&=
\frac{p_t\,f_{R,\Pphi}(r,\pphi\mid R>t)}
{p_t\,s t^s r^{-(s+1)}\varphi^{(t)}(\pphi)}
-1  = 
\frac{f_{R,\Pphi}(r,\pphi\mid R>t)}
{s t^s r^{-(s+1)}\varphi^{(t)}(\pphi)}
-1 .
\end{align*}
Therefore, Proposition~\ref{prop:conditional_tail_factorization} gives
\begin{equation}\label{eqn:unif_density_bdd}
    \sup_{r>t,\;\pphi}
\left|
\frac{f_{R,\Pphi}(r,\pphi)}{\fgen^{(t)}(r,\pphi)}-1
\right|
\to 0 .
\end{equation}
This proves \eqref{eqn:heavy_vanishing_RE}. It remains to show that \(\XX_u^{\rm gen}\in\mathcal T(P)\). Let \(B\subset
\mathcal E\setminus\{\mathbf 0\}\)
\(\nu^\star\)-continuity set with $0\not\in \operatorname{cl}[B]$. Thus, there exists $\delta_0$ such that $B_{\delta_0}\cap B =\varnothing$.  Since $t(u)/u\to 0$ as $u\to\infty$, 
this gives the existence of $u_0$ such that $u B\subset \{R>t(u)\}$ for all $u>u_0$.  
Hence the uniform density bound \eqref{eqn:unif_density_bdd} applies throughout
\(uB\). Let
\[
    \alpha_u
    :=
    \sup_{r>t(u),\;\pphi}
    \left|
    \frac{f_{R,\Pphi}(r,\pphi)}
         {\fgen^{(t(u))}(r,\pphi)}
    -1
    \right|.
\]
Then \(\alpha_u\to0\), and uniformly on \(uB\),
\[
|\fgen^{(t(u))}(r,\pphi)
    -
    f_{R,\Pphi}(r,\pphi)| \leq  \frac{\alpha_u}{1-\alpha_u}\ f_{R,\pphi}(r,\pphi).
\]
for all $(r>t(u),\pphi)$.
Hence
\[
|u^s\mathbb P(u^{-1}{\XX}_u^{\rm gen}\in B)-
u^s\mathbb P(u^{-1}\XX\in B)| \leq  \frac{\alpha_u}{1-\alpha_u} u^{s} \mathbb P(u^{-1}\XX\in B) 
\]
 By
Lemma~\ref{lem:tail_asymp}, $u^s\mathbb P(u^{-1}\XX\in B)\to\nu^\star(B)$,
and therefore $u^s\mathbb P(u^{-1}{\XX}_u^{\rm gen}\in B)\to\nu^\star(B)$.
Thus \(\XX_u^{\rm gen}\in\mathcal T(P)\).\qed

\noindent\textbf{Proof of Proposition~\ref{prop:stress_summaries}:} If \(\zz_u\to\zz\), 
\(\mv L_u(\zz_u)\to\mv L^\star(\zz)\), and at every \(\zz\) where \(W\) is
continuous at \(\mv L^\star(\zz)\) it follows that
\(W_u(\zz_u)\to W^\star(\zz)\). Observe that \(\ZZ_u\Rightarrow\ZZ^\star\)
 and
\(\widetilde\ZZ_u\Rightarrow\ZZ^\star\) by Corollary~\ref{cor:stress_law_preserve}. Note that for every $\zz\in D_0$ and $\zz_u\to\zz$, $\mv W_u(\zz_u) \to \mv W^\star(\zz)$, and $\mathbb P(\ZZ^\star \in D_0) = 1$. The claim follows
from the extended continuous mapping theorem
\cite[Theorem~1.11.1]{van1996weak}. \qed
\subsection{Proofs from Section~\ref{sec:RST_pop}}
\noindent \textbf{Proof of Lemma~\ref{lem:containment}:}
Let \(\{\XX_u\}\in\mathcal D(P)\) with \(\XX_u\sim P_u\), and write
\[
    \Psi_u^P(\zz):=u^{s+d}f_P(u\zz),
    \qquad
    \Psi_u^{(P_u)}(\zz):=u^{s+d}f_{P_u}(u\zz).
\]
Recall from the argument following \eqref{eqn:Ljs} that
\(L_{\mathcal J,u}\to L_{\mathcal J}^\star\) continuously and
\(L_{\mathcal J}^\star\) is homogeneous.
The proof of Lemma~\ref{lem:homogenous_epi_conv}, applied to the closed level
sets
\[
    \Gamma_u(\mathcal J)
    =
    \{\zz: L_{\mathcal J,u}(\zz) \geq 1\},
    \qquad
    \mathcal S^\star(\mathcal J)
    =
    \{\zz: L_{\mathcal J}^\star(\zz)\geq 1\},
\]
guarantees, for every \(\zz\in\mathcal S^\star(\mathcal J)\), a recovery
sequence \(\zz_u\to\zz\) with \(\zz_u\in\Gamma_u(\mathcal J)\) for all
sufficiently large \(u\). By regularity choose
\(\zz_0\in\mathcal S^\star(\mathcal J)\) with \(\varphi^\star(\zz_0)>0\), and let
\(\zz_u\to\zz_0\) be such a recovery sequence. Since \(P\) satisfies Assumption~\ref{assume:ht_data} and
\(\{P_u\}\) satisfy condition (D1),
\[
    \Psi_u^P(\zz_u)\to \varphi^\star(\zz_0),
    \qquad
    \Psi_u^{(P_u)}(\zz_u)\to \varphi^\star(\zz_0).
\]
Hence there exist \(\alpha>0\) and \(u_0\) with
\(\inf_{u>u_0}\min\{\Psi_u^P(\zz_u),\Psi_u^{(P_u)}(\zz_u)\}>\alpha\). Since
\(\zz_u\in\Gamma_u(\mathcal J)\), for all such \(u\)
\begin{equation}\label{eqn:opt_bigger_alpha}
    \sup_{\zz\in\Gamma_u(\mathcal J)}\Psi_u^P(\zz)>\alpha,
    \qquad
    \sup_{\zz\in\Gamma_u(\mathcal J)}\Psi_u^{(P_u)}(\zz)>\alpha ;
\end{equation}
the optimal values of both finite-level problems are thus bounded away from
zero uniformly in large \(u\).

We next show that \eqref{eqn:opt_bigger_alpha} prevents maximizers from
escaping to infinity. Consider first the true distribution $P$. Write \(\zz=\rho\pphi\),
where \(\rho=\|\zz\|\) and \(\pphi\in\Theta_{\mathcal E}\). Then
\[
    \Psi_u^P(\zz)
    =
    u^{s+d}f_P(u\rho\pphi)
    =
    \rho^{-(s+d)}
    \left\{(u\rho)^{s+d}f_P(u\rho\pphi)\right\}.
\]
By Assumption~\ref{assume:ht_data}, the term in braces converges uniformly 
over \(\pphi\in\Theta_{\mathcal E}\) and is therefore bounded once \(u\rho\) is large. Since
\(u\rho\ge u\) whenever \(\rho\ge1\), there exists \(C<\infty\) such
that, for all sufficiently large \(u\),
\[
    \Psi_u^P(\zz)\le C\|\zz\|^{-(s+d)},
    \qquad \|\zz\|\ge1 .
\]
Choose \(M_1>1\) with \(CM_1^{-(s+d)}<\alpha\); then
\(\sup_{\zz\in B_{M_1}^c}\Psi_u^P(\zz)<\alpha\) for all sufficiently large
\(u\). For the approximating family, condition (D2) supplies \(M_2\) and
\(u_1\) such that \(\sup_{\zz\in B_{M_2}^c}\Psi_u^{(P_u)}(\zz)<\alpha\) for all
\(u\ge u_1\). Setting \(M=\max\{M_1,M_2\}\), both optimal values in
\eqref{eqn:opt_bigger_alpha} exceed \(\alpha\) while both objectives are
strictly below \(\alpha\) on \(B_M^c\); consequently no maximizer of either
problem lies in \(B_M^c\). Further, since the feasible set $\Gamma_u(\mathcal J)$ is contained in $\{\|\zz\| > c\}$, the solution sets are contained in
\(\overline{B_M\cap B_{c}^{c}}\), where $c$ is as defined in Lemma~\ref{lem:0-not-contained}, for all sufficiently large \(u\). Take $K_\mathcal J$ as the closure of this set and observe that in this case $\mathcal Y_u(\mathcal J;P) \cup \mathcal Y_u(\mathcal J;P_u) \subset \mathcal K_\mathcal J$ for all $u$ large enough.\qed

\noindent{\textbf{Proof of Theorem~\ref{thm:rst_convergence}:}}
Fix \(\{\XX_u\}\in\mathcal D(P)\) with \(\XX_u\sim P_u\).

\noindent\textbf{Step 1:}
Recall from \eqref{eqn:Ljs} that \(L_{\mathcal J,u}\) converges continuously to
\(L_{\mathcal J}^\star\). Since each \(L_j^\star\) is homogeneous of degree one,
so is \(L_{\mathcal J}^\star\), and the regularity of the loss system
\((L_1,\ldots,L_K)\) gives \(\mathcal S^\star(\mathcal J)\neq\varnothing\). Then,
by Lemma~\ref{lem:homogenous_epi_conv}, the characteristic functions converge epigraphically:
\[
    \iota_{\Lev_1^+[L_{\mathcal J,u}]}
    \to
    \iota_{\Lev_1^+[L_{\mathcal J}^\star]}
\]

\noindent\textbf{Step 2:}
Let \(K_{\mathcal J}\subset \mathcal E\setminus\{\mv 0\}\) denote the compact
localization set supplied by Lemma~\ref{lem:containment}, so that for all
sufficiently large \(u\) the true and approximate reverse-stress solution sets
are contained in \(K_{\mathcal J}\). It therefore suffices to establish
objective convergence on this compact set.

For the true law, Assumption~\ref{assume:ht_data} gives
\(\Psi_u^{P}(\zz):=u^{s+d}f_P(u\zz)\to\varphi^\star(\zz)\) continuously on
\(\mathcal E\setminus\{\mv 0\}\); for the approximating family, condition (D1)
gives \(\Psi_u^{(P_u)}(\zz)\to\varphi^\star(\zz)\) in the same sense. Since
\(\varphi^\star\) is continuous and \(K_{\mathcal J}\) is compact, continuous
convergence is uniform on \(K_{\mathcal J}\); hence
\begin{equation}\label{eqn:both_converge}
     \sup_{\zz\in K_{\mathcal J}}
    \bigl|\Psi_u^{P}(\zz)-\varphi^\star(\zz)\bigr|\to0,
    \qquad
    \sup_{\zz\in K_{\mathcal J}}
    \bigl|\Psi_u^{(P_u)}(\zz)-\varphi^\star(\zz)\bigr|\to0 .
\end{equation}

\noindent\textbf{Step 3:}
Define the extended-valued minimization objectives
\[
    F_u^P(\zz)
    :=
    -\Psi_u^P(\zz)
    +
    \iota_{\Gamma_u(\mathcal J)}(\zz),
    \qquad
    F_u^{(P_u)}(\zz)
    :=
    -\Psi_u^{(P_u)}(\zz)
    +
    \iota_{\Gamma_u(\mathcal J)}(\zz),
\]
with common limiting objective
\[
    F^\star(\zz)
    :=
    -\varphi^\star(\zz)
    +
    \iota_{\mathcal S^\star(\mathcal J)}(\zz).
\]
The minimizers of \(F_u^P\) and \(F_u^{(P_u)}\) are exactly
\(\mathcal Y_u(\mathcal J;P)\) and \(\mathcal Y_u(\mathcal J;P_u)\),
respectively, while those of \(F^\star\) are
\[
    \mathcal Y^\star(\mathcal J)
    =
    \arg\max\{\varphi^\star(\zz):\zz\in\mathcal S^\star(\mathcal J)\}.
\]
By Steps 1 and 2, both \(F_u^P\) and \(F_u^{(P_u)}\) converge epigraphically to
\(F^\star\); from \eqref{eqn:both_converge}, $\Psi_u^P$ and $\Psi_u^{(P_u)}$ both converge uniformly, and thus continuously (see \cite{RockafellarWets1998}, Theorem 7.14) on \(K_{\mathcal J}\) and the characteristic functions converge epigraphically. Their sum therefore converges epigraphically too (see \cite{RockafellarWets1998}, Theorem 7.46). By
Lemma~\ref{lem:containment} the corresponding solution sets are eventually
contained in the common compact set \(K_{\mathcal J}\subset\mathcal E\setminus
\{\mv 0\}\). Therefore, using \cite{bonnans2013perturbation}, Proposition~4.4 (see the discussion on Pg. 265 therein), 
\[
    \limsup_{u\to\infty}\mathcal Y_u(\mathcal J;P)
    \subseteq
    \mathcal Y^\star(\mathcal J),
    \qquad
    \limsup_{u\to\infty}\mathcal Y_u(\mathcal J;P_u)
    \subseteq
    \mathcal Y^\star(\mathcal J).
\]
This proves \eqref{eqn:rst_main_gen}.

\noindent\textbf{Step 4:}
Let
\[
    V_u^P
    :=
    \sup_{\zz\in\Gamma_u(\mathcal J)}\Psi_u^P(\zz),
    \qquad
    V_u^{(P_u)}
    :=
    \sup_{\zz\in\Gamma_u(\mathcal J)}\Psi_u^{(P_u)}(\zz),
    \qquad
    V^\star
    :=
    \sup_{\zz\in\mathcal S^\star(\mathcal J)}\varphi^\star(\zz).
\]
The argument of Step \textbf{3} also yields \(V_u^P\to V^\star\) and
\(V_u^{(P_u)}\to V^\star\). Since the system is regular and
\(\varphi^\star>0\) on the relevant tail region, \(V^\star>0\), so
\(V_u^{(P_u)}/V_u^P\to1\). Because
\(V_u^P=u^{s+d}\sup_{\zz\in\Gamma_u(\mathcal J)}f_P(u\zz)\), and likewise for the
approximating family, the scaling \(u^{s+d}\) cancels in the ratio, giving
\[
    \frac{
    \sup_{\zz\in \Gamma_u(\mathcal J)}f_{P_u}(u\zz)
    }{
    \sup_{\zz\in \Gamma_u(\mathcal J)}f_P(u\zz)
    }
    \to1,
\]
which proves \eqref{eqn:rst_main_value}. Finally, if
\(\mathcal Y^\star(\mathcal J)=\{\zz^\star\}\), the outer limits in
\eqref{eqn:rst_main_gen} are singletons, so any selections
\(\zz_u\in\mathcal Y_u(\mathcal J;P)\) and
\(\tilde\zz_u\in\mathcal Y_u(\mathcal J;P_u)\) satisfy \(\zz_u\to\zz^\star\) and
\(\tilde\zz_u\to\zz^\star\). \qed

\noindent\textbf{Proof of Lemma~\ref{lem:ssgen_in_D}:}
Write \(\Psi_u^{\rm gen}(\zz):=u^{s+d}f_{P_{\rm gen}^{(t(u))}}(u\zz)\) and
\(\Psi_u^{P}(\zz):=u^{s+d}f_P(u\zz)\). First note that the relative error guarantee of
Theorem~\ref{thm:heavy_vanishing_RE} transfers to the risk-factor densities.
Set
\[
    \beta_u
    :=
    \sup_{r>t(u),\,\pphi}
    \left|
    \frac{f_{R,\Pphi}(r,\pphi)}{\fgen^{(t(u))}(r,\pphi)}-1
    \right| ,
\]
so that \(\beta_u\to0\) by \eqref{eqn:heavy_vanishing_RE}. For all large \(u\)
we have \(\beta_u<1\), whence
\[
    \sup_{r>t(u),\,\pphi}
    \left|
    \frac{\fgen^{(t(u))}(r,\pphi)}{f_{R,\Pphi}(r,\pphi)}-1
    \right|
    \le
    \frac{\beta_u}{1-\beta_u}
    =:\tilde\beta_u
    \to0 .
\]
Since \(f_{R,\Pphi}(r,\pphi)=J(r,\pphi)f_P(r\pphi)\) and
\(\fgen^{(t)}(r,\pphi)=J(r,\pphi)f_{P_{\rm gen}^{(t)}}(r\pphi)\) with a common
polar Jacobian \(J\), the same bound holds for the risk-factor densities:
\begin{equation}\label{eqn:rel_err_cartesian}
    \sup_{\|\yy\|>t(u)}
    \left|
    \frac{f_{P_{\rm gen}^{(t(u))}}(\yy)}{f_P(\yy)}-1
    \right|
    \le
    \tilde\beta_u
    \to0 .
\end{equation}

\noindent\textbf{(D1):} Let \(\zz_u\to\zz\in\mathcal E\setminus\{\mv 0\}\) and put
\(c:=\|\zz\|/2>0\), so that \(\|\zz_u\|\ge c\) for all large \(u\). Since
\(t(u)=u^\theta\) with \(\theta\in(0,1)\), we have
\(\|u\zz_u\|\ge cu>t(u)\) for all sufficiently large \(u\), and therefore
\eqref{eqn:rel_err_cartesian} applies at \(\yy=u\zz_u\):
\[
    \Psi_u^{\rm gen}(\zz_u)
    =
    \left[
    \frac{f_{P_{\rm gen}^{(t(u))}}(u\zz_u)}{f_P(u\zz_u)}
    \right]
    \Psi_u^{P}(\zz_u)
    =
    \bigl(1+O(\tilde\beta_u)\bigr)\,\Psi_u^{P}(\zz_u).
\]
By Assumption~\ref{assume:ht_data}, \(\Psi_u^{P}(\zz_u)\to\varphi^\star(\zz)\),
and \(\tilde\beta_u\to0\); hence \(\Psi_u^{\rm gen}(\zz_u)\to
\varphi^\star(\zz)\), which is (D1).

\noindent\textbf{(D2):} As established in the proof of
Lemma~\ref{lem:containment}, Assumption~\ref{assume:ht_data} yields a constant
\(C<\infty\) with \(\Psi_u^{P}(\zz)\le C\|\zz\|^{-(s+d)}\) for all
\(\|\zz\|\ge1\) and all sufficiently large \(u\). If \(\|\zz\|\ge M\ge 1\), then
\(\|u\zz\|\ge u>t(u)\) for large \(u\), so \eqref{eqn:rel_err_cartesian} gives
\(\Psi_u^{\rm gen}(\zz)\le(1+\tilde\beta_u)\Psi_u^{P}(\zz)\). Consequently, for
all sufficiently large \(u\),
\[
    \sup_{\|\zz\|\ge M}\Psi_u^{\rm gen}(\zz)
    \le
    2C\,M^{-(s+d)} ,
\]
and letting \(M\to\infty\) establishes (D2). Thus
\(\XX_u^{\rm gen}\in\mathcal D(P)\).
\qed

\subsection{Proofs from Section~\ref{sec:data_driven_guarantees}: }
The  following lemmas are essential to proving the results of this section 
\begin{lemma}[\textbf{Conditioning preserves relative likelihood error}]
\label{lem:conditioning_relative_error}
Let \(a\) and \(b\) be two positive densities on a measurable set \(A\). Suppose
\[
    \sup_{x\in A}\left|\frac{a(x)}{b(x)}-1\right|\le \eta
\]
for some \(\eta\in(0,1)\). Let \(P_A^a\) and \(P_A^b\) denote the corresponding
conditional laws on \(A\). Then
\[
    \operatorname{TV}(P_A^a,P_A^b)
    \le
    \frac{\eta}{1-\eta}.
\]
\end{lemma}

\begin{lemma}[\textbf{Stability under normalization}]
\label{lem:normalization_stability}
Let \(q,\widehat q\) be nonnegative integrable functions, and
write $Z=\int q$,  $Z=\int \widehat q$ with $Z,\widehat Z>0$. 
Let \(Q\) and \(\widehat Q\) be the probability measures
with densities \(q/Z\) and \(\widehat q/\widehat Z\), respectively. Then $\operatorname{TV}(\widehat Q,Q)
    \le
    \frac{1}{Z}\int |\widehat q-q|$.
\end{lemma}
\begin{lemma}[\textbf{Radial quantile and threshold order}]
\label{lem:radial_quantile}
Let \(v_{1-t}^{(R)}(P)\) denote the \((1-t)\)-quantile of \(R=\|\XX\|\),
\(\XX\sim P\). Under Assumption~\ref{assume:ht_data}:
\begin{enumerate}
    \item[(i)] \(v_{1-t}^{(R)}(P)\sim (C_\star/s)^{1/s}\,t^{-1/s}\) as
    \(t\downarrow0\), where
    \(C_\star=\int_{\Theta_{\mathcal E}}\varphi^\star(\pphi)\,d\pphi\).
    \item[(ii)] If \(k_n=\lfloor n^{1-q}\rfloor\) for some \(q\in(0,1)\), then
    \(t_n:=R_{(k_n)}\) satisfies
    \(t_n/v_{1-k_n/n}^{(R)}(P)\to1\) in probability; in particular
    \(t_n\asymp_P n^{q/s}\).
\end{enumerate}
\end{lemma}

\noindent{\textbf{Proof of Theorem~\ref{thm:rate_TV}}:} Let \(\Pi_n^0\) denote the scaled stressed law induced by the population splice at
the empirical threshold \(t_n\). That is, if
\(\XX_n^0\sim P_{\rm gen}^{(t_n)}\), where \(P_{\rm gen}^{(t_n)}\) uses the true
\(s\), \(p_{t_n}\), and \(\Phi^{(t_n)}\), then
$\Pi_n^0
:=
\Law\!\left(
u_n^{-1}\XX_n^0
\,\middle|\,
\XX_n^0\in \mathcal S(u_n;\mathcal J)
\right)$. By the triangle inequality,
\[
    \textsf{TV}\!\left(\widehat\Pi_n,\Pi_n\right)
    \le
    \underbrace{\textsf{TV}\!\left(\Pi_n^0,\Pi_n\right)}_{\text{Bias}}
    +
    \underbrace{\textsf{TV}\!\left(\widehat\Pi_n,\Pi_n^0\right)}_{\text{Statistical error}}.
\]
We bound the two terms separately.

\smallskip
\noindent\textbf{i) Bias term.}
By Assumption~\ref{assume:2nd_order} and the polar identity
\(f_{R,\Pphi}(r,\pphi)=r^{d-1}f_{\XX}(r\pphi)\),
\[
    f_{R,\Pphi}(r,\pphi)
    =
    r^{-(s+1)}\varphi^\star(\pphi)
    \{1+\varepsilon_t(r,\pphi)\},
    \qquad r>t,
\]
where
\[
    \eta_t
    :=
    \sup_{r>t,\;\pphi\in\Theta_{\mathcal E}}
    |\varepsilon_t(r,\pphi)|
    =
    O(t^{-\gamma}).
\]
Here we use that \(\varphi^\star\) is bounded away from zero on
\(\Theta_{\mathcal E}\). By definition of the finite-threshold angular density,
\[
    \varphi^{(t)}(\pphi)
    =
    \frac{\varphi^\star(\pphi)}{p_t}
    \int_t^\infty y^{-(s+1)}
    \{1+\varepsilon_t(y,\pphi)\}\,dy.
\]
Therefore, for \(r>t\),
\[
\begin{aligned}
    \fgen^{(t)}(r,\pphi)
    &=
    p_t\,s t^s r^{-(s+1)}\varphi^{(t)}(\pphi)\\
    &=
    r^{-(s+1)}\varphi^\star(\pphi)
    \left[
    s t^s\int_t^\infty y^{-(s+1)}
    \{1+\varepsilon_t(y,\pphi)\}\,dy
    \right].
\end{aligned}
\]
Since $s t^s\int_t^\infty y^{-(s+1)}\,dy=1$ and \(\sup_{y>t,\pphi}|\varepsilon_t(y,\pphi)|=\eta_t\), the bracketed term is
\(1+O(\eta_t)\), uniformly in \(\pphi\). Hence, uniformly over \(r>t\) and
\(\pphi\in\Theta_{\mathcal E}\),
\[
    \fgen^{(t)}(r,\pphi)
    =
    r^{-(s+1)}\varphi^\star(\pphi)\{1+O(t^{-\gamma})\}.
\]
Combining this with the expansion for \(f_{R,\Pphi}\) gives
\begin{equation}\label{eqn:rate_density}
        \frac{f_{R,\Pphi}(r,\pphi)}
         {\fgen^{(t)}(r,\pphi)}
    =
    1+O(t^{-\gamma})
\end{equation}

uniformly for \(r>t\) and \(\pphi\in\Theta_{\mathcal E}\). Since the Cartesian and
radial--angular densities differ by the same polar Jacobian, the same relative
bound holds for the corresponding densities on \(\mathcal E\). Thus, at the random
threshold \(t_n\),
\[
    \sup_{\|\xx\|>t_n}
    \left|
    \frac{f_P(\xx)}
         {f_{P_{\rm gen}^{(t_n)}}(\xx)}
    -1
    \right|
    =
    O_P(t_n^{-\gamma}).
\]

Let $A_n:=\mathcal S(u_n;\mathcal J)$.
Since \(t_n/u_n\to0\) in probability, Lemma~\ref{lem:0-not-contained} gives
\(A_n\subseteq\{\xx:\|\xx\|>t_n\}\) with probability tending to one. Hence
\[
    \sup_{\xx\in A_n}
    \left|
    \frac{f_P(\xx)}
         {f_{P_{\rm gen}^{(t_n)}}(\xx)}
    -1
    \right|
    =
    O_P(t_n^{-\gamma}).
\]
By Lemma~\ref{lem:conditioning_relative_error}, applied with
\(a= f_P \), \(b=f_{P_{\rm gen}^{(t_n)}}\), and \(A=A_n\),
\[
    \textsf{TV} \!\left(
    \Law(\XX_n^0\mid \XX_n^0\in A_n), \; 
    \Law(\XX\mid \XX\in A_n)
    \right)
    =
    O_P(t_n^{-\gamma}).
\]
Applying the data-processing inequality to the map \(\xx\mapsto u_n^{-1}\xx\)
yields $\mathsf{TV}\!\left(
    \Pi_n^0,\Pi_n
    \right)
    =
    O_P(t_n^{-\gamma})$.
Since \(k_n=n^{1-q}\), Lemma~\ref{lem:radial_quantile} implies that the  \(k_n/n\) tail-quantile of $R$ 
is of order \(n^{q/s}\), and the intermediate order statistic satisfies the same
asymptotic order. Therefore $t_n^{-\gamma} =
    O_P(n^{-q\gamma/s}),$
and hence
\begin{equation}\label{eqn:bias}
    {\sf{TV}}\!\left(
    \Pi_n^0,\Pi_n
    \right)
    =
    O_P(n^{-q\gamma/s}).
\end{equation}
\smallskip
\noindent\textbf{ii) Statistical error.} Let $\varphi_n:=\varphi^{(t_n)},\
    \widehat\varphi_n:=\widehat\varphi^{(t_n)},\ 
    \widehat s_n:=\widehat s_{t_n}$
and write $\Gamma_n:=\Gamma_{u_n}(\mathcal J)=
    u_n^{-1}\mathcal S(u_n;\mathcal J).$ 
Let $E_n^{\rm tail}$ be the set of outcomes for which $\left\{
        \mathcal S(u_n;\mathcal J)\subseteq \{\xx:\|\xx\|>t_n\}
    \right\}$. Note that through the choice of threshold, $t_n = o_p( u_n)$. Then, 
by Lemma~\ref{lem:0-not-contained} and \(t_n/u_n\to0\),  $\mathbb P(E_n^{\rm tail})\to1$. On \(E_n^{\rm tail}\), after the scaling \(r=u_n y\), the population and data-driven SSGEN stressed laws
are obtained by normalizing 
\[
    q_n(y,\pphi)
    :=
    s\,y^{-(s+1)}\varphi_n(\pphi)
    \mathbf 1\{y\pphi\in\Gamma_n\},
    \qquad
    \widehat q_n(y,\pphi)
    :=
    \widehat s_n\,y^{-(\widehat s_n+1)}
    \widehat\varphi_n(\pphi)
    \mathbf 1\{y\pphi\in\Gamma_n\}.
\]
The factors \(p_{t_n}\), \(t_n^s\), \(u_n^{-s}\), and their estimated analogues are
constant in \((y,\pphi)\), and are absorbed into the normalizing constants. Let
\[
    Z_n:=\int q_n(y,\pphi)\,dy\,d\pphi,
    \qquad
    \widehat Z_n:=\int \widehat q_n(y,\pphi)\,dy\,d\pphi .
\]

We first show that \(Z_n\) is bounded away from zero. By the expansion used in the
proof of Proposition~\ref{prop:conditional_tail_factorization},
\[
    f_{R,\Pphi}(r,\pphi)
    =
    r^{-(s+1)}\varphi^\star(\pphi)\{1+\varepsilon(r,\pphi)\},
    \qquad
    \sup_{r>t,\pphi}|\varepsilon(r,\pphi)|\to0 .
\]
Writing \(\eta_t:=\sup_{r>t,\pphi}|\varepsilon(r,\pphi)|\), we have
\[
\begin{aligned}
    \varphi^{(t)}(\pphi)
    &=
    \frac{\int_t^\infty f_{R,\Pphi}(r,\pphi)\,dr}
         {\int_{\Theta_{\mathcal E}}\int_t^\infty f_{R,\Pphi}(r,\theta)\,dr\,d\theta}  = 
    \frac{\varphi^\star(\pphi)\{1+O(\eta_t)\}}
         {C_\star\{1+O(\eta_t)\}},
\end{aligned}
\]
uniformly over \(\pphi\in\Theta_{\mathcal E}\), where
\[
    C_\star:=\int_{\Theta_{\mathcal E}}\varphi^\star(\theta)\,d\theta .
\]
Hence there exists \(\alpha_0>0\) such that
\(\varphi^{(t)}(\pphi)\ge \alpha_0\) for all large \(t\) and all
\(\pphi\in\Theta_{\mathcal E}\).

By regularity and homogeneity, choose \(N<\infty\) so that the compact set \(K_N\)
defined in \eqref{eqn:K_N} has positive Lebesgue measure. Following the lower-bound
argument in Proposition~\ref{prop:asymp_tails}, for all sufficiently large \(n\),
\[
    \Gamma_n\cap B_N\supset K_N .
\]
Therefore,
\[
Z_n
\ge
s\alpha_0
\int
y^{-(s+1)}
\mathbf 1\!\left\{y\pphi\in K_N\right\}
dy\,d\pphi > 0. 
\]
The right-hand side is a strictly positive constant. Hence there exists
\(z_0>0\) such that $\mathbb P(Z_n\geq z_0)\to 1.$ 
Now, on the event $ E_{n}^{\rm tail} \cap \{\min\{Z_n,\hat Z_n\}>z_0/2\}$, Lemma~\ref{lem:normalization_stability} shows that 
\[
    \textsf{TV}\!\left(\widehat\Pi_n,\Pi_n^0\right)
    \le
    \frac{2}{z_0}\int
    \left|
    \widehat q_n(y,\pphi)-q_n(y,\pphi)
    \right|
    dy\,d\pphi .
\]
Define $D_n = \int\left|
        \widehat q_n(y,\pphi)-q_n(y,\pphi) \right| dy\,d\pphi.$
We first bound \(D_n\) before normalizing. Observe that
\begin{align}\label{eqn:I_split}
\int |\widehat q_n-q_n|
&\le
\int
s y^{-(s+1)}
\left|\widehat\varphi_n(\pphi)-\varphi_n(\pphi)\right|
\mathbf 1\{y\pphi\in\Gamma_n\}\,dy\,d\pphi  \nonumber\\
&\quad+
\int
\left|
\widehat s_n y^{-(\widehat s_n+1)}
-
s y^{-(s+1)}
\right|
\widehat\varphi_n(\pphi)
\mathbf 1\{y\pphi\in\Gamma_n\}\,dy\,d\pphi  \nonumber\\
&=: I_{n,\varphi}+I_{n,s}.
\end{align}

From Lemma~\ref{lem:0-not-contained}, there exist \(c>0\) and \(n_0\) such that
\(y\pphi\in\Gamma_n\) implies \(y\ge c\) for all \(n\ge n_0\). Therefore
\[
    \sup_{n\ge n_0,\;\pphi}
    \int_0^\infty
    s y^{-(s+1)}
    \mathbf 1\{y\pphi\in\Gamma_n\}\,dy
    \le
    \int_c^\infty s y^{-(s+1)}\,dy
    =
    C .
\]
It follows that
\[
    I_{n,\varphi}
    \le
    C\int
    \left|\widehat\varphi_n(\pphi)-\varphi_n(\pphi)\right|\,d\pphi
    =
    2C\,\textsf{TV}\!\left(\widehat\Phi^{(t_n)},\Phi^{(t_n)}\right),
\]
We next bound \(I_{n,s}\). Since \(\widehat s_n\to s\) in probability, the event
\[
    E_n:=\left\{s/2\le \widehat s_n\le 3s/2\right\}
\]
satisfies \(\mathbb P(E_n)\to1\). By the mean-value theorem, on \(E_n\),
\[
\left|
\widehat s_n y^{-(\widehat s_n+1)}
-
s y^{-(s+1)}
\right|
\le
|\widehat s_n-s|
\sup_{a\in[s/2,3s/2]}
y^{-(a+1)}(1+a|\log y|).
\]
On the region of integration \(y\ge c\), and hence
\[
\sup_{a\in[s/2,3s/2]}
y^{-(a+1)}(1+a|\log y|)
\le
C\left[
\mathbf 1\{c\le y\le1\}
+
(1+\log y)y^{-(s/2+1)}\mathbf 1\{y>1\}
\right].
\]
The right-hand side is integrable on \([c,\infty)\), since \(s/2>0\). Since
\(\widehat\varphi_n\) is a density, we obtain, on \(E_n\),
\[
    I_{n,s}
    \le
    C|\widehat s_n-s|.
\]
As \(\mathbb P(E_n)\to1\), this gives
\[
     I_{n,s}
    =
    O_P\!\left(|\widehat s_n-s|\right).
\]
Substituting, the preceding bounds into \eqref{eqn:I_split} yields $D_n=o_P(1)$.

Since $|\widehat Z_n-Z_n|\leq D_n$,
we have \(|\widehat Z_n-Z_n|=o_P(1)\). Together with
\(\mathbb P(Z_n\geq z_0)\to1\), this implies
\(\mathbb P(\widehat Z_n>z_0/2)\to1\). Now
combining the two bounds, with $\mathbb  P(E_n^{\rm tail}) \to 1$,
\[
    \textsf{TV}\!\left(\widehat\Pi_n,\Pi_n^0\right)
    \le
    C\left(
    |\widehat s_{t_n}-s|
    +
    \textsf{TV}\!\left(
        \widehat\Phi^{(t_n)},\Phi^{(t_n)}
    \right)
    \right)
\]
with probability tending to one. By Assumption~\ref{assume:data_driven},
\[
    \textsf{TV}\!\left(\widehat\Pi_n,\Pi_n^0\right)
    =
    O_P\!\left(k_n^{-(m\wedge\ell)}\right)
    =
    O_P\!\left(n^{-(1-q)(m\wedge\ell)}\right).
\]
The theorem follows from this bound and \eqref{eqn:bias}.
\qed

\noindent\textbf{Proof of Corollary~\ref{cor:stress_summary_dd}:} Note that by definition,
\begin{align*}
    \mathsf{TV}(\Pi_n\circ W_{n}^{-1} \ , \hat \Pi_n \circ W_{n}^{-1})  & =\sup_{B\in \mathscr B(\mathsf  Y)} |(\Pi_n\circ W_{n}^{-1})(B) - (\hat\Pi_n\circ W_{n}^{-1})(B)| \\
    &\leq \sup_{C \in \mathscr B(\mathcal E)} |\Pi_n(C) - \hat \Pi_n(C)|  = \mathsf{TV}(\Pi_n,\hat \Pi_n).
\end{align*}
Here, the first inequality follows since $W_n$ are measurable, and thus $W_n^{-1}(B) \in \mathscr B(\mathcal E)$ whenever $B\in \mathscr B(\mathsf{Y})$. Now, using the bound in Theorem~\ref{thm:rate_TV}, the conclusion follows\qed

\subsection{Proofs from Section~\ref{sec:dd_rst}}
The following lemma is required for the proofs
\begin{lemma}[\textbf{Data-driven RST localization}]
\label{lem:rst_dd_localization}
Suppose the conditions of Lemma~\ref{lem:rst_dd_local_objective} hold. Then there exists a compact  $\mathcal K_\mathcal J$ such that
\[
    \mathbb P\!\left(
        \widehat Y_n\subseteq \mathcal K_{\mathcal J}
    \right)\to1 .
\]
\end{lemma}
\noindent \textbf{\noindent Proof of Lemma~\ref{lem:rst_dd_local_objective}: }
\smallskip
\noindent\textbf{Step 1:} Since \(\mathcal K_{\mathcal J}\) is compact and does not contain the origin,
there exist constants \(0<c_{\mathcal J}<C_{\mathcal J}<\infty\) such that
\(c_{\mathcal J}\le\|\zz\|\le C_{\mathcal J}\) for
\(\zz\in\mathcal K_{\mathcal J}\). By Assumption~\ref{assume:ht_data},
\(\Psi_n(\zz)=u_n^{s+d}f_P(u_n\zz)\) converges uniformly on
\(\mathcal K_{\mathcal J}\) to \(\varphi^\star\), and is in particular
uniformly bounded and, for all large \(n\), positive there. Hence
\[
    \Delta_n
    =
    \sup_{\zz\in\Gamma_{u_n}(\mathcal J)\cap\mathcal K_{\mathcal J}}
    \left|
    \widehat\Psi_n(\zz)-\Psi_n(\zz)
    \right|
    \leq
    \sup_{\zz\in\mathcal K_{\mathcal J}}
    \Psi_n(\zz)
    \sup_{\zz\in\mathcal K_{\mathcal J}}
    \left|
    \frac{\widehat\Psi_n(\zz)}{\Psi_n(\zz)}-1
    \right|,
\]
and it suffices to control the ratio uniformly on
\(\mathcal K_{\mathcal J}\).

\smallskip
\noindent\textbf{Step 2:}
For \(\zz\neq\mathbf 0\), write \(\zz=\rho\pphi\) with \(\rho=\|\zz\|\) and
\(\pphi=\zz/\|\zz\|\in\Theta_{\mathcal E}\). The polar change of variables
gives \(f_P(u_n\zz)=(u_n\rho)^{-(d-1)}f_{R,\Pphi}(u_n\rho,\pphi)\), and
similarly for the data-driven SSGEN density, so the Jacobians cancel:
\[
    \frac{\widehat\Psi_n(\zz)}{\Psi_n(\zz)}
    =
    \frac{
    \widehat f_{R,\Pphi}^{(t_n)}(u_n\rho,\pphi)
    }{
    f_{R,\Pphi}(u_n\rho,\pphi)
    }.
\]
Since \(\rho\ge c_{\mathcal J}\) on \(\mathcal K_{\mathcal J}\) and
\(t_n/u_n = o_p(1)\), the region \(u_n\rho>t_n\) covers
\(\mathcal K_{\mathcal J}\) with probability tending to one. On it, the
radial--angular form of the SSGEN density and the expansion of
Assumption~\ref{assume:2nd_order} give
\[
    \frac{\widehat\Psi_n(\zz)}{\Psi_n(\zz)}
    =
    \frac{\widehat p_{t_n}}{p_{t_n}}
    \cdot
    \frac{\widehat\varphi^{(t_n)}(\pphi)}
         {\varphi^{(t_n)}(\pphi)}
    \cdot
    \frac{\widehat s_n}{s}
    \cdot
    \left(\frac{u_n\rho}{t_n}\right)^{-(\widehat s_n-s)}
    \cdot
    \{1+\eta_n(u_n\rho,\pphi)\}^{-1},
\]
where, by Assumption~\ref{assume:2nd_order} and \(t_n\to\infty\) in probability,
\[
    \sup_{r>t_n,\ \pphi\in\Theta_{\mathcal E}}
    |\eta_n(r,\pphi)|
    =
    O_P(t_n^{-\gamma}).
\]

\smallskip
\noindent\textbf{Step 3:} We bound the above terms

\noindent \textbf{i) Angular ratio:} Since
\(\varphi^{(t_n)}\to\varphi^\star/C_\star\) uniformly on
\(\Theta_{\mathcal E}\), as in the proof of Theorem~\ref{thm:rate_TV}, and
\(\varphi^\star\) is continuous and positive,
\(\inf_{\pphi}\varphi^{(t_n)}(\pphi)\) is bounded away from zero for all
large \(n\). With the uniform angular condition of the lemma,
\[
    \sup_{\zz\in\mathcal K_{\mathcal J}}
    \left|
    \frac{\widehat\varphi^{(t_n)}(\pphi)}
         {\varphi^{(t_n)}(\pphi)}
    -1
    \right|
    =
    O_P(k_n^{-\ell}).
\]

\noindent \textbf{ii) Tail index:} Assumption~\ref{assume:data_driven}(i) gives
\(|\widehat s_n/s-1|=O_P(k_n^{-m})\).

\noindent \textbf{iii) Extrapolation factor:} Since \(\rho\in[c_{\mathcal J},C_{\mathcal J}]\),
\(\log(u_n\rho/t_n)=\log(u_n/t_n)+O(1)\) uniformly on
\(\mathcal K_{\mathcal J}\), and
\(k_n^{-m}\log(u_n/t_n)=o_P(1)\) since \(k_n\sim n^{1-q}\) and
\(\log(u_n/t_n)=O_P(\log n)\) by Lemma~\ref{lem:radial_quantile}(ii). Hence
\[
    \sup_{\zz\in\mathcal K_{\mathcal J}}
    \left|
    \left(\frac{u_n\rho}{t_n}\right)^{-(\widehat s_n-s)}-1
    \right|
    =
    O_P\!\left(k_n^{-m}\log(u_n/t_n)\right).
\]

\noindent \textbf{iv) Tail probability:} With \(\widehat p_{t_n}=k_n/n\) and, in the coupling of
Lemma~\ref{lem:radial_quantile}, \(p_{t_n}=\bar F_R(t_n)=U_{k_n:n}\),
\[
    \frac{\widehat p_{t_n}}{p_{t_n}}-1
    =
    \frac{k_n/n}{U_{k_n:n}}-1
    =
    O_P(k_n^{-1/2}),
\]
since \(U_{k_n:n}\) has mean \(k_n/(n+1)\) and variance of order
\(k_n/n^2\), so that
\((n/k_n)U_{k_n:n}-1=O_P(k_n^{-1/2})\) by Chebyshev's inequality.

\noindent{\textbf{v) Residue term:}} Since \(\sup|\eta_n|=O_P(t_n^{-\gamma})=o_P(1)\),
\[
    \sup_{\zz\in\mathcal K_{\mathcal J}}
    \left|
    \{1+\eta_n(u_n\rho,\pphi)\}^{-1}-1
    \right|
    =
    O_P(t_n^{-\gamma}).
\]

\smallskip
\noindent\textbf{Step 4:}
Each factor is \(1+o_P(1)\), and
\(\prod_i(1+\epsilon_{i,n})-1=O(\sum_i|\epsilon_{i,n}|)\) whenever the
\(\epsilon_{i,n}\) are \(o_P(1)\). Therefore
\[
    \sup_{\zz\in\mathcal K_{\mathcal J}}
    \left|
    \frac{\widehat\Psi_n(\zz)}{\Psi_n(\zz)}-1
    \right|
    =
    O_P\!\left(
    k_n^{-(1/2\wedge \ell)}
    +
    k_n^{-m}\log(u_n/t_n)
    +
    t_n^{-\gamma}
    \right),
\]
and since \(\sup_{\zz\in\mathcal K_{\mathcal J}}\Psi_n(\zz)=O(1)\) by
Step~1, the same bound holds for \(\Delta_n\). As \(\ell\le1/2\), the first
term is \(k_n^{-\ell}\), which is the stated rate. \qed

\noindent \textbf{Proof of Theorem~\ref{thm:rst_solution_rate}: }
Let \(\widehat\zz_n\in\widehat Y_n\) and let \(\zz_n^\star\in Y_n\). Note that for all $n$ large enough, Lemma~\ref{lem:containment}, $Y_n\subseteq \mathcal K_\mathcal J$. Then
on the event \(\widehat Y_n\subseteq K_{\mathcal J}\), both
\(\widehat\zz_n\) and \(\zz_n^\star\) belong to
\(\Gamma_{u_n}(\mathcal J)\cap K_{\mathcal J}\). Hence, by definition of
\[
    \Delta_n
    :=
    \sup_{\zz\in\Gamma_{u_n}(\mathcal J)\cap K_{\mathcal J}}
    \left|
        \widehat\Psi_n(\zz)-\Psi_n(\zz)
    \right|,
\]
we have
\begin{equation}\label{eqn:unif_delta}
    \left|
        \widehat\Psi_n(\zz)-\Psi_n(\zz)
    \right|
    \le \Delta_n,
    \qquad
    \zz\in\Gamma_{u_n}(\mathcal J)\cap K_{\mathcal J}.
\end{equation}
Since \(\widehat\zz_n\) maximizes \(\widehat\Psi_n\) over
\(\Gamma_{u_n}(\mathcal J)\), we have $\widehat\Psi_n(\widehat\zz_n)
    \ge
    \widehat\Psi_n(\zz_n^\star)$.
Therefore,
\begin{align*}
    \Psi_n(\zz_n^\star)-\Psi_n(\widehat\zz_n)
    &=
    \left(\Psi_n(\zz_n^\star)-\widehat\Psi_n(\zz_n^\star)\right)
    +
    \left(\widehat\Psi_n(\zz_n^\star)-\widehat\Psi_n(\widehat\zz_n)\right) +
    \left(\widehat\Psi_n(\widehat\zz_n)-\Psi_n(\widehat\zz_n)\right)  \\
    &\le 2\Delta_n .
\end{align*}
Since \(\zz_n^\star\in Y_n\), we have
\[
    V_n
    :=
    \sup_{\zz\in\Gamma_{u_n}(\mathcal J)}
    \Psi_n(\zz)
    =
    \Psi_n(\zz_n^\star).
\]
Thus $V_n-\Psi_n(\widehat\zz_n)
    \le
    2\Delta_n$. Applying Assumption~\ref{assume:rst_margin} now yields
\[
    c\{d(\widehat\zz_n,Y_n)\wedge\delta_0\}^{\kappa}
    \le
    V_n-\Psi_n(\widehat\zz_n).
\]
Further, from \eqref{eqn:unif_delta}, the same $\Delta_n$ works for \emph{every} selection $\hat \zz_n$. Combining the two displays now gives $c\{d(\widehat\zz_n,Y_n)\wedge\delta_0\}^{\kappa}
    \le
    2\Delta_n$ for every selection $\hat \zz_n$.  Hence
\[
    c\sup_{\widehat\zz\in\widehat Y_n} [d(\widehat\zz,Y_n)\land \delta_0]^{\kappa}  \leq 2\Delta_n
\]
Since \(\Delta_n=O_P(a_n)=o_P(1)\), the truncation by \(\delta_0\) is inactive
with probability tending to one.
Further, because  the event
\(\widehat Y_n\subseteq K_{\mathcal J}\) has probability tending
to one by Lemma~\ref{lem:rst_dd_localization}, we conclude that
\[
    \sup_{\widehat\zz\in\widehat Y_n}
    d(\widehat\zz,Y_n)
    =
    O_P(a_n^{1/\kappa}).
\]
If \(Y_n=\{\zz_n^\star\}\), then for every
\(\widehat\zz_n\in\widehat Y_n\), $\|\widehat\zz_n-\zz_n^\star\|
    =
    O_P(a_n^{1/\kappa})$.\qed

\section{Proofs of Technical Lemmas}

\noindent\textbf{Proof of Lemma~\ref{lem:angle_converges}: }By the polar change of variables,
\[
    f_{R,\Pphi}(r,\pphi)
    =
    r^{d-1} f_{\XX}(r\pphi).
\]
By Assumption~\ref{assume:ht_data}, for every compact set
\(K\subset\mathcal E\setminus\{\mathbf 0\}\), $u^{s+d} f_{\XX}(u\zz)\to \varphi^\star(\zz)$ 
uniformly over \(\zz\in K\). Applying this with
\(K=\Theta_{\mathcal E}\), we obtain
\[
    \sup_{\pphi\in\Theta_{\mathcal E}}
    \left|
    r^{s+d}f_{\XX}(r\pphi)-\varphi^\star(\pphi)
    \right|
    \to 0
    \qquad\text{as }r\to\infty .
\]
Equivalently, $f_{\XX}(r\pphi)
    =
    r^{-(s+d)}
    \{\varphi^\star(\pphi)+\delta(r,\pphi)\}$,
where $\sup_{r>t,\ \pphi\in\Theta_{\mathcal E}}
    |\delta(r,\pphi)|
    \to0$ as $t\to\infty$.
Therefore,
\[
    f_{R,\Pphi}(r,\pphi)
    =
    r^{d-1}f_{\XX}(r\pphi)
    =
    r^{-(s+1)}
    \{\varphi^\star(\pphi)+\delta(r,\pphi)\}.
\]
Since \(r^{s+d}f_{\XX}(r\pphi)\to\varphi^\star(\pphi)\) uniformly on
\(\Theta_{\mathcal E}\), and the approximating functions are continuous,
\(\varphi^\star\) is continuous on \(\Theta_{\mathcal E}\). Since
\(\Theta_{\mathcal E}\) is compact and \(\varphi^\star>0\) on
\(\Theta_{\mathcal E}\), it follows that $\inf_{\pphi\in\Theta_{\mathcal E}}\varphi^\star(\pphi)>0 $. Then,
\[
    \varepsilon(r,\pphi)
    :=
    \frac{\delta(r,\pphi)}{\varphi^\star(\pphi)} \implies  \sup_{r>t,\ \pphi\in\Theta_{\mathcal E}}
    |\varepsilon(r,\pphi)|
    \le
    \frac{
    \sup_{r>t,\ \pphi\in\Theta_{\mathcal E}}
    |\delta(r,\pphi)|
    }{
    \inf_{\pphi\in\Theta_{\mathcal E}}\varphi^\star(\pphi)
    }
    \to0 .
\]
Then $f_{R,\Pphi}(r,\pphi)
    =
    r^{-(s+1)}
    \varphi^\star(\pphi)
    \{1+\varepsilon(r,\pphi)\}$, proving the claim.\qed

\noindent \textbf{Proof of Lemma~\ref{lem:homogenous_epi_conv}: }
To prove epi-convergence, we must show that at each $\xx\in\mathcal E$ (see \cite{RockafellarWets1998}, Proposition 7.2)
\begin{subequations}
\begin{equation}\label{eqn:epi_lb}
              \liminf_{n\to\infty} \iota_{\Lev_1^+[f_n]}(\xx_n) \geq \iota_{\Lev_1^+[f]}(\xx) 
              \qquad \text{for all }  \xx_n\to \xx,
        \end{equation}
\begin{equation}\label{eqn:epi_ub}
         \limsup_{n\to\infty} \iota_{\Lev_1^+[f_n]}(\xx_n) \leq \iota_{\Lev_1^+[f]}(\xx) 
         \qquad \text{for some }  \xx_n\to \xx.
        \end{equation}
    \end{subequations}
To prove \eqref{eqn:epi_lb}, consider a sequence $\xx_n\to\xx$. First, if
$f(\xx) \geq 1$, then \(\iota_{\Lev_1^+[f]}(\xx)=0\), and the result follows since
\(\iota_{\Lev_1^+[f_n]}\geq 0\). If instead $f(\xx)<1$, then continuous convergence
of $f_n\to f$ implies that  $f_n(\xx_n) <1$ for all large enough $n$. Hence
\(\xx_n\notin \Lev_1^+[f_n]\) eventually, and therefore the $\liminf$ on the left
equals \(+\infty=\iota_{\Lev_1^+[f]}(\xx)\). This proves \eqref{eqn:epi_lb}. For \eqref{eqn:epi_ub}, if \(f(\xx)<1\) then
\(\iota_{\Lev_1^+[f]}(\xx)=+\infty\) and any \(\xx_n\to\xx\) is a recovery sequence.
It therefore suffices to establish the recovery property: for every
\(\zz\in\Lev_1^+[f]\) there exists \(\zz_n\to\zz\) with
\(\zz_n\in\Lev_1^+[f_n]\) for all sufficiently large \(n\).

Fix \(\zz\in\Lev_1^+[f]\), so \(f(\zz)\ge1\); by homogeneity \(f(\mathbf 0)=0\),
so \(\zz\neq\mathbf 0\). Let \(C_{\zz}:=\{c\zz:1\le c\le 3/2\}\). Continuous
convergence gives uniform convergence on the compact set \(C_{\zz}\)
(\cite{RockafellarWets1998}, Theorem~7.14), so
\(\varepsilon_n:=\sup_{\yy\in C_{\zz}}|f_n(\yy)-f(\yy)|\to0\). For \(n\) large
enough that \(\varepsilon_n\le1/4\), set \(c_n:=1+2\varepsilon_n\in[1,3/2]\) and
\(\zz_n:=c_n\zz\in C_{\zz}\). Then \(\zz_n\to\zz\), and by homogeneity
\(f(\zz_n)=c_nf(\zz)\ge c_n=1+2\varepsilon_n\), whence
\(f_n(\zz_n)\ge f(\zz_n)-\varepsilon_n\ge1+\varepsilon_n\ge1\). Thus
\(\zz_n\in\Lev_1^+[f_n]\) for all sufficiently large \(n\). \qed

\noindent \textbf{Proof of Lemma~\ref{lem:0-not-contained}:} Observe that
\[
    \mathcal S^\star(\mathcal J)
    =
    \{\zz\in\mathcal E:L_{\mathcal J}^\star(\zz)\geq 1\}.
\]
Since \(L_{\mathcal J}^\star\) is continuous and \(\Theta_{\mathcal E}\) is
compact, $M_{\mathcal J}
    :=
    \max_{\pphi\in\Theta_{\mathcal E}}
    L_{\mathcal J}^\star(\pphi)
    <\infty$. 
Moreover, regularity implies that \( S^\star(\mathcal J)\neq\varnothing\),
and hence \(M_{\mathcal J}>0\). For any \(\zz\neq \mathbf 0\), write
\(\zz=r\pphi\), where \(r=\|\zz\|\) and
\(\pphi\in\Theta_{\mathcal E}\). By homogeneity, $L_{\mathcal J}^\star(\zz)
    =
    r L_{\mathcal J}^\star(\pphi)
    \le
    r M_{\mathcal J}$. Therefore, whenever \(\|\zz\|<M_{\mathcal J}^{-1}\), $ L_{\mathcal J}^\star(\zz)<1. $ It follows that, for every \(\delta<M_{\mathcal J}^{-1}\),
\[
    B_\delta\cap \mathcal S^\star(\mathcal J)=\varnothing.
\]
Thus the limiting stress region is bounded away from the origin.

From the first part, there exists \(\delta>0\) with
\(B_\delta\cap\mathcal S^\star(\mathcal J)=\varnothing\). Fix
\(\kappa<\delta/2\) and \(M>\delta\). Since \(\Gamma_u(\mathcal J)\) is the
level set of \(L_{\mathcal J,u}\), the continuous convergence in
\eqref{eqn:Ljs} gives
\(\Gamma_u(\mathcal J)\cap B_M\subset[\mathcal S^\star(\mathcal J)]^\kappa\)
for all sufficiently large \(u\) (apply \cite{RockafellarWets1998},
Theorem~7.11 to translate continuous convergence to epi-convergence, followed
by Proposition~7.7). By the choice of \(\kappa\),
\(\zz\in[\mathcal S^\star(\mathcal J)]^\kappa\implies\|\zz\|>\delta/2\), while
points of \(\Gamma_u(\mathcal J)\cap B_M^c\) satisfy \(\|\zz\|\ge M>\delta\)
trivially. Hence \(\|\zz\|>\delta/2\) for all \(\zz\in\Gamma_u(\mathcal J)\)
and all sufficiently large \(u\).

For the last part, note that the previous implication suggests that $\yy\in\mathcal S(u;\mathcal J)\implies \|\yy\| > u\delta/2 $. Since $t(u) =u^\theta$, there exists a $u_0$ such that for all $u>u_0$, $t(u) <u\delta/2$, so that $\|\yy\| >t(u)$ establishing the claim. \qed

\noindent{\textbf{Proof of Lemma~\ref{lem:homogeneous_sets_separate}: }}
We complete the proof in three steps

\noindent \textbf{Step i) }We first prove that the class \(\{A_h:h\in C_+(\Theta_{\mathcal E})\}\), where
\(C_+(\Theta_{\mathcal E})\) denotes the strictly positive continuous functions on
\(\Theta_{\mathcal E}\), separates Radon measures on
\(\mathcal E\setminus\{\mathbf 0\}\). First observe that the class is closed under finite intersections, since
\[
    A_{h_1}\cap A_{h_2}
    =
    A_{h_1\vee h_2},
\]
and \(h_1\vee h_2\in C_+(\Theta_{\mathcal E})\). 
It is therefore a $\pi$-system. To demonstrate the separating class property, we must therefore demonstrate that it generates the
Borel \(\sigma\)-field on each set bounded away from the origin (see \cite{Billingsley}, Pg. 9).

 Fix \(a>0\) and set $E_a:=\{\zz:\|\zz\|\ge a\}$.
Let $\mathcal C_a
    :=
    \{A_h\cap E_a:h\in C_+(\Theta_{\mathcal E})\}$.
Then \(\mathcal C_a\) is a \(\pi\)-system containing \(E_a\). We next show that $\sigma(\mathcal C_a)=\mathcal B(E_a)$.
The inclusion \(\sigma(\mathcal C_a)\subseteq\mathcal B(E_a)\) is immediate,
because for each \(h\in C_+(\Theta_{\mathcal E})\), $A_h\cap E_a
    =
    \{(r,\pphi)\in E_a:r\ge h(\pphi)\}$
is closed in \(E_a\). For the reverse inclusion, it is enough to show that
\(\sigma(\mathcal C_a)\) contains a base for the topology of \(E_a\). Constant
choices of \(h\) imply that radial upper sets $\{r\ge b\}\cap E_a$
belong to \(\sigma(\mathcal C_a)\). Hence radial bands of the form
$\{\alpha<r<\beta\}\cap E_a$
also belong to \(\sigma(\mathcal C_a)\). Moreover, for any
\(h\in C_+(\Theta_{\mathcal E})\),
\[
    \{(r,\pphi)\in E_a:r>h(\pphi)\}
    =
    \bigcup_{m\ge1} \left(A_{(1+1/m)h}\cap E_a\right)
    \in \sigma(\mathcal C_a).
\]

Now fix a basic open rectangle \((\alpha,\beta)\times U\) in polar coordinates,
with \(\alpha>a\), and fix a point
\((r_0,\pphi_0)\in(\alpha,\beta)\times U\). Choose \(h\in C_+(\Theta_{\mathcal E})\) such that
\[
    h(\pphi_0)<r_0,
    \qquad
    h(\pphi)>\beta
    \quad \text{for all } \pphi\notin U.
\]
Then
\[
    G
    :=
    \{(r,\pphi)\in E_a:r>h(\pphi)\}
    \cap
    \{\alpha<r<\beta\}
\]
belongs to \(\sigma(\mathcal C_a)\). By construction,
\[
    (r_0,\pphi_0)\in G
    \qquad\text{and}\qquad
    G\subseteq(\alpha,\beta)\times U.
\]
Thus every point of every basic open rectangle \(R\subset E_a\) has an open \(\sigma(\mathcal C_a)\)-measurable neighborhood contained in \(R\). Hence the sets \(G\) constructed above form an open cover of \(R\). Since \(E_a\) is second countable, and this cover admits a countable sub-cover. Moreover, each such \(G\) is contained in \(R\). Therefore $ R=\bigcup_{m=1}^{\infty}G_m$ for some \(G_m\in\sigma(\mathcal C_a)\), and hence \(R\in\sigma(\mathcal C_a)\). Since \(E_a\) is second countable, every open subset of \(E_a\) is a countable union of basic open rectangles. As each such rectangle belongs to \(\sigma(\mathcal C_a)\), every open subset of \(E_a\) belongs to \(\sigma(\mathcal C_a)\). Hence \[ \mathcal B(E_a)\subseteq\sigma(\mathcal C_a). \]
\noindent \textbf{Step ii)} Suppose that \(\mu(A_h)=\nu(A_h)\) for every \(h\in C_+(\Theta_{\mathcal E})\)
such that \(A_h\) is a continuity set for both \(\mu\) and \(\nu\). We first show
that the continuity-set restriction can be removed. Fix \(h\in C_+(\Theta_{\mathcal E})\). For \(c\in(1/2,1)\), the boundaries
\[
    \partial A_{ch}=\{(r,\pphi):r=ch(\pphi)\}
\]
are pairwise disjoint. They are also contained in a set bounded away from the
origin and hence have finite total \((\mu+\nu)\)-mass. Therefore at most countably
many \(c\in(1/2,1)\) satisfy
\[
    (\mu+\nu)(\partial A_{ch})>0.
\]
Choose \(c_m\uparrow1\), \(c_m<1\), such that each \(A_{c_mh}\) is a continuity set
for both \(\mu\) and \(\nu\). Since \(A_{c_mh}\downarrow A_h\) and
\((\mu+\nu)(A_{c_1h})<\infty\), continuity from above yields
\[
    \mu(A_h)
    =
    \lim_{m\to\infty}\mu(A_{c_mh})
    =
    \lim_{m\to\infty}\nu(A_{c_mh})
    =
    \nu(A_h).
\]
Thus \(\mu(A_h)=\nu(A_h)\) for every \(h\in C_+(\Theta_{\mathcal E})\). Now fix \(a>0\). Since
\[
    A_h\cap E_a=A_{h\vee a},
\]
and \(h\vee a\in C_+(\Theta_{\mathcal E})\), the restrictions of \(\mu\) and
\(\nu\) to \(E_a\) agree on \(\mathcal C_a\). From Step \textbf{i)}, \(\mathcal C_a\) is
measure-determining on \(E_a\), and therefore, the restrictions agree on all of
\(\mathcal B(E_a)\).

\noindent \textbf{Step iii)}
Since \(a>0\) is arbitrary, the preceding conclusion holds for every
\(E_a=\{\zz:\|\zz\|\ge a\}\). To conclude that \(\mu=\nu\) on
\(\mathcal E\setminus\{\mathbf 0\}\), let \(B\subseteq
\mathcal E\setminus\{\mathbf 0\}\) be Borel. Then, for every \(m\ge1\),
\(B\cap E_{1/m}\in\mathcal B(E_{1/m})\), and hence
\[
    \mu(B\cap E_{1/m})=\nu(B\cap E_{1/m}).
\]
Moreover, $B\cap E_{1/m}\uparrow B$ as $m\to\infty$,
because $\mathcal E\setminus\{\mathbf 0\}
    =
    \bigcup_{m\ge1}E_{1/m}$.
By continuity from below of measures,
\[
    \mu(B)
    =
    \lim_{m\to\infty}\mu(B\cap E_{1/m})
    =
    \lim_{m\to\infty}\nu(B\cap E_{1/m})
    =
    \nu(B).
\]
Thus \(\mu=\nu\) on \(\mathcal E\setminus\{\mathbf 0\}\), contradicting the
hypothesis. Therefore there exists
\(h\in C_+(\Theta_{\mathcal E})\) such that \(A_h\) is a continuity set for both
measures and \(\mu(A_h)\neq\nu(A_h)\). Finally set $\psi(\pphi):=h(\pphi)^{-s}$. Then
\[
    A_h
    =
    \{(r,\pphi):r\psi(\pphi)^{1/s}\ge1\}
    =
    A_\psi,
\]
which proves the claim.\qed

\noindent\textbf{Proof of Lemma~\ref{lem:conditioning_relative_error}:} 
Write $a(x)=b(x)\{1+\varepsilon(x)\}$, $x\in A$. Under the lemma assumptions \(\sup_{x\in A}|\varepsilon(x)|\le \eta\). Define $Z_a:=\int_A a(x)\,dx$ and
    $Z_b:=\int_A b(x)\,dx$.
Then
\[
    Z_a
    =
    \int_A b(x)\{1+\varepsilon(x)\}\,dx
    =
    Z_b\{1+\bar\varepsilon_A\} \quad \text{where }\quad 
    \bar\varepsilon_A
    :=
    \frac{1}{Z_b}\int_A \varepsilon(x)b(x)\,dx .
\]
Since \(|\varepsilon(x)|\le\eta\), we have \(|\bar\varepsilon_A|\le\eta\), and hence
\(1+\bar\varepsilon_A\ge 1-\eta>0\).

The conditional densities of \(P_A^a\) and \(P_A^b\) on \(A\) are \(a/Z_a\) and
\(b/Z_b\), respectively. Therefore,
\[
\begin{aligned}
    \textsf{TV}(P_A^a,P_A^b)
    &=
    \frac12\int_A
    \left|
    \frac{a(x)}{Z_a}
    -
    \frac{b(x)}{Z_b}
    \right|\,dx  =
    \frac12\int_A
    \frac{b(x)}{Z_b}
    \left|
    \frac{1+\varepsilon(x)}{1+\bar\varepsilon_A}
    -1
    \right|\,dx  \\
    &=
    \frac12\int_A
    \frac{b(x)}{Z_b}
    \frac{|\varepsilon(x)-\bar\varepsilon_A|}
         {|1+\bar\varepsilon_A|}
    \,dx \le
    \frac{1}{2(1-\eta)}
    \int_A
    \frac{b(x)}{Z_b}
    \left(|\varepsilon(x)|+|\bar\varepsilon_A|\right)
    \,dx  \\
    &\le
    \frac{1}{2(1-\eta)}(\eta+\eta)
    =
    \frac{\eta}{1-\eta}.
\end{aligned}
\]
This proves the claim.\qed

\noindent\textbf{Proof of Lemma~\ref{lem:normalization_stability}:}  Let \(\Delta_n=\int|\widehat q_n-q_n|\). Since $|\widehat Z_n-Z_n|
\le \Delta_n$, we have
\[
\begin{aligned}
\int
\left|
\frac{\widehat q_n}{\widehat Z_n}
-
\frac{q_n}{Z_n}
\right|
&\le
\frac{1}{Z_n}\int|\widehat q_n-q_n|
+
\left|
\frac{1}{\widehat Z_n}
-
\frac{1}{Z_n}
\right|
\int \widehat q_n  \\
&=
\frac{\Delta_n}{Z_n}
+
\frac{|\widehat Z_n-Z_n|}{Z_n}
\le
\frac{2\Delta_n}{Z_n}.
\end{aligned}
\]
Therefore,
\[
\operatorname{TV}\!\left(\hat Q, Q\right)
\le
\frac{1}{Z_n}\int|\widehat q_n-q_n|.
\]

\noindent{\textbf{Proof of Lemma~\ref{lem:radial_quantile}}:}
\textbf{(i)} For \(c>0\) and \(\zz\in\mathcal E\setminus\{\mv 0\}\), applying
\eqref{eqn:ht_data} along \(t\) and along \(ct\) gives
\(\varphi^\star(c\zz)=c^{-(s+d)}\varphi^\star(\zz)\), so \(\varphi^\star\) is
homogeneous of order \(-(s+d)\). Taking the compact set
\(\Theta_{\mathcal E}\) in \eqref{eqn:ht_data},
\[
    \delta_u
    :=
    \sup_{r\ge u,\;\pphi\in\Theta_{\mathcal E}}
    \left|r^{s+d}f_{\XX}(r\pphi)-\varphi^\star(\pphi)\right|
    \to0,
    \qquad u\to\infty .
\]
Passing to polar coordinates, for \(u\) large,
\[
    \Prob(R>u)
    =
    \int_{\Theta_{\mathcal E}}\int_u^\infty
    r^{d-1}f_{\XX}(r\pphi)\,dr\,d\pphi
    =
    \int_{\Theta_{\mathcal E}}\int_u^\infty
    r^{-(s+1)}
    \left\{\varphi^\star(\pphi)+O(\delta_u)\right\}
    dr\,d\pphi
    =
    \frac{C_\star}{s}\,u^{-s}\{1+o(1)\},
\]
where \(C_\star=\int_{\Theta_{\mathcal E}}\varphi^\star(\pphi)\,d\pphi\in(0,\infty)\)
and the \(O(\delta_u)\) term is uniform since \(\Theta_{\mathcal E}\) has finite
measure. Fix \(\varepsilon>0\) and set
\(u_\pm:=\{(1\pm\varepsilon)C_\star/s\}^{1/s}\,t^{-1/s}\). Then
\(\Prob(R>u_+)<t<\Prob(R>u_-)\) for all \(t\) small, so
\(u_-\le v_{1-t}^{(R)}(P)\le u_+\). Since \(\varepsilon\) is arbitrary,
\[
    v_{1-t}^{(R)}(P)
    \sim
    \left(\frac{C_\star}{s}\right)^{1/s} t^{-1/s},
    \qquad t\downarrow0 
\]
\noindent\textbf{(ii)} Recall \(R_{(k)}\) denotes the \(k\)th largest radius;
write \(U_{k:n}\) for the \(k\)th smallest among \(n\) i.i.d.\ uniform draws
on \((0,1)\). Without loss of generality, realize \(R_1,\ldots,R_n\) as
\(R_i=\bar F_R^{\leftarrow}(U_i)\), where \(\bar F_R(u)=\Prob(R>u)\). Since
\(\bar F_R^{\leftarrow}\) is non-increasing, the \(k_n\)th largest radius
corresponds to the \(k_n\)th smallest uniform:
\[
    R_{(k_n)}
    =
    \bar F_R^{\leftarrow}\!\left(U_{k_n:n}\right)
    \quad\text{almost surely}.
\]
Set \(\beta_n:=U_{k_n:n}\). Since \(k_n\to\infty\) and \(k_n/n\to0\),
\[
    \frac{n}{k_n}\,\beta_n\to1
    \quad\text{in probability}
\]
(see \cite{einmahl1988strong}, Theorem~3(III), with \(\nu=1/2\)). By
part~(i), \(\bar F_R^{\leftarrow}(t)=v_{1-t}^{(R)}(P)\) is regularly varying
at zero of index \(-1/s\), with
\(\bar F_R^{\leftarrow}(t)=(C_\star/s)^{1/s}t^{-1/s}\{1+o(1)\}\) as
\(t\downarrow0\). Therefore
\[
    t_n
    =
    \bar F_R^{\leftarrow}(\beta_n)
    =
    (C_\star/s)^{1/s}\left(\frac{n}{k_n}\right)^{1/s}\{1+o_P(1)\}
    =
    v_{1-k_n/n}^{(R)}(P)\,\{1+o_P(1)\},
\]
Since \(n/k_n\sim n^{q}\), it follows that \(t_n=\Theta_P(n^{q/s})\). \qed

\noindent{\textbf{Proof of Lemma~\ref{lem:rst_dd_localization}}:}
The proof follows the localization argument of Lemma~\ref{lem:containment},
with the population objective replaced by the data-driven objective. We first
derive a decay bound for \(\widehat\Psi_n\) whose constant does not depend on
the localization radius, then choose the radius.

\smallskip
\noindent\textbf{Step 1:}
Let \(E_n:=\{s/2\le\widehat s_{t_n}\le 3s/2\}\); by
Assumption~\ref{assume:data_driven}(i), \(\mathbb P(E_n)\to1\). By
Lemma~\ref{lem:radial_quantile}(ii) and the tail expansion in its part (i),
\[
    p_{t_n}
    =
    \frac{C_\star}{s}\,t_n^{-s}\{1+o_P(1)\}
    =
    \frac{k_n}{n}\{1+o_P(1)\},
\]
and since \(\widehat p_{t_n}=k_n/n\) by construction of \(t_n\),
\[
    \frac{\widehat p_{t_n}}{p_{t_n}}= 1+o_p(1)
    \qquad
    p_{t_n}t_n^s= C^\star/s +o_p(1).
\]
By Assumption~\ref{assume:data_driven}(i) and
Lemma~\ref{lem:radial_quantile}(ii), \(|\widehat s_{t_n}-s|=O_P(k_n^{-m})\)
and \(\log(u_n/t_n)=O_P(\log n)\), so \(k_n^{-m}\log n\to0\) gives
\[
    \left(\frac{u_n}{t_n}\right)^{-(\widehat s_{t_n}-s)}
    =
    \exp\!\left\{-(\widehat s_{t_n}-s)\log(u_n/t_n)\right\}
    =
    1+o_P(1).
\]
Finally, under Assumption~\ref{assume:ht_data},
\(\varphi^{(t_n)}\) converges uniformly to $C \varphi^\star$ (here $C>0$ is a normalising constant) and is therefore uniformly bounded for all large enough $n$. Uniform consistency of the estimated density $\hat\varphi^{(t_n)}$ now implies that \(P(\|\widehat\varphi^{(t_n)}\|_\infty < C_0) \to 1\) for some constant $C_0$.

\smallskip
\noindent\textbf{Step 2:}
Write \(\zz=\rho\pphi\) with \(\rho=\|\zz\|\ge1\). Since \(t_n/u_n\to0\) in
probability, the region \(\{u_n\rho>t_n,\ \rho\ge1\}\) has probability
tending to one, and on it the generated tail density gives
\[
    \widehat\Psi_n(\rho\pphi)
    =
    u_n^{s+d}f_{\widehat P_{\rm gen}^{(t_n)}}(u_n\rho\pphi)
    =
    \frac{\widehat p_{t_n}}{p_{t_n}}
    \left(p_{t_n}t_n^{s}\right)
    \widehat s_{t_n}
    \left(\frac{u_n}{t_n}\right)^{-(\widehat s_{t_n}-s)}
    \rho^{-(\widehat s_{t_n}+d)}
    \widehat\varphi^{(t_n)}(\pphi).
\]
On \(E_n\), \(\rho^{-(\widehat s_{t_n}+d)}\le\rho^{-(s+d)/2}\) for
\(\rho\ge1\), and \(\widehat s_{t_n}\le 3s/2\). Combining with Step~1, there
exists a constant \(C<\infty\), not depending on the localization radius
below, such that with probability tending to one,
\[
    \widehat\Psi_n(\zz)
    \le
    C\|\zz\|^{-(s+d)/2},
    \qquad
    \zz\in\Gamma_{u_n}(\mathcal J),\ \|\zz\|\ge1 .
\]

\smallskip
\noindent\textbf{Step 3:}
By the construction in Lemma~\ref{lem:containment}, applied with radius
\(M>1\) large enough that
\[
    CM^{-(s+d)/2}<\frac{\alpha_0}{4},
\]
there exist \(\alpha_0>0\), a compact set
\(\mathcal K_{\mathcal J}\subset\mathcal E\setminus\{\mathbf 0\}\), and
feasible points
\(\zz_n^0\in\Gamma_{u_n}(\mathcal J)\cap\mathcal K_{\mathcal J}\) such that,
for all sufficiently large \(n\),
\[
    \Gamma_{u_n}(\mathcal J)\cap\overline B_M
    \subseteq
    \mathcal K_{\mathcal J},
    \qquad
    \Psi_n(\zz_n^0)\ge\alpha_0,
    \qquad
    \sup_{\zz\in\Gamma_{u_n}(\mathcal J)\cap B_M^c}
    \Psi_n(\zz)
    \le
    \frac{\alpha_0}{4}.
\]
With this choice of \(M\), Step~2 gives, with probability tending to one,
\[
    \sup_{\zz\in\Gamma_{u_n}(\mathcal J)\cap B_M^c}
    \widehat\Psi_n(\zz)
    \le
    \frac{\alpha_0}{4}.
\]

\smallskip
\noindent\textbf{Step 4:}
Since 
\(\zz_n^0\in\Gamma_{u_n}(\mathcal J)\cap\mathcal K_{\mathcal J}\) and
\(a_n\to0\), the uniform convergence in Lemma~\ref{lem:rst_dd_local_objective} gives
\[
    \widehat\Psi_n(\zz_n^0)
    =
    \Psi_n(\zz_n^0)+o_P(1)
    \ge
    \alpha_0+o_P(1),
\]
so that, with probability tending to one,
\[
    \sup_{\zz\in\Gamma_{u_n}(\mathcal J)}
    \widehat\Psi_n(\zz)
    \ge
    \widehat\Psi_n(\zz_n^0)
    \ge
    \frac{\alpha_0}{2}.
\]

\smallskip
\noindent\textbf{Step 5:}
On the intersection of the events of Steps~3 and~4, every maximizer of
\(\widehat\Psi_n\) over \(\Gamma_{u_n}(\mathcal J)\) attains a value of at
least \(\alpha_0/2\), whereas the objective does not exceed \(\alpha_0/4\)
on \(\Gamma_{u_n}(\mathcal J)\cap B_M^c\). Hence
\(\widehat Y_n\subseteq\Gamma_{u_n}(\mathcal J)\cap\overline B_M
\subseteq\mathcal K_{\mathcal J}\) for all sufficiently large \(n\), with
probability tending to one. \qed

\section{A matched-marginals reinsurance example}\label{app:matched}

Consider Example~\ref{eg:reinsurance} with $d=K=3$ object-loss factors and
unit exposures $\mv a_i=\mv e_i$. Then $L_i(\XX)=(X_i-c_i)_+$ and
$L_i^\star(\zz)=z_i$ by Lemma~\ref{lem:verify_reinsurance}. We use the
$L^1$ radial coordinate on $\mathbb R^3_+$, writing
$\zz=r\pphi$, where $r=\sum_i z_i$ and
$\pphi\in\Delta:=\{\pphi>0:\sum_i\pphi_i=1\}$. This is only a
reparametrization of the tail measure.

Let $L_{\rm agg}:=\sum_i L_i$, and condition on the aggregate stress event
$L_{\rm agg}(\XX)\ge u$. Since $L_{\rm agg}^\star(\zz)=\sum_i z_i=r$, the
limiting stress region for $\mathcal J=\{\mathrm{agg}\}$ is
$\mathcal S^\star(\mathcal J)=\{\zz:\sum_i z_i\ge1\}=\{r\ge1\}$. Thus aggregate
stress is purely radial in these coordinates. If
$\nu^\star(\mathrm dr,\mathrm d\pphi)=s r^{-(s+1)}\,\mathrm dr\,H(\mathrm d\pphi)$,
then conditioning on $\{r\ge1\}$ leaves the angular law equal to $H$.

We fix $s=2$. For an angular law $H$, define
\[
    M(H):=\int_\Delta \pphi_1^2\,H(\mathrm d\pphi),
    \qquad
    \Pi(H):=\int_\Delta \min(\pphi_1,\pphi_2)^2\,H(\mathrm d\pphi).
\]
Then $\mathbb P(X_i>x)\sim M(H)x^{-2}$ and
$\mathbb P(X_i>x,X_j>x)\sim \Pi(H)x^{-2}$, so
$\lambda(H):=\Pi(H)/M(H)$ is the same-threshold pairwise tail-dependence
coefficient. These quantities are lower-dimensional summaries of $H$; they need
not determine the breadth probability
\[
    \Phi_3(H)
    :=
    \mathbb P_H\!\left(W_{1/4}^{\rm breadth}=3\right)
    =
    \mathbb P_{\Pphi\sim H}\!\left(\min_i\Pphi_i>\frac14\right).
\]

We now construct two smooth angular laws that agree on $M$, $\Pi$, and $\lambda$,
but disagree on $\Phi_3$. Let $H_0$ be the uniform probability law on $\Delta$.
Define
\[
    g_1(\pphi):=\sum_i\left(\pphi_i-\frac13\right)^2-\frac16,
    \qquad
    g_2(\pphi):=\sum_i\left(\pphi_i-\frac13\right)^3-\frac1{45}.
\]
The constants center the functions: $\int_\Delta g_1\,H_0=0$ and
$\int_\Delta g_2\,H_0=0$. Set $g_\delta:=-2g_1+7g_2$, and for $|\theta|\le1$
define
\[
    H_\theta(\mathrm d\pphi)
    :=
    \bigl(1+\theta g_\delta(\pphi)\bigr)H_0(\mathrm d\pphi).
\]
A direct calculation gives $-7/20\le g_\delta(\pphi)\le2/5$ on $\Delta$. Hence
each $H_\theta$ is a bounded, strictly positive angular density.

\begin{lemma}\label{lem:matched}
For every $|\theta|\le1$,
\[
    M(H_\theta)=\frac16,
    \qquad
    \Pi(H_\theta)=\frac1{24},
    \qquad
    \lambda(H_\theta)=\frac14.
\]
However,
\[
    \Phi_3(H_\theta)
    =
    \frac1{16}+\frac{51}{5120}\theta .
\]
\end{lemma}

\begin{proof}
For the uniform law on $\Delta$, the moments are
$\mathbb E_{H_0}[\Pphi_1^a\Pphi_2^b\Pphi_3^c]
=2a!b!c!/(a+b+c+2)!$. Direct integration gives
\[
    \int_\Delta \pphi_1^2 g_1\,H_0=\frac1{180},
    \qquad
    \int_\Delta \pphi_1^2 g_2\,H_0=\frac1{630},
\]
and
\[
    \int_\Delta \min(\pphi_1,\pphi_2)^2 g_1\,H_0=-\frac1{360},
    \qquad
    \int_\Delta \min(\pphi_1,\pphi_2)^2 g_2\,H_0=-\frac1{1260}.
\]
Therefore
\[
    \int_\Delta \pphi_1^2 g_\delta\,H_0=0,
    \qquad
    \int_\Delta \min(\pphi_1,\pphi_2)^2 g_\delta\,H_0=0.
\]
Since $H_\theta=(1+\theta g_\delta)H_0$, the perturbation changes neither $M$ nor
$\Pi$. Thus $M(H_\theta)=M(H_0)=1/6$ and
$\Pi(H_\theta)=\Pi(H_0)=1/24$, giving $\lambda(H_\theta)=1/4$.

Finally, $H_0(\min_i\pphi_i>1/4)=(1-3/4)^2=1/16$, and direct integration over the
same smaller simplex gives
\[
    \int_\Delta
    \mathbf 1\!\left\{\min_i\pphi_i>\frac14\right\}
    g_\delta(\pphi)\,H_0(\mathrm d\pphi)
    =
    \frac{51}{5120}.
\]
Hence $\Phi_3(H_\theta)=1/16+(51/5120)\theta$.
\end{proof}

Take $H_A:=H_{-1}$ and $H_B:=H_{+1}$. By Lemma~\ref{lem:matched}, the two systems
have identical marginal tail coefficients, identical same-threshold pairwise
co-stress coefficients, identical same-threshold tail-dependence coefficients,
and the same aggregate-stress tail. Nevertheless,
\[
    \mathbb P_A\!\left(W_{1/4}^{\rm breadth}=3\right)
    =
    \frac{269}{5120}
    \approx 0.053,
    \qquad
    \mathbb P_B\!\left(W_{1/4}^{\rm breadth}=3\right)
    =
    \frac{371}{5120}
    \approx 0.072.
\]
The relative difference is about $38\%$.

Thus marginal tail calibration and same-threshold pairwise tail-dependence
calibration do not identify the law of breadth under aggregate stress. The
missing information is the allocation of angular tail mass across all
components. A model calibrated only to these lower-dimensional summaries must
impose that remaining angular geometry by assumption. SSGEN instead targets it
through the learned angular law.

\begin{remark}\label{rem:realizable}\em
The construction uses bounded, strictly positive angular densities, not atomic
spectral measures. Paired with a Pareto$(s)$ radius, each $H_\theta$ yields a
risk-factor density satisfying Assumption~\ref{assume:ht_data}. Since
$\eta=1/4\in(0,1/3)$ and $H_\theta$ is absolutely continuous, no loss share has
an atom at $\eta$, so the continuity condition for breadth is satisfied.
\end{remark}
\section{Implementation Details and Diffusion Model diagnostics}\label{app:implement}
Algorithm~\ref{alg:rst_dd} implements the data-driven reverse-stress procedure. It
generates SSGEN tail scenarios, filters them by the target stress constraints,
fits a Kernel Density Estimator (KDE) to the scaled stress sample, and returns approximate density modes. The
algorithm is procedural: KDE bandwidth selection and numerical mode-finding are
standard implementation choices, not the focus of this paper. 

\begin{algorithm}[h]
\small
\caption{Data-driven reverse stress testing using \(\SSGEN\)}
\label{alg:rst_dd}

\KwIn{Trained SSGEN model \((\widehat\varphi^{(t)},\widehat s_t)\);
losses \(L_1,\ldots,L_K\); stress set \(\mathcal J\subseteq[K]\);
levels \(u,t\); retained sample size \(N_{\text{KDE}}\); KDE bandwidth \(h\);
tolerance \(\varepsilon_{\rm opt}\).}
\KwOut{Approximate reverse-stress set
\(\widehat{\mathcal Y}_{u}(\mathcal J)\).}

\noindent\textbf{I) Generate scaled stress samples};
\begin{enumerate}
    \item[(i)] Independently generate SSGEN tail samples
    $H_t$ using Algorithm~\ref{algo:generate_extremes}, and retain
    those satisfying
    \[
        L_j(H_t)\ge u,\qquad \forall j\in\mathcal J,
    \]
    until $N_{\mathrm{KDE}}$ samples
    $H_{t,1},\ldots,H_{t,N_{\mathrm{KDE}}}$ have been retained.

    \item[(ii)] Set
    \[
        \ZZ_{u,m}:=u^{-1}H_{t,m},
        \qquad m=1,\ldots,N_{\mathrm{KDE}},
    \]
    and define
    \[
        \widehat{\mathcal Z}_{u,\mathcal J}
        :=
        \{\ZZ_{u,m}:m=1,\ldots,N_{\mathrm{KDE}}\}.
    \]
\end{enumerate}

\noindent\textbf{II) Estimate density};
\begin{enumerate}
    \item[(i)] Fit a KDE $\widehat g_h$ to
    $\widehat{\mathcal Z}_{u,\mathcal J}$.
\end{enumerate}

\noindent \textbf{III) Return near-maximizers}\;
\begin{enumerate}
    \item[(i)] Define
    \[
        \Gamma_u(\mathcal J)
        :=
        \{\zz:L_j(u\zz)\ge u,\ \forall j\in\mathcal J\},
        \qquad
        \widehat V_u
        :=
        \sup_{\zz\in\Gamma_u(\mathcal J)}
        \widehat g_h(\zz).
    \]

    \item[(ii)] Return
    \[
        \widehat{\mathcal Y}_{u}(\mathcal J)
        :=
        \{\zz\in\Gamma_u(\mathcal J):
        \widehat g_h(\zz)\ge \widehat V_u-\varepsilon_{\rm opt}\}.
    \]
\end{enumerate}
\end{algorithm}

We give the implementation diagnostics for the angular diffusion model
used in the numerical experiments of Section~\ref{sec:numericals}. The theoretical
analysis treats the angular learner abstractly through the learned law
\(\widehat\varphi^{(t)}\). In the experiments, this law is implemented by
training a diffusion model on the intermediate exceedance directions
\[
    \Theta_{\mathcal Y}
    =
    \left\{
        \frac{\XX_i}{\|\XX_i\|}: R_i\ge t
    \right\}.
\]
Since the raw output of the diffusion model need not lie exactly on the angular
domain, each generated output is projected back to $\{\pphi\in \mathcal E:\|\pphi\|=1\}$ 
before being paired with an independent Pareto radial component. Writing
\(\Pi_{\Theta_\mathcal E}\) for this projection, the generated angular sample used by
SSGEN is
\[
    \widehat\Pphi=\Pi_{\Theta_\mathcal E}(\widetilde\Pphi),
\]
where \(\widetilde\Pphi\) denotes the raw diffusion-model output.

Figure~\ref{fig:diffusion_training_diagnostics} reports representative training
diagnostics. The left panel shows the training loss over iterations. The middle
panel compares the learned and empirical angular distributions through random
one-dimensional projections, using Kolmogorov--Smirnov and Wasserstein
distances averaged over projections. The right panel reports the cosine
similarity between the mean generated direction and the mean empirical
exceedance direction. The diagnostics stabilize after the initial training
phase; we therefore use \(5000\) diffusion iterations in the numerical
experiments.

\begin{figure}[t]
    \centering
    \includegraphics[width=\textwidth]{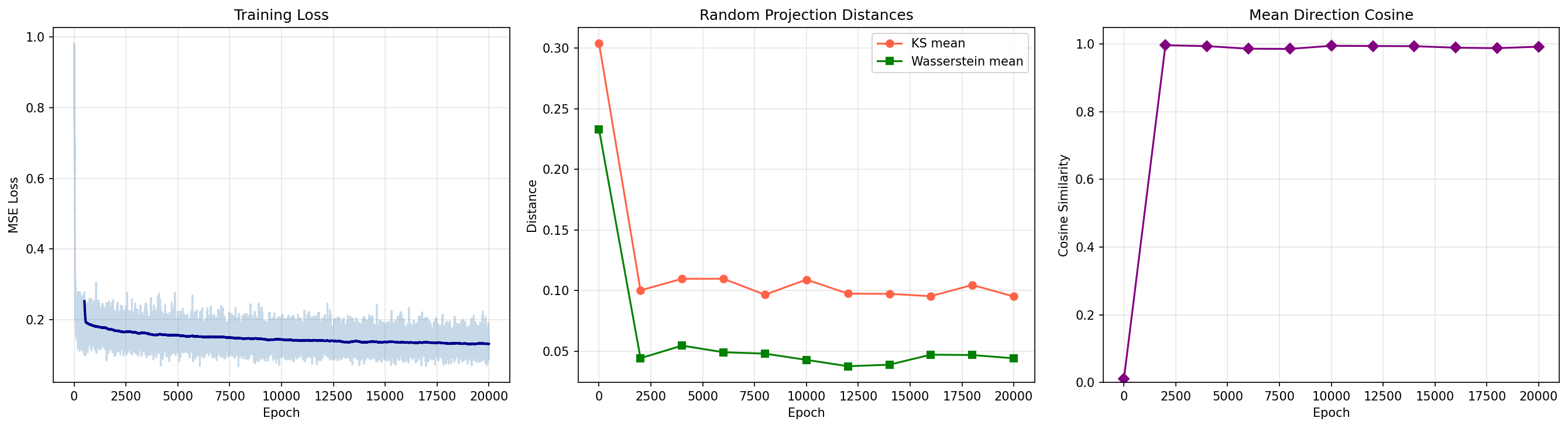}
    \caption{
   Left: training
    loss. Middle: random-projection distributional distances between generated
    and empirical exceedance directions. Right: cosine similarity between the
    mean generated direction and the mean empirical exceedance direction.
    Generated samples are projected back to \(\Theta_E\) before evaluation.
    }
    \label{fig:diffusion_training_diagnostics}
\end{figure}
\end{document}